\documentclass[11pt,reqno,a4paper,final]{amsart}
\usepackage[british]{babel}
\usepackage{amssymb,amstext,amsthm,eucal,bbm,amscd,bbold,mathabx}
\usepackage[margin=1in]{geometry}
\usepackage{array}
\usepackage[svgnames,table]{xcolor}
\usepackage[T1]{fontenc}
\usepackage[utf8]{inputenc}
\usepackage{mathrsfs}
\usepackage{verbatim}
\usepackage{pdfsync}
\usepackage[final]{microtype}
\usepackage[nosort]{cite}
\usepackage{adjustbox}
\usepackage{stackengine}
\usepackage{booktabs} %for \toprule,\midrule,\bottomrule
\usepackage{caption}
\usepackage{appendix}
\usepackage{chngcntr}
\usepackage{apptools}
\usepackage{romanbar}
\usepackage{enumitem}
\usepackage{braket}
\usepackage[scr=boondox]{mathalpha}  % lowercase mathscr
\newcommand\ext[1]{\scalebox{.8}[1]{$\bigwedge^{\!#1}$}}
\usepackage{tikz,pgfplots}
\pgfplotsset{compat=1.18}
\usepgfplotslibrary{fillbetween}
\usetikzlibrary{calc,arrows,arrows.meta,cd,shapes.geometric,positioning,through,decorations.markings,decorations.text}

\tikzset{>=latex}
\usepackage[unicode,pagebackref=true]{hyperref}
\hypersetup{%
  pdftitle   = {Quantum carrollian bosonic strings},
  pdfkeywords = {carrollian strings, BRST, flat space holography, twistors},
  pdfauthor  = {José Figueroa-O'Farrill, Emil Have, Niels Obers},
  pdfcreator = {\LaTeX\ with package \flqq hyperref\frqq},
  linkcolor=NavyBlue,
  citecolor=ForestGreen,
  urlcolor=OrangeRed,
  anchorcolor=OrangeRed,
  colorlinks=true}
\renewcommand*{\backref}[1]{}
\renewcommand*{\backrefalt}[4]{%
  \ifcase #1%
  \or [Page~#2.]%
  \else [Pages~#2.]%
  \fi%
}
\theoremstyle{plain}
\newtheorem{lemma}{Lemma}
\newtheorem{proposition}[lemma]{Proposition}

\newtheorem{corollary}[lemma]{Corollary}

\theoremstyle{definition}

\newcommand*{\dd}{\mathop{}\!\mathrm{d}}
\newcommand{\g}{\mathfrak{g}}

\renewcommand{\d}{\partial}

\newcommand{\ft}{\mathfrak{t}}

\newcommand{\fM}{\mathfrak{M}}
\newcommand{\fe}{\mathfrak{e}}
\newcommand{\Tgh}{T^{\text{gh}}}
\newcommand{\Mgh}{M^{\text{gh}}}
\newcommand{\Gpgh}{G^{+\,\text{gh}}}
\newcommand{\Gmgh}{G^{-\,\text{gh}}}
\newcommand{\Ttot}{T^{\text{tot}}}
\newcommand{\Ltot}{L^{\text{tot}}}
\newcommand{\Mtot}{M^{\text{tot}}}

\newcommand{\p}{\boldsymbol{p}}

\newcommand{\bzero}{\boldsymbol{0}}

\newcommand{\sB}{\mathsf{B}}

\newcommand{\sZ}{\mathsf{Z}}
\newcommand{\sC}{\mathsf{C}}
\newcommand{\sH}{\mathsf{H}}
\newcommand{\Crel}{\mathsf{C}_{\mathrm{rel}}}
\newcommand{\Hrel}{\mathsf{H}_{\mathrm{rel}}}
\newcommand{\Brel}{\mathsf{B}_{\mathrm{rel}}}
\newcommand{\Zrel}{\mathsf{Z}_{\mathrm{rel}}}

\newcommand{\eL}{\mathcal{L}}

\newcommand{\eO}{\mathcal{O}}

\newcommand{\vac}{\ket{p}}

\newcommand{\reg}{\operatorname{reg}}

\newcommand{\Ort}{\operatorname{O}}

\newcommand{\dvol}{\operatorname{dvol}}

\newcommand{\Tr}{\operatorname{Tr}}
\newcommand{\tr}{\operatorname{tr}}

\newcommand{\hh}{\mathbb{h}}

\newcommand{\RR}{\mathbb{R}}
\newcommand{\MM}{\mathbb{M}}
\newcommand{\LL}{\mathbb{L}}
\newcommand{\NN}{\mathbb{N}}

\newcommand{\ZZ}{\mathbb{Z}}

\newcommand{\CC}{\mathbb{C}}

\newcommand{\1}{\mathbb{1}}

\newcommand{\GL}{\operatorname{GL}}
\newcommand{\ISO}{\operatorname{ISO}}
\newcommand{\SO}{\operatorname{SO}}

\newcommand{\U}{\operatorname{U}}

\newcommand{\D}{{\partial}}

\newcommand{\zAdS}{\mathsf{AdS}}

\newcommand{\vv}{\mathbbm{v}}
\newcommand{\ee}{\mathbbm{e}}
\renewcommand{\tt}{\mathbbm{t}}
\newcommand{\qq}{\mathbbm{q}}

\newcommand{\pder}[2]{\frac{\partial #1}{\partial #2}}

\definecolor{dkgr}{rgb}{0,0.6,0}
\definecolor{gris}{rgb}{0.5,0.5,0.5}

\allowdisplaybreaks[1]
\numberwithin{equation}{section}
\AtAppendix{\counterwithin{lemma}{section}}

\begin{document}

\title{Quantum carrollian bosonic strings}
\author[Figueroa-O'Farrill]{José Figueroa-O'Farrill}
\author[Have]{Emil Have}
\author[Obers]{Niels A. Obers}
\address[JMF]{Maxwell Institute and School of Mathematics, The University
  of Edinburgh, James Clerk Maxwell Building, Peter Guthrie Tait Road,
  Edinburgh EH9 3FD, Scotland, United Kingdom}
\address[EH,NO]{Center of Gravity, Niels Bohr Institute, University of Copenhagen, Blegdamsvej 17, DK-2100 Copenhagen Ø, Denmark}
\address[NO]{Nordita, KTH Royal Institute of Technology and Stockholm University, Hannes Alfv\'{e}ns v\"{a}g 12, SE-106 91 Stockholm, Sweden}
\email[JMF]{\href{mailto:j.m.figueroa@ed.ac.uk}{j.m.figueroa@ed.ac.uk}, ORCID: \href{https://orcid.org/0000-0002-9308-9360}{0000-0002-9308-9360}}
\email[EH]{\href{mailto:emil.have@nbi.ku.dk}{emil.have@nbi.ku.dk}, ORCID: \href{https://orcid.org/0000-0001-8695-3838}{0000-0001-8695-3838}}
\email[NO]{\href{mailto:obers@nbi.ku.dk}{obers@nbi.ku.dk}, ORCID: \href{https://orcid.org/0000-0003-4947-8526}{0000-0003-4947-8526
}}
\begin{abstract}
  We study a recently discovered carrollian bosonic string, described
  classically by a sigma model where both the worldsheet and the
  target spacetime are carrollian.  After fixing the carrollian
  analogue of the conformal gauge, we determine the Lie algebra of
  residual gauge symmetries and show it is isomorphic to the
  three-dimensional extended BMS algebra. We quantise the sigma model
  à la BRST and determine the spectrum of the corresponding string
  theory by computing the BRST cohomology.  In contrast with the usual
  bosonic string, the spectrum of the carrollian string is
  finite-dimensional.  The cohomology displays Poincaré duality and
  can be interpreted, for a given momentum, as inducing
  representations for unitary irreducible representations of the
  26-dimensional Carroll group.  Furthermore we interpret (most of)
  the cohomology as deformations of the spacetime carrollian structure
  augmented by the Kalb--Ramond field to which the carrollian string
  couples.
\end{abstract}
\maketitle
\tableofcontents

\section{Introduction}
\label{sec:introduction}

Originally billed as a ``theory of everything'', promising to unify
gravity with all other known forces, string theory has since blossomed
into a broad subject with particularly deep connections to
mathematics.  Within the last decade or so, a new direction of
research has emerged within string theory that upends the lorentzian
hegemony that has long reigned supreme, by investigating what happens
to strings when the target spacetime and/or the worldsheet are instead
described by \textbf{nonlorentzian geometry}, where, roughly speaking,
the structure group of the tangent bundle does not preserve a
(lorentzian) inner product.  The precursor of such nonlorentzian
strings was the Gomis--Ooguri
string~\cite{Gomis:2000bd,Danielsson:2000gi}, which arose by
considering a near-critical Kalb--Ramond field limit that generalises
the $c\to \infty$ limit (with $c$ the speed of light) known for
particles. The Gomis--Ooguri string has a flat target spacetime, and
it was only realised much later
\cite{Harmark:2017rpg,Bergshoeff:2018yvt,Bidussi:2021ujm} that
nonrelativistic strings couple to a string version of Newton--Cartan
geometry (see~\cite{Oling:2022fft} for a review of nonrelativistic
string theory). The Gomis--Ooguri string has a lorentzian worldsheet
and the associated conformal field theoretic description makes the
computation of its spectrum amenable to homological techniques, as is
the case for the usual lorentzian strings.  Recently, a bosonic string
reminiscent of the Gomis--Ooguri string has been constructed by
gauging a Wess--Zumino--Witten model \cite{Figueroa-OFarrill:2025nmo}
and that too is amenable to homological techniques.  Interestingly,
there exists a further limit that turns the worldsheet nonlorentzian,
leading to a novel type of strings that are related to spin matrix
theory~\cite{Harmark:2014mpa, Harmark:2017rpg, Harmark:2018cdl,
  Harmark:2020vll,Bidussi:2023rfs}.

In this paper, we study quantum aspects of \textbf{carrollian
  strings}. Carrollian symmetries arise when the speed of light goes
to zero, and were first discussed by Lévy-Leblond~\cite{Levy1965} and Sen
Gupta~\cite{SenGupta1966OnAA}.  Much of the recent interest
in carrollian symmetry and geometry stems from its connection with
flat space holography~\cite{Hartong:2015usd,
  Ciambelli:2018wre,Figueroa-OFarrill:2021sxz, Donnay:2022aba,
  Bagchi:2022emh, Donnay:2022wvx, Hartong:2025jpp}
(see~\cite{Donnay:2023mrd} for a review and further references).
Carrollian symmetry also arises when considering the small speed of
light (or ultra-local) expansion of General Relativity
\cite{Henneaux:1979vn,Dautcourt:1997hb,Hansen:2021fxi}, giving rise to
Carroll gravity actions (see also, e.g., \cite{Henneaux:2021yzg,
  Figueroa-OFarrill:2022mcy, Campoleoni:2022ebj, deBoer:2023fnj}).
Furthermore, Carroll symmetry might also be relevant for de Sitter
cosmology and inflation \cite{deBoer:2021jej,Blair:2025nno}.  In
addition, Carrollian symmetry rears its head in condensed matter
physics in the context of fractons~\cite{Bidussi:2021nmp,
  Figueroa-OFarrill:2023vbj, Figueroa-OFarrill:2023qty}, which are
hypothetical quasiparticles with restricted mobility relevant for
quantum computing. It appears in the description of flat band
structures~\cite{Bagchi:2022eui, Ara:2024fbr}, which for example arise
in magic-angle superconductivity in twisted bilayer graphene, where
the dispersion relation becomes trivial.  In all, there is by now a
fast growing literature on the broad array of aspects of Carroll
symmetry in field theory, gravity, holography, and we refer to the
recent review \cite{Bagchi:2025vri} for further references.

Given this wealth of physical examples where carrollian symmetry plays
a rôle, it is natural to investigate what we may learn by studying
strings with carrollian symmetries. As it turns out, the answer is
quite a lot: the ambitwistor string~\cite{Mason:2013sva,
  Casali:2015vta}, whose correlators encode massless scattering
through the Cachazo--He--Yuan
formalism~\cite{Cachazo:2013hca,Cachazo:2013iea,Cachazo:2014xea}, has
a carrollian worldsheet.\footnote{Another construction of a string
  theory with carrollian worldsheet and relativistic target spacetime
  was obtained in \cite{Parekh:2023xms} via current-current
  deformations of the usual bosonic string.}  Ambitwistor strings are
furthermore central to the celestial approach to flat space
holography, where it was shown in in~\cite{Adamo:2021zpw,
  Adamo:2022wjo} that the worldsheet operator product expansions
(OPEs) of vertex operators of four-dimensional ambitwistor string
theory generate the coefficients of celestial OPEs, and thus provides
a dynamical principle for computing these. Moreover, by looking at
strings near black hole horizons as seen by an observer at infinity,
one finds that the effective description of these strings acquires a
carrollian target spacetime, where the distance to the horizon emerges
as an ``effective'' speed of light~\cite{Hartong:2023yxo,
  Hartong:2024ydv}. In addition to this, carrollian strings show up in
a vast duality web comprising ordinary relativistic strings,
ambitwistor strings, and the Gomis--Ooguri string~\cite{Blair:2023noj,
  Gomis:2023eav, Blair:2024aqz}, which was shown to provide a novel
and deep reinterpretation of the anti-de Sitter/conformal field theory
(AdS/CFT) correspondence, where the field theory dual to gravity in
the bulk resides in an asymptotic region rather than on the
boundary. Within this interpretation, the dual field theory takes the
form of a nonrelativistic matrix theory arising as a near-BPS
limit. In particular, this perspective allows for immediate
generalisation by taking multiple near-BPS limits, leading to a
cascade of novel holographic dualities, most of which have a
nonlorentzian flavour. Very recently, it was shown
in~\cite{Blair:2025nno} (see also~\cite{Argandona:2025jhg} for related
results and~\cite{Fontanella:2025tbs} for another approach in this
context) that the same procedure that leads to a nonrelativistic
holographic description in the context of anti-de Sitter spacetime
suggests that a holographic description of de Sitter spacetime becomes
carrollian.  In particular, this involves $p$-brane generalisations of
carrollian geometry and the formulation naturally leads to string
worldsheet actions that probe such geometries.  While in the setup of
\cite{Blair:2023noj, Gomis:2023eav, Blair:2024aqz, Blair:2025nno} the
strings are not the fundamental (light) degrees of freedom, the
procedure nevertheless leads to interesting novel classes of string
theories. The present paper focuses on strings that probe ``particle''
(or $0$-brane) carrollian geometry, i.e., the carrollian
string.\footnote{On such ``particle'' carrollian backgrounds, the
  fundamental degrees of freedom are D-instantons rather than
  strings.}

Carrollian strings of the type we consider have both a carrollian
worldsheet and a carrollian target spacetime and first
appeared\footnote{The phase space form of this action in a flat target
  spacetime was derived in~\cite{Cardona:2016ytk}.}
in~\cite{Blair:2023noj,Gomis:2023eav,Hartong:2024ydv} (see
also~\cite{Casalbuoni:2024jmj,Harksen:2024bnh} for other recent work
on carrollian strings).  Like null
strings~\cite{Isberg:1993av,Bagchi:2013bga,Bagchi:2020fpr}, and
therefore in particular like the ambitwistor
string~\cite{Casali:2016atr}, the symmetry algebra of the worldsheet
CFT of the carrollian string is BMS$_3$. This algebra is an abelian
extension of the Witt algebra, and it has the following nonzero
brackets
\begin{equation}
\label{eq:BMS-intro}
    [L_n,L_m] = (n-m)L_{m+n},\qquad [L_n,M_m] = (n-m)M_{m+n},
\end{equation}
where we did not include central charges.

Using BRST methods, we will quantise the carrollian bosonic
string.\footnote{Quantisation of the supersymmetric version of the
  carrollian string was considered in~\cite{Chen:2025gaz}.}
This is made possible by recent advances in the understanding of BRST
quantisation of chiral BMS-like
theories~\cite{Figueroa-OFarrill:2024wgs}, and we shall make heavy use
of the machinery developed in that work. In
particular,~\cite{Figueroa-OFarrill:2024wgs} considers a one-parameter
family of BMS-like algebras $\g_\lambda$, which for $\lambda = -1$
reduces to the BMS$_3$ algebra~\eqref{eq:BMS-intro}. By reformulating
the description of $\g_\lambda$ in terms of the operator product
expansion of fields depending on a formal variable $z$, the BRST
operator for all $\g_\lambda$ is realised as a semi-infinite
differential of $\g_\lambda$. The realisation that the BRST cohomology
of a CFT in certain cases coincide with the semi-infinite cohomology
of the symmetry algebra goes back to~\cite{MR865483}.

\subsection*{Overview}

This paper is organised as follows. In
Section~\ref{sec:carroll-string} we study the sigma model description
of the carrollian bosonic string. Since both the worldsheet and target
spacetime geometries are carrollian, we start in
Section~\ref{sec:car-geom} with a brief review of carrollian geometry
using the language of $G$-structures, including a review of the
intrinsic torsion of carrollian $G$-structures. This is followed in
Section~\ref{sec:action-and-syms} by a discussion of the carrollian
string sigma model action and its symmetries, which include Carroll
boosts and Weyl rescalings. In Section~\ref{sec:gauge-fixing-WS} we
show that any two-dimensional carrollian surface (e.g., the
worldsheet) admits local coordinates where the carrollian structure is
``flat'' (to be defined below).  We only treat the case of the target
being Carroll spacetime, for which the gauge-fixed action is given by
\begin{equation*}
  S[X^i,\tilde P_i,X^0] = \int_\Sigma \dd t\dd x \left[ \tfrac12 \tau^2 \left(  \D_t X^i \tilde P_i +   \D_x X^i \D_x X^i \right) + \tfrac12 \D_t X^0\D_t X^0 \right],
\end{equation*}
where $\tau$ is the string tension, $X^0, X^i$ are the embedding
fields and $\tilde P_i$ are Lagrange multipliers imposing the
constraint $\D_t X^i = 0$.  The infinitesimal generators of residual
diffeomorphisms satisfy the (extended) BMS$_3$ algebra. In
Section~\ref{sec:mode-exps-and-constraints}, we derive the on-shell
mode expansions, along with the Poisson brackets between the currents
that originate from the constraints.  The modes satisfy the BMS$_3$
algebra in equation~\eqref{eq:BMS-intro}.  We introduce an auxiliary
complex variable (not to be interpreted as a worldsheet coordinate) to
allow us to rephrase the mode algebra in the language of operator
product expansions (OPEs).  We introduce $\Pi_i = \tfrac12
\tau^2\tilde P_i$ and $\Pi_0 = \D_t X^0$, satisfying the OPE
\begin{equation*}
  \Pi_\mu(z) X^\nu(w) = \frac{\delta_\mu^\nu \1}{z-w} + \reg,
\end{equation*}
with $\mu = 0,1,\dots,D$, in terms of which the BMS$_3$ generators take the form
\begin{equation*}
  T = \d X^\mu \Pi_\mu \qquad\text{and}\qquad M = \tfrac12 \Pi_0 \Pi_0 - \tfrac12 \tau^2 \delta_{ij} \d X^i \d X^j.
\end{equation*}
These fields satisfy the following OPEs:
\begin{equation*}
  \begin{aligned}
    T(z) T(w) &= \frac{\tfrac12 c_L \1}{(z-w)^4} + \frac{2 T(w)}{(z-w)^2} + \frac{\d T(w)}{z-w} + \reg,\\
    T(z) M(w) &= \frac{\tfrac12 c_M \1}{(z-w)^4} + \frac{2 M(w)}{(z-w)^2} + \frac{\d M(w)}{z-w} + \reg,\\
    M(z) M(w) &= \reg,
  \end{aligned}
\end{equation*}
with $c_L = 2(D+1)$ and $c_M = 0$.

Having investigated classical aspects of the carrollian string, we
then turn our attention to its BRST quantisation in
Section~\ref{sec:brst-quantisation}.  To set the stage, we first
provide a brief recap of BRST quantisation in
Section~\ref{sec:brief-recap-brst}.  For the BMS$_3$ algebra with
generating fields $T(z)$ and $M(z)$, the BRST operator is the zero
mode of the current
\begin{equation*}
  J = c T + C M + \tfrac12 c \Tgh + \tfrac12 C \Mgh,
\end{equation*}
where we have introduced fermionic ghost systems $(b,c)$ and $(B,C)$ of
weights $(2,-1)$, in terms of which
\begin{equation*}
  \Tgh = -2 b \d c - \d b c - 2 B \d C - \d B C
  \qquad\text{and}\qquad \Mgh = - 2 B \d c - \d B c.
\end{equation*}
The BRST operator squares to zero when $c_M = 0$ and $c_L = 52$, which
corresponds to a $26$-dimensional Carroll spacetime target.  In
Section~\ref{sec:brst-complex} we construct the BRST complex in
momentum space for the quantum carrollian string.  In contrast with
the case of the standard bosonic strings, which have a Kaluza--Klein
tower of massive states in the spectrum, the BRST complex of the
carrollian string is finite-dimensional, leading to a
finite-dimensional spectrum.  To compute the cohomology it is
convenient to first consider the relative BRST subcomplex (the kernel
of the Virasoro antighost zero mode $b_0$) and we discuss this
subcomplex explicitly in Section~\ref{sec:relative-subcomplex}, where
we also show that the absolute and relative complexes are related by
the short exact sequence~\eqref{eq:ses-brst} of differential
complexes, inducing a long exact sequence~\eqref{eq:les} in cohomology
which will be exploited later, once the relative cohomology has been
calculated.

The relative BRST cohomology is worked out in
Section~\ref{sec:relat-brst-cohom}, with each ghost number having its
own dedicated story in Sections~\ref{sec:hrel0}--\ref{sec:hrel5}, with
a summary provided in Section~\ref{sec:summary}. We find that there is
no cohomology unless the ``energy'' $p_0=0$. We therefore must
distinguish two cases: the generic case $p = (0, \p)$ with
$\p \neq \bzero$ and the degenerate case where $p= 0$. We expect from
both critical and noncritical bosonic strings that there should be an
enhancement of the cohomology at $p=0$ and this is indeed what
happens. The calculations are summarised in Table~\ref{tab:Hrel0p} for
$\p \neq \bzero$ and in Table~\ref{tab:Hrelpeqzero} for $p=0$, which
list the isomorphism class of the cochains, cocycles, coboundaries and
cohomologies as representations of the stabiliser of the momentum $p$.

In Section~\ref{sec:from-relat-absol} we exploit the long exact
sequence~\eqref{eq:les} relating the relative and absolute BRST
cohomologies to calculate the absolute BRST cohomology.  Once more
each ghost number has its own section, from Section~\ref{sec:H1p} to
Section~\ref{sec:H5p}, and a summary in
Section~\ref{sec:summary-absolute-cohomology}. The calculations are
summarised in Table~\ref{tab:Habs0p} for $\p \neq \bzero$ and in
Table~\ref{tab:Habspeqzero} for $p=0$.

Having derived the BRST cohomology of the quantum carrollian string,
in Section~\ref{sec:interpr-terms-carr} we interpret the cohomology in
terms of unitary irreducible representations (UIRs) of the
$26$-dimensional Carroll group. We find that the cohomology at $p=0$
gives rise to finite dimensional UIRs where the boosts and
translations act trivially, whereas the cohomology at $\p \neq \bzero$
gives infinite-dimensional UIRs which are carried by square-integrable
sections of homogeneous vector bundles over the $24$-sphere,
associated to the $\SO(24)$-representations given by the cohomology.
In Section~\ref{sec:interpr-terms-carr-1} we study the geometric
interpretation of the cohomology for the case $\p \neq \bzero$.  We
find that most of the cohomology classes can be accounted for as
deformations of the carrollian structure of Carroll spacetime,
augmented by a Kalb--Ramond field, which although vanishing in Carroll
spacetime, can in principle be turned on and to which the string may
couple.  There is the perhaps puzzling absence of a cohomology class
for the first-order deformation of the carrollian vector field along
itself. This geometric cohomology appears with four-fold multiplicity:
the different ``pictures'' corresponding to the ghost zero modes.  We
conclude the paper with Section~\ref{sec:conclusions}, where we make
some conclusions and point to future work we find interesting. The
paper has one appendix (Appendix~\ref{sec:acti-brst-diff}) where we
record some useful formulae for the action of the BRST differential on
fields that were used in the calculation of the relative BRST
cohomology in Section~\ref{sec:relat-brst-cohom}.

\section{The carrollian string}
\label{sec:carroll-string}

In this section, we discuss the carrollian string, which was
discovered in the context of decoupling limits of string theory
in~\cite{Blair:2023noj,Gomis:2023eav}, using a series of duality
transformations, and subsequently
in~\cite{Bagchi:2023cfp,Bagchi:2024rje} by looking at strings near
black holes,\footnote{In~\cite{Bagchi:2023cfp,Bagchi:2024rje}, the
  carrollian string we consider here was called the ``electric
  carrollian string''.} where the distance to the horizon provides an
effective speed of light, allowing the application of the expansion
techniques for string theory developed
in~\cite{Hartong:2021ekg,Hartong:2022dsx,Hartong:2024ydv}.  Before
that we review the basic notions of carrollian geometry which will
help us understand the carrollian string sigma model.

\subsection{Brief review of carrollian geometry}
\label{sec:car-geom}

A carrollian geometry on a $(D+1)$-dimensional manifold $M$ consists
of a nowhere-vanishing vector field $\kappa\in \mathscr{X}(M)$, called
the \textbf{carrollian vector field}, and a degenerate
positive-semidefinite tensor $h\in \Gamma(\odot^2 T^*M)$ that, roughly
speaking, measures spatial distances and is therefore known as the
\textbf{carrollian ruler}. The carrollian vector field lies in the
radical of $\kappa$ --- i.e., $h(\kappa,-) = 0$ --- and moreover spans
it. Together, the pair $(\kappa,h)$ make up a \textbf{carrollian
  structure} on $M$.

Alternatively, a carrollian structure on $M$ is a reduction of the
structure group of the frame bundle.  To see how this arises, we can
always choose local frames $(e_0,e_1,\dots,e_D)$ where $e_0 = \kappa$
(restricted to the domain of the local frame) and $e_1,\dots,e_D$ is
an orthonormal frame for $h$: $h(e_i, e_j) = \delta_{ij}$.  On
overlaps, two such frames are related by a local transformation taking
values in the subgroup $G \subset \GL(D+1,\RR)$ which preserves the
vector $(1,0,\dots,0)^t \in \RR^{D+1}$ and the symmetric bilinear form
with matrix
\begin{equation}
  \begin{pmatrix}
    0 & 0 \\
    0 & I_D
  \end{pmatrix},
\end{equation}
where $I_D$ is the $D\times D$ identity matrix.  The subgroup $G$ is
the \textbf{homogeneous Carroll group} and is isomorphic to the
euclidean group $\Ort(D) \ltimes \RR^D$, where $\RR^D$ acts as
carrollian boosts.  As a subgroup of $\GL(D+1,\RR)$, the structure
group $G$ is given by matrices of the form\footnote{The sign of
  $\Lambda^t$ in~\eqref{eq:first-matrix} is chosen such that we get a
  plus sign in the transformation of the local one-form dual to the
  carrollian vector field (cf.~\eqref{eq:carroll-boosts}).}
\begin{equation}
\label{eq:first-matrix}
  \begin{pmatrix}
    1 & -\Lambda^t \\
    0 & A
  \end{pmatrix}
\end{equation}
where $\Lambda^t = (\Lambda^1,\dots,\Lambda^D)$ and $A^t A = I_D$.  Its Lie algebra $\g$
is given by matrices of the form
\begin{equation}
  \begin{pmatrix}
    0 & -\lambda^t \\
    0 & X
  \end{pmatrix}
\end{equation}
where again $\lambda^t = (\lambda^1,\dots,\lambda^D)$ and now $X^t = -
X$.

Choose a basis $J_{ij},B_i$ for $\g$ such that the above matrix has the
expansion
\begin{equation}
  \tfrac12 X^{ij} J_{ij} - \lambda^i B_i.
\end{equation}
Relative to this basis, the Lie bracket of $\g$ is given by
\begin{equation}
  \begin{aligned}\relax
    [J_{ij},J_{k\ell}] &= \delta_{jk} J_{i\ell} - \delta_{ik} J_{j\ell} - \delta_{j\ell} J_{ik} + \delta_{i\ell} J_{jk}\\
    [J_{ij}, B_k] &= \delta_{jk} B_i - \delta_{ik} B_j\\
    [B_i, B_j] &= 0.
  \end{aligned}
\end{equation}
The action of $\g$ on the local frame is given by
\begin{equation}
  \begin{aligned}
    J_{ij} \cdot e_0 &= 0\\
    J_{ij} \cdot e_k &= \delta_{jk} e_i - \delta_{ik} e_j\\
  \end{aligned}
  \qquad\qquad
  \begin{aligned}
    B_i\cdot e_0 &= 0\\
    B_i \cdot e_j &= \delta_{ij} e_0\\
  \end{aligned}
\end{equation}
from where we deduce the action of $\g$ on the canonically dual
coframe $(\theta^0,\theta^i)$:
\begin{equation}
  \begin{aligned}
    J_{ij} \cdot \theta^0 &= 0\\
    J_{ij} \cdot \theta^k &= \left( \delta_i^k \delta_{j\ell}  - \delta_j^k \theta_{i\ell} \right)\theta^\ell\\
  \end{aligned}
  \qquad\qquad
  \begin{aligned}
    B_i \cdot \theta^0 &= - \delta_{ij} \theta^j\\
    B_i \cdot \theta^j &= 0.\\
  \end{aligned}
\end{equation}

We introduce the local one-form $\xi = \theta^0$ and the local
symmetric $(2,0)$-tensor field $\gamma = \delta^{ij} e_i \otimes e_j$.
On overlaps, both are invariant under rotations, but they transform
nontrivially under the carrollian boosts:
\begin{equation}  
  B_i \cdot \xi = - \delta_{ij} \theta^j \qquad\text{and}\qquad B_i
  \cdot\gamma = \kappa \otimes e_i + e_i \otimes \kappa.
\end{equation}

It is often desirable to describe these objects relative to a local
chart $x^\mu$, relative to which
\begin{equation}
  \kappa = \kappa^\mu \d_\mu, \quad h = h_{\mu\nu} dx^\mu dx^\nu,
  \quad \xi = \xi_\mu dx^\mu, \quad\text{and}\quad \gamma = \tfrac12
  \gamma^{\mu\nu} (\d_\mu \otimes \d_\nu + \d_\nu \otimes \d_\mu).
\end{equation}
These local functions $\kappa^\mu = e_0^\mu$, $\xi_\mu =
\theta^0_\mu$, $h_{\mu\nu} = \delta_{ij} \theta^i_\mu \theta^j_\nu$ and
$\gamma^{\mu\nu} = \delta^{ij} e_i^\mu e_j^\nu$ are subject to the
following equations, restating the fact that $(e_0,e_i)$ is a frame
in our $G$-structure and $(\theta^0,\theta^i)$ is the canonically dual
coframe:
\begin{equation}
  \xi_\mu \kappa^\mu = 1, \quad h_{\mu\nu} \kappa^\nu = 0, \quad
  \gamma^{\mu\nu} \xi_\nu = 0, \quad\text{and}\quad h_{\mu\rho}
  \gamma^{\rho\nu} + \xi_\mu \kappa^\nu = \delta_\mu^\nu.
\end{equation}

On overlaps, $\kappa^\mu$ and $h_{\mu\nu}$ are of course invariant,
but not so $\xi_\mu$ and $\gamma^{\mu\nu}$, which transform under
infinitesimal Carroll boosts as follows:
\begin{equation}
  \delta_C \xi_\mu = -\lambda^i B_i \cdot \theta^0_\mu = \lambda^i
  \delta_{ij} \theta^j_\mu = \lambda_\mu,
\end{equation}
which defines the local one-form $\lambda_\mu := \lambda^i \delta_{ij}
\theta^j_\mu$, and
\begin{equation}
  \delta_C \gamma^{\mu\nu} = -\lambda^k B_k \cdot (\delta^{ij} e_i^\mu
  e_j^\nu) = -\lambda^k \delta^{ij} (\delta_{ki} e_0^\mu e_j^\nu +
  \delta_{kj} e_i^\mu e_0^\nu ) = -2 \lambda^{(\mu} \kappa^{\nu)},
\end{equation}
where we have defined the local vector field $\lambda^\mu := \lambda^i
e_i^\mu$.  Since $\lambda^\mu = \gamma^{\mu\nu}\lambda_\nu$, we may
rewrite the previous equation also as
\begin{equation}
  \delta_C \gamma^{\mu\nu} =-2 \lambda_\rho \gamma^{\rho(\mu} \kappa^{\nu)}.
\end{equation}
In summary, in overlaps the local data $(\kappa,h,\xi,\gamma)$
transforms under local homogeneous Carroll transformations in such a
way that they are all rotationally invariant and under infinitesimal
Carroll boosts,
\begin{equation}
  \label{eq:carroll-boosts}
  \delta_C \kappa^\mu = 0, \quad \delta_C h_{\mu\nu} = 0, \quad
  \delta_C \xi_\mu =  \lambda_\mu \quad\text{and}\quad \delta_C
  \gamma^{\mu\nu} = -2 \lambda_\rho \gamma^{\rho(\mu} \kappa^{\nu)}.
\end{equation}

The Carroll $G$-structure is \textbf{integrable} if the adapted frames
can be chosen to be coordinate frames.  Equivalently, if we can find
local coordinates $x^\mu$ where
\begin{equation}
\label{eq:flat-carroll}
    \kappa^\mu = \delta^\mu_0,\quad \xi_\mu = \delta^0_\mu, \quad
    h_{\mu\nu} = \delta_\mu^i\delta_\nu^j\delta_{ij}
    \quad\text{and}\quad \gamma^{\mu\nu} =
    \delta^\mu_i\delta^{\nu}_j\delta^{ij}.
\end{equation}
This is roughly equivalent to the existence of a flat Cartan
connection for the $G$-structure, so that the carrollian manifold is
locally isomorphic to the Klein model, which here is Carroll
spacetime: namely, the homogeneous space $C/G$, where $C$ is the
Carroll group and $G$ the homogeneous Carroll group.

Carroll $G$-structures admit intrinsic torsion: a feature absent from
lorentzian geometry due to the existence and uniqueness of the
Levi-Civita connection.  The intrinsic torsion of a carrollian
structure $(\kappa, h)$ is given by $\eL_\kappa h$, which is analogous
to the extrinsic curvature of a non-degenerate hypersurface in a
(pseudo-)riemannian manifold.  Using that analogy, we may classify
\cite{Figueroa-OFarrill:2020gpr} carrollian $G$-structures as totally
geodesic ($\eL_\kappa h= 0$), totally umbilical ($\eL_\kappa h = f h$),
minimal ($\Tr \eL_\kappa h= 0$) or generic.  In general dimension,
these four classes are distinct, but in two dimensions they collapse
to only two, depending on whether or not $\eL_\kappa h = 0$.

\subsection{Action and worldsheet symmetries}
\label{sec:action-and-syms}

The action of the Carroll string on a carrollian background is
\cite{Gomis:2023eav,Bagchi:2024rje}
\begin{align}
\label{eq:carroll-action}
\begin{aligned}
    S &= \int_\Sigma \,\dvol(\Sigma)\left[  \tfrac12\tau^2 \left( \vv^\alpha \D_\alpha X^\mu \tilde P_\mu + h_{\mu\nu} \ee^\alpha\ee^\beta \D_\alpha X^\mu \D_\beta X^\nu \right) + \tfrac12 \xi_\mu \xi_\nu \vv^\alpha \vv^\beta \D_\alpha X^\mu\D_\beta X^\nu \right]+\int_\Sigma X^*B,
\end{aligned}
\end{align}
where $\tau$ is the string tension, $X^\mu$ is the embedding field and
$\tilde P_\mu$ is a spatial (i.e., $\kappa^\mu \tilde P_\mu = 0$)
Lagrange multiplier, and where we included a Wess--Zumino term involving
the pullback of the two-form Kalb--Ramond field $B_{\mu\nu}$ to the worldsheet. The 
pair $(\vv^\alpha,\ee^\alpha)$ with
$\alpha,\beta = 0,1$ is a carrollian frame on the worldsheet $\Sigma$
with coordinates $\sigma^\alpha$.  More precisely, $\vv^\alpha$ is the
carrollian vector field of the carrollian worldsheet, while
$\ee^\alpha$ is the completion of $\vv^\alpha$ to a frame on the
worldsheet. The dual coframe is given by
\begin{equation}
    (\tt_\alpha,\qq_\alpha),
\end{equation}
where
\begin{equation}
  \tt_\alpha \vv^\alpha = \qq_\alpha \ee^\alpha = 1,\qquad \qq_\alpha \vv^\alpha = \ee^\alpha \tt_\alpha = 0,\qquad \delta^\alpha_\beta = \vv^\alpha\tt_\beta + \ee^\alpha\qq_\beta.
\end{equation}
The worldsheet carrollian structure is therefore $(\vv, \hh = \qq^
2)$.  The volume form $\dvol(\Sigma)$ on the worldsheet is given by
\begin{equation}
  \label{eq:volume-form}
  \dvol(\Sigma) = \tt\wedge\qq = (\tt_0 \qq_1 - \tt_1 \qq_0) d\sigma^0\wedge d\sigma^1 =: \mathbb{E}\,d^2\sigma,
\end{equation}
where we have introduced $\mathbb{E}= \epsilon^{\alpha\beta}
\tt_\alpha \qq_\beta$.  The Carroll string action~\eqref{eq:carroll-action}
is invariant under target spacetime Carroll boosts, which act as
(cf.~\eqref{eq:carroll-boosts}) 
\begin{equation}
\label{eq:target-space-syms}
  \delta_C \tilde P_\mu = -2\tau^{-2}\vv^\alpha\D_\alpha X^\nu\xi_\nu \lambda_\mu,\qquad \delta_C \xi_\mu = \lambda_\mu,\qquad \delta_C X^\mu = 0.
\end{equation}
In addition to this, the worldsheet possesses $2$-dimensional
conformal carrollian (or BMS$_3$) symmetry. The worldsheet Carroll
boosts, parameterised by $\lambda_\Sigma$, act as
\begin{equation}
  \label{eq:WS-carroll}
  \delta_{C_\Sigma} \tt_\alpha = \lambda_\Sigma\qq_\alpha,\qquad
  \delta_{C_\Sigma}\qq_\alpha = \delta_{C_\Sigma} \vv^\alpha = 0,
  \qquad \delta_{C_\Sigma} \ee^\alpha = -\lambda_\Sigma
  \vv^\alpha,\qquad \delta_{C_\Sigma} \tilde P_\mu = 2\lambda_\Sigma
  \ee^\alpha \D_\alpha X^\nu h_{\nu\mu},
\end{equation}
where we wrote ``${C_\Sigma}$'' to emphasise that these are worldsheet
Carroll boost transformations in contrast to target spacetime Carroll
boosts $C$ that appear in~\eqref{eq:carroll-boosts} and above. The 
embeddings scalars $X^\mu$ are invariant under worldsheet Carroll 
boosts, i.e., $\delta_{C_\Sigma} X^\mu = 0$.
The worldsheet Weyl symmetry $W_\Sigma$ with parameter $\omega_\Sigma$
acts as 
\begin{align}
\label{eq:WS-Weyl}
\begin{aligned}
    \delta_{W_\Sigma} \tt_\alpha &= \omega_\Sigma \tt_\alpha,& \delta_{W_\Sigma} \qq_\alpha &= \omega_\Sigma \qq_\alpha,& \delta_{W_\Sigma} \vv^\alpha &= -\omega_\Sigma \vv^\alpha,\\
    \delta_{W_\Sigma} \ee^\alpha &= -\omega_\Sigma \ee^\alpha,& \delta_{W_\Sigma} \tilde P_\mu &= -\omega_\Sigma \tilde P_\mu, &\delta_{W_\Sigma} X^ \mu &= 0.
\end{aligned}
\end{align}
The determinant $\mathbb{E}$ defined in~\eqref{eq:volume-form} transforms under worldsheet Weyl symmetry as
\begin{equation}
    \delta_{W_\Sigma} \mathbb{E} = 2\omega_\Sigma \mathbb{E}.
\end{equation}
In the language of~\cite{deBoer:2021jej}, the Carroll
string~\eqref{eq:carroll-action} is made up of an electric
Carroll scalar field theory in the timelike direction and a magnetic
Carroll scalar field theory in the spacelike directions.

Taking the target to be the Carroll spacetime -- so that
\eqref{eq:flat-carroll} holds -- and taking the background Kalb--Ramond field
to vanish, i.e.,
\begin{equation}
    B_{\mu\nu} = 0,
\end{equation}
the Carroll string action is given by
\begin{equation}
  \label{eq:flat-string}
    S[\vv^\alpha,\ee^\alpha,X^i,\tilde P_i,X^0] = \int_\Sigma
    \dd^2\sigma\,\mathbb{E}\left[ \tfrac12 \tau^2 \left( \vv^\alpha \D_\alpha
        X^i \tilde P_i +  \ee^\alpha\ee^\beta \D_\alpha X^i \D_\beta
        X^i \right) + \tfrac12 \vv^\alpha \vv^\beta \D_\alpha X^0 \D_\beta X^0 \right],
\end{equation}
where $\sigma^\alpha = (\sigma^0,\sigma^1)$ are the worldsheet
coordinates. These are local coordinates on the cylinder $\RR \times
S^1$, where $\sigma^0 \in \RR$ and $\sigma^1 \in
\RR/2\pi\ZZ$. We will assume closed string boundary conditions for the
embedding fields of the strings, i.e.,
\begin{equation}
\label{eq:boundary-conds-embedding-fields}
  X^\mu(\sigma^0,\sigma^1+2\pi) = X^\mu(\sigma^0,\sigma^1),
\end{equation}
where $\mu= (i,0)$.  The constraint imposed by $\tilde P_i$ tells us
that
\begin{equation}
\label{eq:constraint}
   \chi^i := \tfrac12\tau^2\vv^\alpha \D_\alpha X^i = 0.
\end{equation}
The variation of the action~\eqref{eq:flat-string} gives
\begin{equation}
    \delta S = \int_\Sigma \dd^2\sigma\,\mathbb{E}\left[ J_\alpha \delta \vv^\alpha + Q_\alpha \delta \ee^\alpha + \chi^i\delta \tilde P_i + \mathcal{E}_i \delta X^i + \mathscr{e}\delta X^0 \right],
\end{equation}
where, since we are dealing with closed strings, we discarded boundary
terms. The currents obtained by varying the carrollian data on the
worldsheet are given by
\begin{equation}
\label{eq:currents}
\begin{aligned}
    J_\alpha &= -\tfrac12\tau^2\tt_\alpha \ee^\beta \ee^\gamma \D_\beta X^i \D_\gamma X^i + \tfrac12\tau^2 \qq_\alpha \ee^\beta \D_\beta X^i \tilde P_i - \tfrac{1}{2} \tt_\alpha \vv^\beta \vv^\gamma \D_\beta X^0 \D_\gamma X^0 + \D_\alpha X^0 \vv^\beta \D_\beta X^0 ,\\
    Q_{\alpha} &= \tfrac12\tau^2( 2\tt_\alpha\vv^\gamma + \qq_\alpha \ee^\gamma )\ee^\beta \D_\beta X^i \D_\gamma X^i - \tfrac12\tau^2 \qq_\alpha \vv^\beta \D_\beta X^i \tilde P_i - \tfrac{1}{2} \qq_\alpha \vv^\beta \vv^\gamma \D_\beta X^0 \D_\gamma X^0 ,
\end{aligned}
\end{equation}
where we used that
\begin{equation}\label{eq:delta-E}
  \delta \mathbb{E} = - \mathbb{E}\big( \tt_\alpha\delta \vv^\alpha + \qq_{\alpha}\delta\ee^{\alpha} \big),
\end{equation}
as well as
\begin{equation}
  \begin{aligned}
    \delta \tt_\alpha &= -\tt_\alpha \tt_\beta \delta \vv^\beta -  \qq_\alpha \tt_\beta \delta \ee^\alpha\\
    \delta \qq_\alpha &= -\qq_\alpha \qq_\beta \delta \ee^\beta -  \tt_\alpha \qq_\beta \delta \vv^\alpha.
  \end{aligned}
\end{equation}
In particular, these allow us to derive~\eqref{eq:delta-E} by using $\tt \wedge \qq = \mathbb{E} d^2\sigma$. The equations of motion for the embedding fields are
\begin{equation}
  \label{eq:eoms}
  \mathcal{E}_i = -\tfrac12\tau^2 \mathbb{E}^{-1} \D_\alpha (\mathbb{E}\vv^\alpha \tilde P_i) -\tau^2 \mathbb{E}^{-1} \D_\alpha (\mathbb{E}\ee^\alpha \ee^\beta \D_\beta X_i) ,\qquad \mathscr{e} = -\mathbb{E}^{-1}\D_\alpha(\mathbb{E}\vv^\alpha \vv^\beta \D_\beta X^0).
\end{equation}
On-shell configurations of~\eqref{eq:flat-string} satisfy the equations of motion for the fields, i.e.,
\begin{equation}
\label{eq:on-shell-conf}
    \begin{aligned}
        \chi^i &= \mathcal{E}_i = \mathscr{e} = 0,\\
        J_\alpha &= Q_\alpha = 0.
    \end{aligned}
\end{equation}
Later, when we fix the gauge on the worldsheet, the conditions on the second line will turn into constraints. 

The currents $J_\alpha$ and $Q_\alpha$ are constrained by worldsheet Carroll boost invariance~\eqref{eq:WS-carroll} and Weyl invariance~\eqref{eq:WS-Weyl}. The Ward identities for these symmetries imply that
\begin{equation}
\begin{aligned}
    C_\Sigma\text{-symmetry:}&\qquad  \vv^\alpha Q_\alpha = 2\chi^i \ee^\alpha\D_\alpha X^i ,\\
    W_\Sigma\text{-symmetry:}&\qquad  \vv^\alpha J_\alpha + \ee^\alpha Q_\alpha =  -\chi^i\tilde P_i.
\end{aligned}
\end{equation}
These Ward identities are always satisfied, and, in particular, automatically hold when using the explicit expressions for the currents we found above. The energy-momentum tensor $T^\alpha{_\beta}$ of a two-dimensional conformal carrollian field theory is defined as
\begin{equation}
\label{eq:EMT-def}
    T^\alpha{_\beta} = T^\alpha \tt_\beta + \mathcal{T}^\alpha \qq_\beta,
\end{equation}
in terms of the responses to varying the action with respect to $(\tt_\alpha,\qq_\alpha)$ as
\begin{equation}
\label{eq:current-var}
    \delta_{\tt,\qq} S = \int \dd^2\sigma \mathbb{E}\left( T^\alpha\delta\tt_\alpha + \mathcal{T}^\alpha \delta \qq_\alpha \right).
\end{equation}
These currents are subject to the Ward identities
\begin{equation}
\label{eq:WIs-for-Ts}
\begin{aligned}
    C_\Sigma\text{-symmetry:}&\qquad  T^\alpha\qq_\alpha = -2\chi^i \ee^\alpha\D_\alpha X^i ,\\
    W_\Sigma\text{-symmetry:}&\qquad  T^\alpha \tt_\alpha + \mathcal{T}^\alpha \qq_\alpha =  \chi^i\tilde P_i.
\end{aligned}
\end{equation}
Using the relations 
\begin{equation}
\begin{aligned}
    \delta \vv^\alpha &= -\vv^\alpha \vv^\beta\delta\tt_\beta - \ee^\alpha \vv^\beta\delta\qq_\beta,\\
    \delta \ee^\alpha &= -\vv^\alpha \ee^\beta \delta\tt_\beta - \ee^\alpha \ee^\beta \delta\qq_\beta,
\end{aligned}
\end{equation}
we find that
\begin{equation}
\begin{aligned}
    T^\alpha &= -\vv^\alpha \vv^\beta J_\beta - \ee^\alpha \vv^\beta Q_\beta,\\
    \mathcal{T}^\alpha &= -\vv^\alpha \ee^\beta J_\beta - \ee^\alpha \ee^\beta Q_\beta,
\end{aligned}
\end{equation}
and hence the energy-momentum tensor written in terms of $J_\alpha$ and $Q_\alpha$ is
\begin{equation}
\label{eq:EMT-relation}
    T^\alpha{_\beta} = -\tt_\beta(\vv^\alpha \vv^\gamma J_\gamma + \ee^\alpha \vv^\gamma Q_\gamma) - \qq_\beta(\vv^\alpha \ee^\gamma J_\gamma +\ee^\alpha \ee^\gamma Q_\gamma).
\end{equation}
Since we will not need the explicit expression, we refrain from writing it.

\subsection{Fixing the gauge and residual gauge symmetries}
\label{sec:gauge-fixing-WS}

The carrollian structure on the worldsheet is given by $(\vv,\hh)$, where
\begin{equation}
    \hh_{\alpha\beta} = \qq_\alpha\qq_\beta,
\end{equation}
is the carrollian ruler on the worldsheet. As mentioned above, in two
dimensions there are only two intrinsic-torsion classes of carrollian
geometries~\cite{Figueroa-OFarrill:2020gpr}:
\begin{enumerate}
    \item[{(i)}] either the intrinsic torsion vanishes, $\eL_\vv \hh= 0$,
    \item[{(ii)}] or the intrinsic torsion is generic: $\eL_\vv \hh \neq 0$.
\end{enumerate}
We will see below that even if $(\vv,\hh)$ has non-vanishing intrinsic
torsion, we may find a Weyl-related carrollian structure
$(e^{-\omega_\Sigma}\vv, e^{2\omega_\Sigma}\hh)$ that has vanishing
intrinsic torsion.  We will then show that if $(\vv,\hh)$ has
vanishing intrinsic torsion there is a carrollian analogue of the
conformal gauge whose residual gauge symmetries are determined to be
isomorphic to the BMS$_3$ algebra.

We first notice that
\begin{align*}
  \eL_{\vv} \hh &= 2 \qq \eL_{\vv} \qq \\
                &= 2 \qq \iota_\vv d\qq & \tag{since $\iota_\vv \qq = 0$}\\
                &= 2 \qq (f \qq + g \tt) & \tag{since $(\tt,\qq)$ is a coframe}\\
                &= 2 f \qq^2. & \tag{since $g = \iota_\vv (f \qq + g \tt) = \iota_\vv \iota_\vv d\qq = 0$}
\end{align*}

Now let $\vv' = e^{-\omega_\Sigma} \vv$ and $\hh' = e^{2\omega_\Sigma}
\hh$ and let us calculate the intrinsic torsion of $(\vv', \hh')$:
\begin{align*}
  \eL_{\vv'} \hh' &= \eL_{\vv'} \left( e^{2\omega_\Sigma} \qq^2 \right)\\
                  &= 2 e^{2\omega_\Sigma} (\eL_{\vv'} \omega_\Sigma \qq^2 + \qq \eL_{\vv'} \qq)\\
                  &= 2 e^{\omega_\Sigma} \qq^2 \eL_{\vv} \omega_\Sigma + 2 e^{2\omega_\Sigma} \qq \eL_{\vv'} \qq & \tag{since $\eL_{\vv'} \omega_\Sigma = e^{-\omega_\Sigma} \eL_\vv \omega_\Sigma$}\\
  &= 2 e^{\omega_\Sigma} \qq^2 \eL_{\vv} \omega_\Sigma + 2 e^{\omega_\Sigma} \qq \eL_{\vv} \qq & \tag{since $\eL_{\vv'} \qq = \iota_{\vv'} d\qq = e^{-\omega_\Sigma} \iota_\vv d\qq$}\\
  &= 2 e^{\omega_\Sigma} \left(\eL_{\vv}\omega_\Sigma + f \right) \hh. & \tag{since $\iota_\vv d\qq = f \qq$}
\end{align*}
Therefore if we choose $\omega_\Sigma$ such that $\eL_{\vv}
\omega_\Sigma = -f$, we see that $\eL_{\vv'} \hh' = 0$.

Since the action~\eqref{eq:carroll-action} is invariant under
Weyl rescalings, we may assume without loss of generality that
$(\vv,\hh)$ has vanishing intrinsic torsion.

That being the case, we first use the flow-box theorem to choose
local coordinates $(t,y)$ adapted to the carrollian vector:
\begin{equation}
\label{eq:car-vec-field-gauge-fixed}
    \vv = \D_t.
\end{equation}
Since $\hh(\vv,-) = 0$, the carrollian ruler is
\begin{equation}
    \hh = F(t,y) dy^2.
\end{equation}
Since $F > 0$, we may assume that $F(t,y) = f(t,y)^2$ for a
nowhere-vanishing function $f$.  Since $\hh =
\qq^2$, it follows that $\qq = f(t,y) dy$.  Since the intrinsic
torsion vanishes, $\d_t f(t,y) = 0$ and hence $\qq = f(y) dy$, so that
$d\qq = 0$. Therefore $\qq$ is locally exact: $\qq = dx$ for some
$x(y)$.  In the new coordinates $(t,x)$ the carrollian structure is
given by $(\d_t, dx^2)$.

Any other adapted frame $(\d_t, \ee)$ will be of the form $\ee = \d_x
+ g(t,x) \d_t$ with canonically dual coframe $(\tt, dx)$ with $\tt =
dt - g(t,x) dx$.  Under \textit{finite} worldsheet Carroll boosts,
$\tt$ and $\ee$ transform according to
\begin{equation}
    \begin{aligned}
        \tt \to \tt + \Lambda_\Sigma\qq,\\
        \ee \to \ee - \Lambda_\Sigma \vv,
    \end{aligned}
\end{equation}
where $\Lambda_\Sigma(t,x)$ is the parameter of a finite worldsheet
Carroll boost. Choosing $\Lambda_\Sigma(t,x) = g(t,x)$ leads to
\begin{equation}
\label{eq:inv-gauge-fixed}
    \begin{aligned}
        \tt &= dt,\\
        \ee &= \d_x,
    \end{aligned}
\end{equation}
which completes the gauge fixing procedure for the
action~\eqref{eq:carroll-action}. Similar to the discussion around
Eq.~\eqref{eq:boundary-conds-embedding-fields}, the fact that we are
considering closed strings means that we take all functions to be
periodic in $x$ with period $2\pi$.

Going to conformal gauge only partially fixes the gauge: there remains
those diffeomorphisms that can be undone by Carroll boosts and Weyl
transformations. We may express the vector field $\zeta$ that
generates these residual diffeomorphisms as
\begin{equation}
  \zeta = \zeta^t(t,x) \D_t + \zeta^x(t,x) \D_x,
\end{equation}
where $\zeta^t$ and $\zeta^x$ are periodic in $x$ with period
$2\pi$. The explicit expression for the vector $\zeta$ generating the
residual diffeomorphisms are determined by solving the equations
\begin{equation}
\label{eq:residual-conditions}
    \begin{aligned}
        0 &= \eL_\zeta \vv - \omega^R_\Sigma \vv,\\
        0 &= \eL_\zeta \ee - \omega^R_\Sigma \ee - \lambda^R_\Sigma \vv,\\
        0 &= \eL_\zeta\qq + \omega^R_\Sigma \qq,\\
        0 &= \eL_\zeta \tt + \omega^R_\Sigma \tt + \lambda_\Sigma^R\qq,
    \end{aligned}
\end{equation}
where $\omega^R_\Sigma$ and $\lambda_\Sigma^R$ are the (periodic) Weyl
and boost parameters in terms of which we will express the residual
diffeomorphisms. Using the expressions $(\vv,\ee) = (\d_t,\d_x)$ and
$(\tt,\qq) = (dt, dx)$ for the gauge-fixed carrollian data, the
conditions in~\eqref{eq:residual-conditions} become
\begin{equation}
    \begin{aligned}
        \D_t\zeta^t &=  -\omega^R_\Sigma,\\
        \D_t\zeta^x &= 0 ,\\
        \D_x \zeta^t &= -\lambda_\Sigma^R,\\
        \D_x\zeta^x &= -\omega^R_\Sigma.
    \end{aligned}
\end{equation}
The second equation tells us that $\zeta^x$ is independent of $t$, and so we may write
\begin{equation}
    \zeta^x = \sum_{n\in \ZZ} a_n i e^{inx},
\end{equation}
while the fourth equation now restricts the form of the possible
residual Weyl transformations $\omega^R_\Sigma$ and relates it to
$\zeta^x$ as $\D_x\zeta^x(x) = -\omega^R_\Sigma$. The fact that the
parameters $\omega_\Sigma^R$ and $\lambda^R_\Sigma$ are restricted is
not surprising and also happens in the case of a lorentzian worldsheet.
The first equation then tells us that
\begin{equation}
    \zeta^t = \sum_{n\in \ZZ} (i n a_n) it e^{inx} + \sum_{n\in\ZZ}b_n i e^{inx},
\end{equation}
while the third equation relates the modes $b_n$ to
$\lambda^R_\Sigma$. Hence, we can write the generator of the residual
diffeomorphisms as
\begin{equation}
    \zeta = \sum_{n\in \ZZ} a_n L_n + \sum_{n\in \ZZ} b_n M_n,
\end{equation}
where we defined
\begin{equation}
    L_n := i e^{inx}(\D_x + int \D_t ),\qquad M_n := i e^{inx}\D_t.
\end{equation}
These satisfy the brackets
\begin{equation}
\label{eq:bms3-alg}
    \begin{aligned}\relax
        [L_n,L_m] &= (n-m)L_{m+n},\\
        [L_n,M_m] &= (n-m) M_{m+n},\\
        [M_n,M_m] &= 0.
    \end{aligned}
\end{equation}
This is the (centreless, extended) BMS$_3$ algebra. 

\subsection{On-shell mode expansions and constraints}
\label{sec:mode-exps-and-constraints}

In conformal gauge, the action for the carrollian string~\eqref{eq:flat-string} is
\begin{equation}
\label{eq:flat-string-gauge-fixed}
    S[X^i,\tilde P_i,X^0] = \int_\Sigma \dd t\dd x \left[ \tfrac12\tau^2 \left(  \D_t X^i \tilde P_i +   \D_x X^i \D_x X^i \right) + \tfrac12 \D_t X^0\D_t X^0 \right].
\end{equation}
The equations of motion for the fields that appear in the action above are (cf.~\eqref{eq:constraint} and~\eqref{eq:eoms})
\begin{equation}
\label{eq:eoms-gauge-fixed}
    \begin{aligned}
        \D_t X^i &= 0,\\
        \D_t^2 X^0 &= 0,\\
        \D_x^2 X^i &= -\tfrac{1}{2} \D_t \tilde P_i.
    \end{aligned}
\end{equation}
The canonical momenta conjugate to $X^\mu$ are given by
\begin{equation}
    \begin{aligned}
        \Pi_i = \pder{\mathcal{L}}{\D_t X^i} = \tfrac12\tau^2\tilde P_i,\qquad  
        \Pi_0 = \pder{\mathcal{L}}{\D_t X^0} = \D_t X^0 ,
    \end{aligned}
\end{equation}
where $\mathcal{L}$ is the lagrangian density of the
action~\eqref{eq:flat-string-gauge-fixed}, i.e., $S = \int_\Sigma  \dd
t\dd x \mathcal{L}$. The first equation implies that $X^i = X^i(x)$,
while the second tells us that
\begin{equation}
    X^0 = x^0(x) + t\Pi_0(x).
\end{equation}
The $t$-derivative of the third equation gives $\D_t^2 \Pi^i = 0$, and
the last equation in~\eqref{eq:eoms-gauge-fixed} now implies that
$\D_x^2 X^i = -\tfrac{1}{\tau^2} \D_t \Pi^i$, which we may combine to find
\begin{equation}
\label{eq:transverse-momentum}
    \Pi_i = \pi_i(x) - {t\tau^2}\D_x^2 X^i(x).
\end{equation}
In addition to the equations of motion in~\eqref{eq:eoms-gauge-fixed}, we must impose the constraints 
\begin{equation}
    J_\alpha = Q_\alpha = 0,
\end{equation}
where $J_\alpha$ and $Q_\alpha$ are given in~\eqref{eq:currents}. In conformal gauge, the components of these currents are
\begin{equation}
\label{eq:current-fields}
    \begin{aligned}
        J_t(x) &= -\tfrac12\tau^2 \D_x X^i(x) \D_x X^i(x) + \tfrac{1}{2}\Pi_0(x)^2,\\
        J_x(t,x) &= \D_x X^\mu(x) \Pi_\mu(t,x),\\
        Q_t &= 0,\qquad Q_x(x) = -J_t(x),
    \end{aligned}
\end{equation}
where we notice that $J_t(x)$ is time-independent. Furthermore, we note that the linear $t$-dependence in $J_x(t,x)$ is related to the $x$-derivative of $J_t$, i.e.,
\begin{equation}
\label{eq:Jx-Jt-rel}
    J_x(t,x) = J_x(0,x) + t \D_x J_t(x).
\end{equation}
The relation between the components of the currents $J_\alpha$ and $Q_\alpha$ and the components of the energy-momentum tensor is (cf.~\eqref{eq:EMT-relation})
\begin{equation}
    T^\alpha{_\beta} = \begin{pmatrix}
        T^0{_0} & T^0{_1}\\
        T^1{_0} & T^1{_1}
    \end{pmatrix} = -\begin{pmatrix}
        J_t & J_x\\
        0 & -J_t
    \end{pmatrix}.
\end{equation} 
The on-shell Ward identity for Weyl symmetry is given in~\eqref{eq:WIs-for-Ts} with $\chi^i = 0$ and reads
\begin{equation}
    T^\alpha \tt_\alpha + \mathcal{T}^\alpha \qq_\alpha =  0,
\end{equation}
and, using~\eqref{eq:EMT-def}, this implies that
\begin{equation}
    T^\alpha{_\beta} = 0.
\end{equation}

The nonzero equal-time Poisson brackets between $X^i$ and $Y$ and their respective conjugate momenta $\Pi_i$ and $\Pi_0$ are
\begin{equation}
\label{eq:canonical-brackets}
    \{X^\mu(t,x),\Pi_\nu(t,y)\} = \delta(x-y)\delta^\mu_\nu.
\end{equation}
Using~\eqref{eq:canonical-brackets}, we find that the brackets between the currents take the form
\begin{equation}
    \begin{aligned}
        \{J_x(t,x),J_x(t,y)\} &= 2J_x(t,x)\delta'(x-y) + 2\D_x J_x(t,x) \delta(x-y),\\
        \{J_x(t,x),J_t(x,y)\} &= 2J_t(t,x)\delta'(x-y) + 2\D_x J_t(t,x) \delta(x-y),\\
        \{J_t(t,x),J_t(t,x)\} &= 0.
    \end{aligned}
\end{equation}
where the prime denotes differentiation with respect to the argument, and where we used the identity
\begin{equation}
    F(t,y)\delta'(x-y) = \D_xF(t,x)\delta(x-y) + F(t,x)\delta'(x-y).
\end{equation}
We write the mode expansions of the currents in~\eqref{eq:current-fields} as
\begin{equation}
    \begin{aligned}
        J_t &= \sum_{n\in \ZZ} \MM_n e^{inx} ,\\
        J_x &= \sum_{n\in\ZZ}\LL_n(t) e^{inx}.
    \end{aligned}
\end{equation}
Now, given a Fourier series 
\begin{equation}
    f(x) = \sum_{n\in \ZZ} f_n e^{inx},
\end{equation}
the coefficient $f_m$ in front of $e^{imx}$ is given by
\begin{equation}
    f_m = \frac{1}{2\pi}\int_0^{2\pi} \dd x\, f(x) e^{-imx},
\end{equation}
which follows from the identity
\begin{equation}
\label{eq:mode-extraction}
    \frac{1}{2\pi}\int_0^{2\pi} \dd x\,e^{i(n-m)x} = \delta_{nm}.
\end{equation}
Using this, we find that the modes satisfy the algebra 
\begin{equation}
    \begin{aligned}
        \{ \LL_n(t),\LL_m(t) \} &= \frac{i}{2\pi}(n-m)\LL_{m+n}(t),\\
        \{ \LL_n(t),\MM_m \} &= \frac{i}{2\pi}(n-m)\MM_{m+n},\\
        \{ \MM_n,\MM_m \} &= 0.
    \end{aligned}
\end{equation}
which we recognise as the BMS$_3$ algebra~\eqref{eq:bms3-alg} with the identifications 
\begin{equation}
\label{eq:bms-modes}
    \LL_n(t) = \tfrac{i}{2\pi}L_n\,,\qquad \MM_n = M_n.
\end{equation}
We can express the embedding fields and their associated momenta as Fourier series by writing
\begin{subequations}
\label{eq:mode-exps}
    \begin{align}
        X^i &= \sum_{n\in \ZZ}x^i_n e^{inx},\\
        X^0 &= \sum_{n\in\ZZ}\big(x^0_n + t(\pi_0)_n\big)e^{inx},\\
        \Pi_i &=\sum_{n\in\ZZ}\big((\pi_i)_n + tn^2\tau^2\delta_{ij}x^j_n\big)e^{inx},\label{eq:Pi-i}\\
        \Pi_0 &=\sum_{n\in\ZZ} (\pi_0)_ne^{inx},
    \end{align}
\end{subequations}
where we used~\eqref{eq:transverse-momentum} in~\eqref{eq:Pi-i}. For the modes, which may be extracted from~\eqref{eq:mode-exps} using~\eqref{eq:mode-extraction}, the Poisson brackets~\eqref{eq:canonical-brackets} lead to
\begin{equation}
    \label{eq:mode-brackets}
    \{ x^\mu_n,(\pi_\nu)_m \} = \tfrac{1}{2\pi}\delta^\mu_\nu\delta_{m+n,0}.
\end{equation}
Next, we express the modes $L_n$ and $M_n$ defined in~\eqref{eq:bms-modes} in terms of the modes that appear in~\eqref{eq:mode-exps}. Due to~\eqref{eq:Jx-Jt-rel}, we only need $L_n(0)$, leading to
\begin{equation}
\begin{aligned}
    L_n(0) &= 2\pi \sum_{k\in\ZZ} k x^\mu_k(\pi_\mu)_{n-k},\\
    M_n &= -\tau^2 \sum_{k\in \ZZ} k(n-k) x^i_k x^i_{n-k} + \frac{1}{2}\sum_{k\in \ZZ} (\pi_0)_k (\pi_0)_{n-k},\\
    L_n(t) &= L_n(0) + 2\pi t n M_n.
\end{aligned}
\end{equation}
We then find that
\begin{equation}
\begin{aligned}
    \{L_n(t),M_m \} &= \{L_n(0),M_m\} = (n-m)M_{m+n},\\
    \{L_n(0),L_m(0) \} &= (n-m) L_{m+n}(0),\\
    \{ L_n(t),L_m(t) \} &= \{L_n(0),L_m(0) \} + 2\pi t\big( n\{ M_n,L_m(0) \} + m \{L_n(0),M_n\} \big)\\
    &= (n-m) L_{m+n}(t),
\end{aligned}
\end{equation}
once more recovering the (centreless, extended) BMS$_3$ algebra~\eqref{eq:bms3-alg}.

It is well known (see, e.g., \cite[Theorem~2.3]{GaoJiangPei}) that
this algebra admits a two-parameter family of central extensions (up
to equivalence). Indeed, equivalence classes of central extensions of
a Lie algebra $\g$ are classified by the second Chevalley--Eilenberg
cohomology $\sH^2(\g)$ with values in the ground field thought of as a
trivial representation. For $\g$ the centreless extended BMS$_3$
algebra, one has $\dim\sH^2(\g) = 2$ with representative $2$-cocycles
$\gamma_1,\gamma_2$ with nonzero components \begin{equation}
  \gamma_1(L_n,L_m) = \tfrac1{12} n(n^2-1) \delta_{n+m,0}
  \qquad\text{and}\qquad
  \gamma_2(L_n,M_m) = \tfrac1{12} n(n^2-1) \delta_{n+m,0}.
\end{equation}

The maximally centrally extended BMS$_3$ algebra may be expressed in
terms of operator product expansions.  To this end we introduce a
formal complex variable $z$ and write down the following generating
functions for the modes $L_n(0)$, $M_n$:
\begin{equation}
  T(z) = \sum_{n \in \ZZ} L_n(0) z^{-n-2} \qquad\text{and}\qquad M(z)
  = \sum_{n \in \ZZ} M_n z^{-n-2}.
\end{equation}
The mode algebra is then equivalent to the following operator product
expansions
\begin{equation}\label{eq:bms-as-opes}
  \begin{aligned}
    T(z) T(w) &= \frac{\tfrac12 c_L \1}{(z-w)^4} + \frac{2 T(w)}{(z-w)^2} + \frac{\d T(w)}{z-w} + \reg,\\
    T(z) M(w) &= \frac{\tfrac12 c_M \1}{(z-w)^4} + \frac{2 M(w)}{(z-w)^2} + \frac{\d M(w)}{z-w} + \reg,\\
    M(z) M(w) &= \reg,
  \end{aligned}
\end{equation}
where $\1$ is the identity in the resulting vertex operator algebra.
In choosing to work with operator products instead of modes we are
actually tacitly making a choice of vacuum and hence of a normal
ordering prescription.  Here we are working with the ``flipped''
vacuum in the language of~\cite{Bagchi:2020fpr}.  When BMS$_3$ is
defined as a limit of two copies of the Virasoro algebra, there is not
a unique choice of vacuum (see, e.g., \cite{Bagchi:2020fpr}) and it
would be interesting to see how the spectrum of the Carroll string
depends on this choice.

We remark again that $z$ is a formal variable and not a coordinate on the
worldsheet.  Formally we can think of restricting the fields $J_t$ and
$J_x$ to the circle $t=0$ on the cylinder and then mapping it to the
unit circle on the punctured complex plane.  The Fourier series
of $J_x$ and $J_t$ extend meromorphically from the unit circle to the
fields $T(z)$ and $M(z)$ on the punctured complex plane.

What about the $t$-dependence of the modes $L_n(t)$?  This can be
understood as the result of a one-parameter family of automorphisms of
the operator product algebra: namely,
\begin{equation}
  T(z) \mapsto T(z) - 2 \pi t (z \d M(z) + 2 M(z)) \qquad\text{and}\qquad
  M(z) \mapsto M(z).
\end{equation}

The Carroll string gives a realisation of this algebra in terms of
$D+1$ bosonic BC systems of weights $(0,1)$: $(X^\mu,\Pi_\mu)$, whose
defining operator product expansions are
\begin{equation}
  \Pi_\mu(z) X^\nu(w) = \frac{\delta^\mu_\nu \1}{z-w} + \reg.
\end{equation}
These correspond to the mode algebra
\begin{equation}\label{eq:mode-algebra-can}
  [(\pi_\mu)_m, x^\nu_n] = \delta^\nu_\mu \delta_{n+m,0},
\end{equation}
where the modes are defined by
\begin{equation}
    X^\mu(z) = \sum_{n\in\ZZ} x^\mu_n z^{-n} \qquad\text{and}\qquad \Pi_\mu(z) = \sum_{n\in\ZZ} (\pi_\mu)_n z^{-n-1}.
\end{equation}
We then have that
\begin{equation}
  T(z) = \d X^\mu \Pi_\mu \qquad\text{and}\qquad M(z) =
  \tfrac12 \Pi_0 \Pi_0 - \tfrac12 \tau^2 \delta_{ij} \d X^i \d X^j,
\end{equation}
where all products are normal-ordered: $(AB)(w)$ is the limit as $z\to
w$ of the \emph{regular} part of the operator product expansion
$A(z)B(w)$.  The fields $T(z)$ and $M(z)$ just defined give a
realisation of the centrally extended BMS$_3$
algebra~\eqref{eq:bms-as-opes} with $c_L = 2(D+1)$ and $c_M=0$. In
particular, for the $D+1$ bosonic BC systems of weights $0$ and $1$,
which are, respectively, $X^\mu = (X^0,X^i)$ and $\Pi_\mu = (\Pi_0,\Pi_i)$, we have the following OPEs
\begin{equation}
    \begin{aligned}
        T(z)X^\mu(w) &= \frac{\D X^\mu(w)}{z-w} + \reg,\qquad T(z)\Pi_\mu(w) = \frac{\Pi_\mu(w)}{(z-w)^2} + \frac{\D \Pi_\mu(w)}{z-w} + \reg,\\
        M(z)X^0(w) &= \frac{\Pi_0(w)}{z-w} + \reg,\qquad M(z)X^i(w) = \reg,\\
        M(z)\Pi_0(w) &= \reg,\qquad M(z)\Pi_i(w) = - \frac{\tau^2 \D X^i(w)}{(z-w)^2} - \frac{\tau^2 \D^2 X^i(w)}{z-w} + \reg,
    \end{aligned}
\end{equation}
where the first line just says that $X^\mu$ are $\Pi_0$ are conformal
primaries with conformal weights $(0,1)$.  We will assume that $\tau
\neq 0$ and therefore we may set $\tau = 1$ via an automorphism of the
operator product expansion: $X^i \mapsto \frac1\tau X^i$ and $\Pi_i
\mapsto \tau \Pi_i$.  We choose to keep $\tau$ explicit for
pedagogical reasons.

The above realisation is reminiscent of the flat space ambitwistor
string\cite{Mason:2013sva}, which also realises the BMS$_3$ algebra
with $c_L = 2(D+1)$ and $c_M=0$. Using similar notation, the
ambitwistor string realisation consists of $D+1$ bosonic BC systems
$(X^\mu,\Pi_\mu)$ of weights $(0,1)$ in terms of which
\begin{equation}
\label{eq:ambi-twistor}
  T(z) = \d X^\mu \Pi_\mu \qquad\text{and}\qquad M(z) = -\tfrac12
  \eta^{\mu\nu} \Pi_\mu \Pi_\nu,
\end{equation}
where $\eta^{\mu\nu}$ is a lorentzian inner product.  The theories are
different: they have different global symmetries: Lorentz in the case
of the ambitwistor string and transverse rotations in the present
case.  In fact, the physical state condition in the ambitwistor string
gives a relativistic equation, for the Carroll string, we expect
Carroll-invariant field equations.

\section{BRST quantisation}
\label{sec:brst-quantisation}

In this section we set up the BRST quantisation of the carrollian
string.

\subsection{Brief recap of BRST quantisation}
\label{sec:brief-recap-brst}

BRST quantisation is a covariant quantisation scheme for theories
whose hamiltonian description has first-class constraints, such as
gauge theories and, especially, string theories.  In symplectic
geometry, first-class constraints are those defining coisotropic
submanifolds.

Notwithstanding the fact that the true phase space of such a theory is
the coisotropic reduction of the original phase space by the
first-class constraints, it is typically difficult to solve explicitly
the constraints and, even if one were able to, this often breaks
global symmetries which one would prefer to keep manifest. The idea
behind BRST quantisation is to implement the coisotropic reduction
\emph{after} quantisation, albeit of a slightly larger theory.  That
theory is obtained by the addition of ghost fields, say $(b_i, c^i)$,
with Poisson bracket $\{b_i,c^j\} = \delta_i^j$.  Before quantisation,
this large theory possesses a classical BRST differential: a function
$Q = c^i \varphi_i + \cdots$ in the larger phase space, where
$\varphi_i$ are the first-class constraints.  The associated hamiltonian
vector field to $Q$ is homological; that is,
$\{Q,\{Q,-\}\} = 0$ where $\{-,-\}$ is the Poisson bracket in the
larger theory.  This typically determines the form of the omitted
terms in the above expression for $Q$: e.g., if the constraints are
the components of the moment map due to a hamiltonian action of a Lie
group, so that $\{\varphi_i, \varphi_j\} = f_{ij}{}^k \varphi_k$, with
$f_{ij}{}^k$ the structure constants of the Lie algebra relative to an
inconsequential choice of basis, then
\begin{equation}
  Q = c^i \varphi_i - \tfrac12 f_{ij}{}^k c^i c^j b_k,
\end{equation}
which obeys $\{Q,Q\} = 0$.  This fact is simply a reflection that on
the larger phase space we again have a hamiltonian action of the Lie
group with the ghosts transforming under the adjoint ($c^i$) and
coadjoint ($b_i$) representations.  The hamiltonian action of the Lie
algebra is given by the Poisson bracket with $\{Q,b_i\} = \varphi_i +
f_{ij}{}^k b_k c^j$.

BRST quantisation consists of quantising the larger phase space
(including the ghosts) and constructing a quantum BRST operator $\hat
Q$, whose classical limit is $Q$, and which now obeys $\hat Q^2 = 0$.
The existence of $\hat Q$ is not guaranteed and may constrain further
some of the parameters of the theory.  Assuming that $\hat Q$ exists, the
physical states of the theory are to be identified with its
cohomology: the kernel (cocycles) modulo the image (coboundaries).
The cohomology is typically graded by ghost number, with $b_i, c^i,
\varphi_i$ having ghost numbers $-1,1,0$, respectively.

In string theory, the gauge symmetries depend on the geometry of the
worldsheet: for example, they are given by the Virasoro algebra (for
open bosonic strings) or two copies of the Virasoro algebra (for
closed bosonic strings). For the Carroll string in this paper, the Lie
algebra of the gauge symmetry is the BMS$_3$ algebra, as is the case
also for the ambitwistor string. These are infinite-dimensional Lie
algebras which are $\ZZ$-graded, yet are finite-dimensional in each
degree. For such Lie algebras, the quantum BRST cohomology can be
re-interpreted \cite{MR865483} as the semi-infinite cohomology (in the
sense of Feigin \cite{MR740035}) relative to the centre. This latter
condition constrains the modules in which the cohomology takes values
and, in particular, determines the values of the central charges
which, in the context of string theory, often determines the critical
dimension; although the precise correspondence depends on the choices
of matter modules.

\subsection{The BRST complex}
\label{sec:brst-complex}

As we saw in Section~\ref{sec:carroll-string}, the quantum carrollian
string can be described by a meromorphic conformal field theory with
fields $X^\mu(z)$ and $\Pi_\mu(z)$, with $\mu \in \{0,\ldots,D\}$, of
conformal weights $0$ and $1$ respectively. The gauge symmetry algebra
is the extended BMS${}_3$ algebra with generators $T = \d X^\mu
\Pi_\mu$ and $M = \frac12 (\Pi_0)^2 - \tfrac12 \tau^2\delta_{ij} \d
X^i \d X^j$. These give a realisation of BMS${}_3$ with central
charges $c_L = 2 (D+1)$ and $c_M =0$. As shown in
\cite{Figueroa-OFarrill:2024wgs}, the BRST differential squares to
zero for the critical central charge $c_L = 52$, so that $D=25$ as for
the critical (relativistic) bosonic string. We shall assume that we
are at the critical central charge, so that the index $\mu$ runs from
$0$ to $25$, and the indices $i,j$ run from $1$ to $25$. The BRST
operator is the zero mode of the BRST current \cite{Figueroa-OFarrill:2024wgs}
\begin{equation}
\label{eq:brst-current}
  J = c T + C M + \tfrac12 c \Tgh + \tfrac12 C \Mgh,
\end{equation}
where we have introduced fermionic ghost systems $(b,c)$ and $(B,C)$ of
weights $(2,-1)$ with the operator product expansions
\begin{equation}
    b(z)c(w) = \frac{\1}{z-w} + \reg \qquad\text{and}\qquad B(z)C(w) = \frac{\1}{z-w} + \reg,
\end{equation}
in terms of which
\begin{equation}
  \Tgh = -2 b \d c - \d b c - 2 B \d C - \d B C
  \qquad\text{and}\qquad \Mgh = - 2 B \d c - \d B c,
\end{equation}
where again all products are normally ordered and associate to the
left, so that $ABC := A(BC)$, et cetera. 

Let us expand our fields as follows:
\begin{equation}\label{eq:mode-expansions}
  \begin{aligned}
    X^\mu(z) &= \sum_{n\in\ZZ} x^\mu_n z^{-n}\\
    b(z) &= \sum_{n\in\ZZ} b_n z^{-n-2}\\
    B(z) &= \sum_{n\in\ZZ} B_n z^{-n-2}
  \end{aligned}
  \qquad\qquad
  \begin{aligned}
    \Pi_\mu(z) &= \sum_{n\in\ZZ} (\pi_\mu)_n z^{-n-1}\\
    c(z) &= \sum_{n\in\ZZ} c_n z^{-n+1}\\
    C(z) &= \sum_{n\in\ZZ} C_n z^{-n+1}.
  \end{aligned}
\end{equation}
The above mode expansions allow us to recover the algebra of modes from the
operator product expansion via the formula
\begin{equation}\label{eq:modes-from-opes}
  [A_n,B_m] =\sum_{\ell \geq 1} \binom{n+h_A-1}{\ell -1} \left(
    [A,B]_\ell \right)_{m+n}, 
\end{equation}
with $\binom{n}{k} = \frac{n!}{k! (n-k)!}$ the usual binomial
coefficient and $A$ and $B$ are fields and $A$ has conformal weight
$h_A$. In writing the above, we used that the operator product expansion
of two fields $A(z),B(z)$ with conformal weights $h_A,h_B$ defines a
family of bilinear brackets $[-,-]_\ell$ via
\begin{equation}
    A(z)B(w) = \sum_{\ell\leq N}\frac{[A,B]_\ell(\omega)}{(z-w)^\ell},
\end{equation}
where $N\in \NN$ is a finite number usually given by the sum of the conformal 
weights of $A$ and $B$, while $A(z) = \sum_{n\in \ZZ} A_n z^{-n-h_A}$ and 
$B(z) = \sum_{n\in \ZZ} B_n z^{-n-h_B}$. In this way we arrive at the 
following nonzero (anti)commutators of modes: 
\begin{equation}
\label{eq:mode-brackets-opes}
    [(\pi_\mu)_m, x^\nu_n] = \delta_\mu^\nu \delta_{n+m,0}, \qquad
    [b_m, c_n] = \delta_{m+n,0} \qquad\text{and}\qquad [B_m, C_n] = \delta_{m+n,0}.
\end{equation}
The state-field correspondence says that a field $\phi(z)$ creates a
state by acting on the projectively invariant vacuum in the limit as
$z \to 0$. This vacuum is sometimes known as the ``flipped''
vacuum~\cite{Bagchi:2020fpr}.  The mode expansions then reveal which
modes annihilate the vacuum.  The modes which do not are the creation
modes.  For the fields in the theory, the creation modes are the
following:
\begin{equation}\label{eq:creation-modes}
  x^\mu_{n\leq 0}, \qquad (\pi_\mu)_{n\leq -1}, \qquad b_{n\leq -2},
  \qquad B_{n\leq -2}, \qquad c_{n\leq 1} \qquad\text{and}\qquad C_{n\leq 1 }.
\end{equation}
For example, letting $\ket{0}$ denote the vacuum, the creation modes of, say,
$B(z)$ are obtained in the following manner: all the terms in 
$B(z)\ket{0} = \sum_{n\in \ZZ}b_n z^{-n-2}\ket{0}$ with $-n-2 < 0 \iff
n\geq -1$ diverge as $z\to 0$, revealing those modes to be the
annihilation modes.  The remaining modes $b_{n\leq -2}$ are the
creation modes.

Let now $d$ denote the BRST differential: it is the zero mode $J_0$ of the
BRST current $J$ in equation~\eqref{eq:brst-current} or, equivalently,
the endomorphism resulting from the first order pole with the BRST
current $d=[J,-]_1$.  It follows that~\cite{Figueroa-OFarrill:2024wgs}
\begin{equation}
    db = \Ttot = T + \Tgh ,\qquad  dB = \Mtot = M + \Mgh.
\end{equation}

The zero mode $\Ltot_0$ of the total Virasoro element acts diagonally
on the modes of the fields in the theory:
$[\Ltot_0,\phi_n] = - n \phi_n$ for any field $\phi(z)$.  Since
$\Ltot_0 = [d,b_0]$ is BRST exact, the BRST complex decomposes into a
direct sum of subcomplexes corresponding to the eigenspaces of
$\Ltot_0$ with different eigenvalues.  A standard argument shows that
any BRST cocycle of non-zero conformal weight relative to $\Ltot_0$ is
also a coboundary, so that there is no cohomology except in the
subcomplex $\ker \Ltot_0$ corresponding to zero total conformal
weight.\footnote{This argument does not work with $\Mtot_0$ since
  $\Mtot_0$ is nilpotent and its action filters the complex instead of
  grading it.  Alas the ensuing filtration is increasing and it does
  not seem to lead to a useful spectral sequence.}  Indeed, if $\psi$
is a cocycle satisfying $\Ltot_0 \psi = h \psi$ with $h\neq 0$, then
$\psi = d\left( \frac1h b_0 \psi \right)$ is also a coboundary.
Furthermore, as we have learned from the BRST quantisation of
(particularly, noncritical) strings, it is actually convenient to
restrict further to the relative subcomplex
$\ker b_0 \cap \ker \Ltot_0$ and then to bootstrap the full cohomology
from that of this subcomplex.

Just as the original BRST complex computes the semi-infinite
cohomology of the BMS algebra relative its centre, the subcomplex
$\ker b_0 \cap \ker \Ltot_0$ computes the semi-infinite cohomology of
the BMS algebra relative to the subalgebra spanned by the centre and
the Virasoro zero mode.  The reason we restrict to this relative
subcomplex is that it is a natural way to break up the calculation.
Indeed, it is easier to calculate the relative cohomology and then by
a mixture of homological algebra and some more detailed calculation we
may extract the absolute cohomology (see
Section~\ref{sec:from-relat-absol}).

Let $\Crel^\bullet = \ker b_0 \cap \ker \Ltot_0$ be the relative
subcomplex, graded by ghost number, with $c, C$ having ghost number
$+1$, $b, B$ ghost number $-1$ and $X^\mu,\Pi_\mu$ ghost number
$0$.  Since in the BRST current the fields $X^\mu$ only appear
differentiated (i.e., $\d X^\mu$) the zero modes $x^\mu_0$ do not
enter in the calculations.  This suggests working in momentum space,
resulting in a differential complex involving no derivatives.  To that
end, we define
\begin{equation}
  \vac := \lim_{z\to 0} \exp(i p_\mu X^\mu(z)) \ket{0} = \exp(i p_\mu x^\mu_0) \ket{0},
\end{equation}
where now $(\pi_\mu)_0 \vac = i p_\mu \vac$, and we denote by
$\sC^\bullet(p)$ and $\Crel^\bullet(p)$ the absolute and relative
complexes at momentum $p$.  More concretely, $\sC^\bullet(p)$ is the
kernel of $\Ltot_0$ acting on the states obtained from $\vac$ by the
action of the following creation operators:
\begin{equation}
  \label{eq:creators}
  x^\mu_{n\leq -1}, \quad (\pi_\mu)_{n\leq -1}, \quad b_{n\leq -2},
  \quad B_{n\leq -2}, \quad c_{n\leq 1} \quad\text{and}\quad C_{n\leq 1},
\end{equation}
where we have omitted $(\pi_\mu)_0$ since it acts by scalar
multiplication by $i p_\mu$. Since $d$ has no $x^\mu_0$, it follows
that $[d,(\pi_\mu)_0] = 0$ and hence $\sC^\bullet(p)$ is indeed a
subcomplex. The relative subcomplex $\Crel^\bullet(p)$ is the graded
subspace $\ker b_0 \subset \sC^\bullet(p)$, which is a subcomplex
since $[d,b_0] = \Ltot_0$ and this acts trivially on $\sC^\bullet(p)$.
It consists of those states in the kernel of $\Ltot_0$ acting on the
states obtained from $\vac$ by the action of the creation modes in
equation~\eqref{eq:creators} \emph{with the exception of} $c_0$.

The relative and absolute complexes are related by a short exact
sequence of differential complexes.  Indeed, $\Crel^n(p) \subset
\sC^n(p)$ as the kernel of $b_0$ and the image of $b_0 : \sC^n(p) \to
\sC^{n-1}(p)$ actually lands in $\Crel^{n-1}(p)$ since $b_0^2 = 0$,
resulting in
\begin{equation}
  \label{eq:ses-brst}
  \begin{tikzcd}
    0 \arrow[r] & \Crel^n(p) \arrow[r] & \sC^n(p) \arrow[r,"b_0"] &
    \Crel^{n-1}(p) \arrow[r] \arrow[l, bend left=45, "c_0" pos=0.45] & 0,
  \end{tikzcd}
\end{equation}
where we have indicated by $c_0 : \Crel^{n-1}(p) \to \sC^n(p)$ the
splitting of the sequence. Indeed for every $\Psi \in \Crel^{n-1}(p)$,
$b_0 c_0 \Psi = [b_0,c_0]\Psi = \Psi$. By the Snake Lemma (see, e.g.,
\cite[Ch.III, §9]{MR1878556}) of homological algebra, the short exact
sequence~\eqref{eq:ses-brst} induces a long exact sequence relating
the relative and absolute cohomologies, which we will study further in
Section~\ref{sec:from-relat-absol}.

\subsection{The relative subcomplex}
\label{sec:relative-subcomplex}

It is convenient to decompose the relative subcomplex as
\begin{equation}
  \Crel^n(p) = \fM^n(p) \oplus C_0 \fM^{n-1}(p),
\end{equation}
where $\fM^\bullet(p)$ is the subspace obtained from $\vac$ by the
action of the creation operators in equation~\eqref{eq:creators} with
the exception of both $c_0$ and $C_0$.  Since those creation operators
all have positive conformal weight, except for $c_1$ and $C_1$, we see
that $\fM^n(p) \neq 0$ only for $n=0,1,2,3,4$.  It is easy to write
down a basis for $\fM^\bullet(p)$, where we list only those operators
which are meant to act on $\vac$:
\begin{equation}
\label{eq:M-modes}
  \begin{aligned}
    \fM^0 &: \1\\
    \fM^1 &: c_1 C_1 b_{-2}, c_1 C_1 B_{-2}, c_1 x^\mu_{-1}, C_1 x^\mu_{-1}, c_1 (\pi_\mu)_{-1}, C_1 (\pi_\mu)_{-1}\\
    \fM^2 &: c_1 c_{-1}, c_1 C_{-1}, C_1 c_{-1}, C_1 C_{-1}, c_1 C_1 x^\mu_{-2}, c_1 C_1 (\pi_\mu)_{-2}, c_1 C_1 x^\mu_{-1} (\pi_\nu)_{-1},\\
    & \quad  c_1 C_1 x^\mu_{-1} x^\nu_{-1}, c_1 C_1 (\pi_\mu)_{-1}(\pi_\nu)_{-1}\\
    \fM^3 &: c_1 C_1 c_{-2}, c_1 C_1 C_{-2}, c_1 C_1 c_{-1} x^\mu_{-1}, c_1 C_1 C_{-1} x^\mu_{-1}, c_1 C_1 c_{-1}(\pi_\mu)_{-1}, c_1 C_1 C_{-1} (\pi_\mu)_{-1} \\ 
    \fM^4 &: c_1 C_1 c_{-1} C_{-1}.
  \end{aligned}
\end{equation}

Every $\Psi \in \fM^\bullet(p)$ is of the form $\Psi = \exp(i p_\mu
x^\mu_0) \psi$, with $\psi \in \fM^\bullet(0)$.  Therefore to
calculate the action of the differential on the relative subcomplex,
it is enough to calculate the action on $\fM^\bullet(0)$ and then use
the following two lemmas.

\begin{lemma}\label{lem:dC0}
  Acting on the relative subcomplex,
  \begin{equation}
  \label{eq:[d,C_0]}
    [d,C_0] = 4 c_2 C_{-2} + 2 c_1 C_{-1} + 2 C_{1} c_{-1} + 4 C_2 c_{-2}.
  \end{equation}
\end{lemma}

\begin{proof}
 Since $d = J_0$, we may use equation~\eqref{eq:modes-from-opes} in
 order to reduce $[d,C_0]$ to operator product expansions:
  \begin{equation*}
    [d,C_0] = [J_0, C_0] = \sum_{\ell \geq 1} \binom{0}{\ell -1} \left(
      [J,C]_\ell \right)_0 = ([J,C]_1)_0.
  \end{equation*}
  A calculation (using the Mathematica package \texttt{OPEdefs}
  written by Kris Thielemans
  \cite{Thielemans:1991uw,Thielemans:1992mu,Thielemans:1994er}) gives
  \begin{equation*}
    [J,C]_1 = c \d C - \d c C
  \end{equation*}
 and hence
  \begin{equation*}
    [d,C_0]=\left( c \d C - \d c C \right)_0.
  \end{equation*}
  We now use the mode expansions~\eqref{eq:mode-expansions}, from
  where we see that
  \begin{equation*}
    \d c = \sum_{n\in\ZZ} (1-n) c_n z^{-n} \qquad\text{and}\qquad
    \d C = \sum_{n\in\ZZ} (1-n) C_n z^{-n}.
  \end{equation*}
  We now calculate, using that the modes of $c$ and $C$ anticommute,
  so that normal ordering is automatic:
  \begin{equation*}
    c\d C = \sum_{n,m\in\ZZ} (1-m) c_n C_m z^{-n-m-1} = \sum_{\ell \in \ZZ} \left(  \sum_{m\in\ZZ} (1-m) c_{\ell -m} C_m \right) z^{-\ell +1},
  \end{equation*}
  whose zero mode is
  \begin{equation*}
    \left( c\d C \right)_0 = \sum_{m\in\ZZ} (1-m) c_{-m} C_m.
  \end{equation*}
  Acting on the relative subcomplex, we need only keep $-2\leq m \leq
  2$ in the sum since any other mode will act trivially. Even though $c_2$ and $C_2$ are annihilation operators, we must keep them around since some states in $\fM^1$ involve $b_{-2}$ and $B_{-2}$ and thus give nontrivial contributions in~\eqref{eq:big-Psi}.  Therefore,
  on the relative subcomplex,
  \begin{equation*}
    \left( c\d C \right)_0 =3 c_2 C_{-2} + 2 c_1 C_{-1} + c_0 C_0 +
    C_2 c_{-2},
  \end{equation*}
  and hence, exchanging $c \leftrightarrow C$,
  \begin{equation*}
    \left( C \d c \right)_0 = 3 C_2 c_{-2} + 2 C_1 c_{-1} + C_0 c_0 +
    c_2 C_{-2}.
  \end{equation*}
  Finally,
  \begin{equation*}
    \begin{aligned}
      \left( c \d C - \d c C \right)_0 &= \left( c \d C + C \d c \right)_0\\
      &= 3 c_2 C_{-2} + 2 c_1 C_{-1} + c_0 C_0 + C_2 c_{-2} + 3 C_2  c_{-2} + 2 C_1 c_{-1} + C_0 c_0 + c_2  C_{-2}\\
      &= 4 c_2 C_{-2} + 2 c_1 C_{-1} + 2 C_1 c_{-1} + 4 C_2 c_{-2}.
    \end{aligned}
  \end{equation*}
  (We observe that, as expected, the $c_0$-terms are absent, which
  provides a check on the calculation.)
\end{proof}

\begin{lemma}\label{lem:df}
  Let $\fe := \exp(i p_\mu x^\mu_0)$.  Then acting on the relative subcomplex
  \begin{equation}
  \label{eq:desired-expression}
    \begin{aligned}\relax
      [d,\fe] &= i p_\mu \fe \left( - 2 c_{-2} x^\mu_2 - c_{-1} x^\mu_1 +
        c_1 x^\mu_{-1} +  2 c_2 x^\mu_{-2} \right) - \tfrac12 p_0^2 \fe C_0\\
      & \qquad {} + i p_0 \fe \left( C_{-2} (\pi_0)_2 + C_{-1} (\pi_0)_1
        + C_0 (\pi_0)_0 + C_1 (\pi_0)_{-1} + C_2 (\pi_0)_{-2} \right).
      \end{aligned}
  \end{equation}
\end{lemma}

\begin{proof}
  We start by calculating $[L_n,\fe]$ and $[M_n,\fe]$.  The only
  modes which do not commute with $\fe$ are $(\pi_\mu)_0$ and from
  the mode algebra \eqref{eq:mode-algebra-can} it follows that these
  act on $\fe$ by scalar multiplication with $ip_\mu$.  Hence we need
  only keep those terms involving $(\pi_\mu)_0$ in $L_n$ and $M_n$.
  It then follows that
  \begin{equation*}
    [M_n, \fe] = [\tfrac12 (\Pi_0^2)_n, \fe] = \tfrac12 \sum_{m \in \ZZ} [(\pi_0)_{n-m} (\pi_0)_m, \fe].
  \end{equation*}
  This breaks up into two cases, depending on whether $n=0$ or $n\neq
  0$:
  \begin{equation*}
    [M_0, \fe] = \tfrac12 [(\pi_0)_0^2, \fe] = - \tfrac12 p^2_0 \fe + i p_0 \fe (\pi_0)_0
  \end{equation*}
  and if $n\neq 0$,
  \begin{equation*}
    [M_n, \fe] =  [(\pi_0)_0 (\pi_0)_n, \fe] = ip_0 \fe (\pi_0)_n.
  \end{equation*}
  Let us now calculate $[L_n,\fe]$.  From the formula for the
  normal-ordered product,
  \begin{equation*}
    L_n = (\d X^\mu \Pi_\mu)_n = \sum_{\ell \leq -1} (\d X^\mu)_\ell
    (\pi_\mu)_{n-\ell} + \sum_{\ell \geq 0} (\pi_\mu)_{n-\ell} (\d X^\mu)_\ell.
  \end{equation*}
  Again we need only keep the terms involving $(\pi_\mu)_0$ in order
  to calculate $[L_n, \fe]$.  Doing so, we find
  \begin{equation*}
    [L_n, \fe] = - i n p_\mu x^\mu_n \fe.
  \end{equation*}

  It follows that
  \begin{equation*}
    [d,\fe] = [J_0, \fe] = [(cT)_0 + (CM)_0, \fe] =
    \sum_{n\in\ZZ} [c_{-n} L_n + C_{-n} M_n, \fe].
  \end{equation*}
  By the same reasoning as above, when acting on the relative subcomplex we need only keep $-2\leq n \leq
  2$ in the sum, and using the above expressions for $[L_n, \fe]$
  and $[M_n, \fe]$ we find
  \begin{equation*}
    [d,\fe] = - \sum_{-2 \leq n \leq 2} i n p_\mu \fe c_{-n} x^\mu_n -
    \tfrac12 p^2_0 \fe C_0 + \sum_{-2\leq n \leq 2} i p_0 \fe  C_{-n} (\pi_0)_n ,
  \end{equation*}
  which when expanded gives the desired expression~\eqref{eq:desired-expression}.
\end{proof}

For calculations it is also convenient to exploit the state-field
correspondence and write the states in $\fM^\bullet(0)$ in terms of
the fields which produce them acting on the vacuum $\ket{0}$,
resulting in the following field basis for $\fM^\bullet(0)$:
\begin{equation}
\label{eq:M-fields}
  \begin{aligned}
    \fM^0 &: \1\\
    \fM^1 &: c C b , c C B, c \d X^\mu, C \d X^\mu, c \Pi_\mu , C \Pi_\mu\\
    \fM^2 &: c \d^2 c, c \d^2 C, C \d^2 c, C \d^2 C, c C \d^2 X^\mu, c C \d \Pi_\mu, c C \d X^\mu \Pi_\nu, c C \d X^\mu \d X^\nu , c C \Pi_\mu \Pi_\nu \\
    \fM^3 &: c C \d^3 c, c C \d^3 C, c C \d^2 c \d X^\mu, c C \d^2 C \d X^\mu, c C \d^2 c \Pi_\mu, c C \d^2 C \Pi_\mu \\
    \fM^4 &: c C \d^2 c \d^2 C.
  \end{aligned}
\end{equation}
More precisely, the state-field correspondence says that in a reduced
normal-ordered expression we may make the following replacements:
\begin{equation}
  \label{eq:state-field-corr}
  \begin{aligned}
    c &\leftrightarrow c_1\\
    \d c &\leftrightarrow c_0\\
    \d^2 c &\leftrightarrow 2 c_{-1}\\
    \d^3 c &\leftrightarrow 6 c_{-2}\\
    b &\leftrightarrow b_{-2}
  \end{aligned}
  \qquad\qquad
  \begin{aligned}
    C &\leftrightarrow C_1\\
    \d C &\leftrightarrow C_0\\
    \d^2 C &\leftrightarrow 2 C_{-1}\\
    \d^3 C &\leftrightarrow 6 C_{-2}\\
    B &\leftrightarrow B_{-2}
  \end{aligned}
  \qquad\qquad
  \begin{aligned}
    \d X^\mu &\leftrightarrow x^\mu_{-1}\\
    \d^2 X^\mu &\leftrightarrow 2 x^\mu_{-2}\\
    \Pi_\mu &\leftrightarrow (\pi_\mu)_{-1}\\
    \d \Pi_\mu &\leftrightarrow (\pi_\mu)_{-2}.\\
  \end{aligned}
\end{equation}

Using these substitutions in the field basis~\eqref{eq:M-fields} leads
to the state basis~\eqref{eq:M-modes} and viceversa. The
relative subcomplex $\Crel^\bullet(p) = \bigoplus_n \Crel^n(p)$ is
such that $\Crel^n(p) \neq 0$ only for $0 \leq n \leq 5$.  Indeed, any
state in $\Crel^n(p)$ can be written as a linear combination of terms of the form
\begin{equation}
\label{eq:big-Psi}
  \Psi = \fe \psi_1 + \fe C_0 \psi_2,
\end{equation}
where $\psi_1 \in \fM^n(0)$ and $\psi_2\in \fM^{n-1}(0)$.  The action
of the BRST differential on such a state is given by
\begin{equation}
  d\Psi = [d,\fe] \psi_1 + \fe d\psi_1 + [d,\fe] C_0 \psi_2 +
  \fe [d,C_0] \psi_2 - \fe C_0 d \psi_2,
\end{equation}
which can be calculated using Lemmas~\ref{lem:dC0} and \ref{lem:df}
together with the action of $d$ on $\fM^\bullet(0)$.  This is easier
to do in the field basis, since we can then use the operator product
expansion with the BRST current and that can be simplified using the
Mathematica package \texttt{OPEdefs} written by Kris Thielemans
\cite{Thielemans:1991uw,Thielemans:1992mu,Thielemans:1994er}.  The
resulting expressions are recorded in
Appendix~\ref{sec:acti-brst-diff}.

\section{The relative BRST cohomology}
\label{sec:relat-brst-cohom}

Now we are poised to start computing the relative BRST cohomology
$\Hrel^\bullet(p)$.  We will use the notation $\Crel^n(p)$, $\Zrel^n(p)$ and
$\Brel^n(p)$ for the relative cochains, cocycles and coboundaries,
respectively, at momentum $p$.

In analysing the cocycle equations $d\Psi = 0$, we need to be careful
with the value of $p$.  To this end we exploit the covariance under
the group $\SO(25)$ of transverse rotations.  The action of $R \in
\SO(25)$ on $p = (p_0, \p)$ is given by $p \mapsto (p_0, R \p)$.  The
invariants of that action are $p_0$ and $\|\p\|$.   We therefore have
four cases to consider:
\begin{enumerate}
\item $p=0$
\item $p_0 \neq 0$, $\p = \bzero$
\item $p_0 = 0$, $\p \neq \bzero$
\item $p_0 \neq 0$, $\p \neq \bzero$, which is the generic case.
\end{enumerate}

Our first result is the following.

\begin{proposition}
  Let $p = (p_0,\p)$.  Then $\Hrel^\bullet(p) = 0$ unless $p_0 = 0$.
\end{proposition}

\begin{proof}
  It is a computation.\cite{Serre} Indeed, this is simply done by
  solving the linear equations with Mathematica and showing that the
  dimension of the space of cocycles equals that of the space of
  coboundaries.
\end{proof}

As a result we will put $p_0=0$ from now on and concentrate on two
cases: $p =0$ and $p=(0,\p)$ with $\p \neq \bzero$.

Our calculation of the relative cohomology is somewhat pedestrian. At
every ghost number we determine the space of cocycles as a
representation of the stabiliser subgroup of the momentum $p$ in
$\SO(25)$. When $p=0$, the stabiliser is all of $\SO(25)$ and if
$p=(0,\p)$ with $\p \neq \bzero$, the stabiliser is isomorphic to
$\SO(24)$, the rotations on the subspace $V^\perp$ perpendicular to
$\p$. We have seen from the explicit tabulation of the basis for
$\fM^\bullet$, that there are four isotypical (complex)
representations of $\SO(25)$: the scalars $\CC$, the vectors $V$, the
symmetric traceless tensors $\odot^2_0 V$ and the skewsymmetric
tensors $\ext{2} V$. These representations decompose as follows under
the $\SO(24)$ which stabilises $\p \neq \bzero$:
\begin{equation}
  \label{eq:branching}
  \begin{aligned}
    \CC &= \CC\\
    V &= \CC\p \oplus V^\perp \cong \CC \oplus V^\perp\\
    \odot^2_0 V &\cong \odot^2_0 V^\perp \oplus V^\perp \oplus \CC\\
    \ext{2} V &\cong \ext{2} V^\perp \oplus V^\perp.
  \end{aligned}
\end{equation}
Therefore if $p=0$, the relative cocycles at ghost number $n$ will be written
as
\begin{equation}
  \Zrel^n(0) \cong z_1 \CC \oplus z_2 V \oplus z_3 \odot^2_0 V \oplus z_4
  \ext{2} V,
\end{equation}
where $z_1,\dots,z_4$ are the multiplicities of the corresponding
representations.

Similarly, if $\p \neq \bzero$, the relative cocycles at ghost number
$n$ will be written as
\begin{equation}
  \Zrel^n(0,\p) \cong z_1 \CC \oplus z_2 V^\perp \oplus z_3 \odot^2_0 V^\perp \oplus z_4
  \ext{2} V^\perp,
\end{equation}
where $z_1,\dots,z_4$ are again the multiplicities of the
corresponding representations.

Since $\Crel^n(0)$ and $\Crel^n(0,\p)$ can be similarly decomposed,
determining $\Zrel^n(0)$ and $\Zrel^n(0,\p)$ automatically allows us
to determine $\Brel^{n+1}(0)$ and $\Brel^{n+1}(0,\p)$ via the Rank
Theorem.  In turn this allows us to determine the relative cohomology
as representations of $\SO(25)$ or $\SO(24)$.  Such representations
are to be interpreted as inducing representations for UIRs of the
Carroll group.  We will review this in
Section~\ref{sec:interpr-terms-carr}.

\subsection{Calculating $\Hrel^0(p)$}
\label{sec:hrel0}

The relative $0$-cochains $\Crel^0(p)$ are spanned by $\vac$ and from
Lemma~\ref{lem:df}, we see that
\begin{equation}
  d \vac = i p_i c_1 x^i_{-1} \vac,
\end{equation}
where we have used that $p_0=0$.  Hence there are no $0$-cocycles
unless $p=0$.  Since there are no $0$-coboundaries, we find that
\begin{equation}
  \Hrel^0(p) =
  \begin{cases}
    \CC\ket{0} & p = 0\\
    0 & p \neq 0.
  \end{cases}
\end{equation}

Also it follows that the $1$-coboundaries are
\begin{equation}
  \label{eq:1-coboundaries}
  \Brel^1(p) =\CC \left<i  p_i c_1 x^i_{-1} \vac \right>,
\end{equation}
where $\CC\left<\ldots\right>$ denotes the complex span.  Summarising,
\begin{equation}
  \Zrel^0(0) \cong \CC, \quad \Brel^1(0) = 0, \quad \Zrel^0(0,\p)
  = 0, \quad \Brel^1(0,\p) \cong \CC.
\end{equation}

\subsection{Calculating $\Hrel^1(p)$}
\label{sec:hrel1}

Using the state-field correspondence~\eqref{eq:state-field-corr}, we
first collect some formulae for the action of the BRST differential $d
\colon \Crel^1(0) \to \Crel^2(0)$ in terms of modes:
\begin{equation}
  \label{eq:d-on-basis-gh-number-1}
  \begin{aligned}
    dC_0 \ket{0} &= 2 (c_1 C_{-1} + C_1 c_{-1})\ket{0}\\
    dc_1 x^0_{-1}\ket{0} &= C_0 c_1 (\pi_0)_{-1} \ket{0} - c_1 C_1 (\pi_0)_{-2}\ket{0}\\
    dc_1 x^i_{-1} \ket{0} &= 0\\
    dC_1 x^0_{-1}\ket{0} &= 2 c_1 C_1 x^0_{-2}\ket{0} - C_0 c_1 x^0_{-1} \ket{0} + C_0 C_1 (\pi_0)_{-1} \ket{0}\\
    dC_1 x^i_{-1} \ket{0} &= - C_0 c_1 x^i_{-1} \ket{0} + 2 c_1 C_1 x^i_{-2}\ket{0}\\
    dc_1 (\pi_0)_{-1}\ket{0} &= 0\\
    dc_1 (\pi_i)_{-1}\ket{0} &= -\tau^2 C_0 c_1 x^i_{-1} \ket{0} + 2\tau^2 c_1 C_1 x^i_{-2} \ket{0}\\
    dC_1 (\pi_0)_{-1}\ket{0} &= - C_0 c_1 (\pi_0)_{-1} \ket{0} + c_1 C_1 (\pi_0)_{-2} \ket{0}\\    
    dC_1 (\pi_i)_{-1}\ket{0} &= - C_0 c_1 (\pi_i)_{-1} \ket{0} + c_1 C_1  (\pi_i)_{-2} \ket{0} - \tau^2 C_0 C_1 x^i_{-1} \ket{0}\\
    dc_1 C_1 b_{-2}\ket{0} &= c_1 C_1 x^\mu_{-1} (\pi_\mu)_{-1} \ket{0} + 2 C_0 c_1 C_1 B_{-2} \ket{0} + 3 c_1 C_{-1} \ket{0} + 3 c_{-1} C_1\ket{0}\\
    dc_1 C_1 B_{-2}\ket{0} &= -\tfrac12 \tau^2 \delta_{ij} c_1 C_1 x^i_{-1} x^j_{-1}\ket{0} + \tfrac12 c_1 C_1 (\pi_0)_{-1}^2 \ket{0} - 3 c_{-1} c_1 \ket{0}.
  \end{aligned}
\end{equation}
The first of these follows from Lemma~\ref{lem:dC0} when acting on the
vacuum with the expression for $[d,C_0]$ in~\eqref{eq:[d,C_0]}; the
others can be read from Appendix~\ref{sec:acti-brst-diff} and the
state-field correspondence \eqref{eq:state-field-corr}.  This will be
the approach to the calculations below, but we will not mention it
explicitly.

\subsubsection{The action of the differential}
\label{sec:action-differential-1}

It is then simply a matter of calculation to work out the action of
$d$ on $\Crel^1(p)$, with $p_0 = 0$:
\begin{equation}
  \label{eq:d-on-Crel1}
  \begin{aligned}
    dC_0 \vac &= - i p_i C_0 c_1 x^i_{-1} \vac + 2 c_1 C_{-1} \vac - 2 c_{-1} C_1 \vac\\
    dc_1 x^0_{-1}\vac &= C_0 c_1 (\pi_0)_{-1} \vac - c_1 C_1 (\pi_0)_{-2} \vac\\
    dc_1 x^i_{-1}\vac &=  0\\
    dC_1 x^0_{-1}\vac &= i p_i c_1 C_1 x^i_{-1}x^0_{-1}\vac + 2 c_1 C_1 x^0_{-2}\vac - C_0 c_1 x^0_{-1} \vac + C_0 C_1 (\pi_0)_{-1} \vac\\
    dC_1 x^i_{-1}\vac &= i p_j c_1 C_1 x^i_{-1}x^j_{-1}\vac - C_0 c_1 x^i_{-1} \vac + 2 c_1 C_1 x^i_{-2} \vac\\
    dc_1 (\pi_0)_{-1}\vac &= 0\\
    dc_1 (\pi_i)_{-1}\vac &= - i p_i c_1 c_{-1} \vac - \tau^2 C_0 c_1 x^i_{-1} \vac + 2 \tau^2 c_1 C_1 x^i_{-2} \vac\\
    dC_1 (\pi_0)_{-1}\vac &= i p_i  c_1 C_1 x^i_{-1} (\pi_0)_{-1}\vac +  c_1 C_1 (\pi_0)_{-2} \vac -  C_0 c_1 (\pi_0)_{-1}  \vac\\
    dC_1 (\pi_i)_{-1}\vac &= i p_j c_1 C_1 x^j_{-1}(\pi_i)_{-1}\vac +  i p_i  c_{-1} C_1 \vac -  C_0 c_1 (\pi_i)_{-1} \vac +  c_1 C_1  (\pi_i)_{-2} \vac\\
    & \qquad {} - \tau^2  C_0 C_1 x^i_{-1} \vac\\
    dc_1 C_1 b_{-2} \vac &=  2 i p_i  c_1 C_1 x^i_{-2} \vac + 2  C_0  c_1 C_1 B_{-2} \vac +  c_1 C_1 x^0_{-1} (\pi_0)_{-1}\vac +   c_1 C_1 x^i_{-1} (\pi_i)_{-1}\vac\\
    & \qquad {} + 3  c_1 C_{-1} \vac + 3  c_{-1} C_1 \vac\\
    dc_1 C_1 B_{-2} \vac &= -\tfrac12 \tau^2  c_1 C_1 x^i_{-1} x^i_{-1} \vac + \tfrac12  c_1 C_1 (\pi_0)^2_{-1} \vac + 3  c_1 c_{-1} \vac,
  \end{aligned}
\end{equation}
in which, having identified the vector and covector representations of
$\SO(25)$, we raise and lower vector indices with impunity.  We will
continue to do this tacitly below.

\subsubsection{Cocycles}
\label{sec:cocycles-1}

The general cochain in $\Crel^1(p)$ is given by the following expression:
\begin{multline}
  \label{eq:Crel1}
  \Psi = \phi^{(1)} C_0 \vac + \phi^{(2)} c_1 x^0_{-1} \vac + A^{(3)}_i c_1 x^i_{-1} \vac + \phi^{(4)} C_1 x^0_{-1}\vac\\
  + A^{(5)}_i C_1 x^i_{-1}\vac + \phi^{(6)} c_1 (\pi_0)_{-1}\vac + A^{(7)}_i c_1 (\pi_i)_{-1} \vac + \phi^{(8)} C_1 (\pi_0)_{-1}\vac \\
+ A^{(9)}_i C_1 (\pi_i)_{-1} \vac + \phi^{(10)} c_1 C_1 b_{-2} \vac + \phi^{(11)} c_1 C_1 B_{-2} \vac,
\end{multline}
where the $\phi^{(i)}, A_i^{(j)}$ are constants, which we can think of
values at $p$ of fields in momentum space. 
Using the formulae in equation~\eqref{eq:d-on-Crel1} we may compute
$d\Psi$.  We may then determine the cocycle conditions, which give
algebraic equations for the scalars $\phi^{(a)}$, for $a \in
\{1,2,4,6,8,10,11\}$ and the vectors $A_i^{(a)}$ for $a \in
\{3,5,7,9\}$ depending on $p$.  Imposing that the different terms in
$\Crel^2(p)$ vanish, we find that $d\Psi = 0$ is equivalent to the
following equations, once we impose $p_0=0$:
\begin{align*}
    - \phi^{(4)} &= 0 & \tag{$C_0 c_1 x^0_{-1}$}\\
  -i p_i \phi^{(1)}  - A^{(5)}_i - \tau^2 A^{(7)}_i &= 0 & \tag{$C_0 c_1 x^i_{-1}$}\\
  -\tau^2 A^{(9)}_i &= 0 & \tag{$C_0 C_1 x^i_{-1}$}\\
  \phi^{(2)} - \phi^{(8)} &= 0 & \tag{$C_0 c_1 (\pi_0)_{-1}$}\\
  - A^{(9)}_i &= 0 & \tag{$C_0 c_1 (\pi_i)_{-1}$}\\
  \phi^{(4)} &= 0 & \tag{$C_0 C_1 (\pi_0)_{-1}$}\\
  2 \phi^{(10)} &= 0 & \tag{$C_0 c_1 C_1 B_{-2}$}\\
- i p_i A^{(7)}_i + 3 \phi^{(11)} &= 0 & \tag{$c_1 c_{-1}$}\\
  2 \phi^{(1)} + 3 \phi^{(10)} &= 0 & \tag{$c_1 C_{-1}$}\\
  -2 \phi^{(1)} + i p_i A^{(9)}_i + 3 \phi^{(10)}  &= 0 & \tag{$c_{-1} C_1$}\\
   2 \phi^{(4)}  &= 0 & \tag{$c_1 C_1 x^0_{-2}$}\\
   2 A^{(5)}_i + 2 \tau^2 A^{(7)}_i + 2 i p_i \phi^{(10)}  &= 0 & \tag{$c_1 C_1 x^i_{-2}$}\\
   -\phi^{(2)} + \phi^{(8)} &= 0 & \tag{$c_1 C_1 (\pi_0)_{-2}$}\\
    A^{(9)}_i &= 0 & \tag{$c_1 C_1 (\pi_i)_{-2}$}\\
    \phi^{(10)} &= 0 & \tag{$c_1 C_1 x^0_{-1} (\pi_0)_{-1}$}\\
     i p_i \phi^{(8)}  &= 0 & \tag{$c_1 C_1 x^i_{-1} (\pi_0)_{-1}$}\\
  i p_i A^{(9)}_j + \delta_{ij} \phi^{(10)} &= 0 & \tag{$c_1 C_1 x^i_{-1} (\pi_j)_{-1}$}\\
  i p_i \phi^{(4)} &= 0 & \tag{$c_1 C_1 x^i_{-1} x^0_{-1}$}\\
  i p_{(i} A^{(5)}_{j)}  - \tfrac12 \delta_{ij} \tau^2 \phi^{(11)}  &= 0 & \tag{$c_1 C_1 x^i_{-1} x^j_{-1}$}\\  
   \tfrac12 \phi^{(11)}  &= 0 & \tag{$c_1 C_1 (\pi_0)_{-1}^2$}.
\end{align*}

\subsubsection{$\Hrel^1(0)$}
\label{sec:hrel10}

If $p = 0$, then it follows that $\phi^{(1)} = \phi^{(4)} =
\phi^{(10)}  = \phi^{(11)} = A_i^{(9)} = 0$, whereas $A^{(7)}_i =
-\frac1{\tau^2}A^{(5)}_i$ and $\phi^{(8)} = \phi^{(2)}$.  Finally,
$A^{(3)}_i$ and $\phi^{(6)}$ are unconstrained.  Therefore,
\begin{equation}
  \label{eq:Zrel10}
  \Zrel^1(0) = \CC \left< \left(c_1 x^0_{-1} + C_1  (\pi_0)_{-1}\right)\ket{0}, c_1 x^i_{-1}\ket{0}, c_1 (\pi_0)_{-1}\ket{0}, \left(C_1 x^i_{-1} - \tfrac1{\tau^2}
    c_1 (\pi_i)_{-1}\right)\ket{0} \right>.
\end{equation}
As a representation of $\SO(25)$, we thus have
\begin{equation}
  \Crel^1(0) \cong 7 \CC \oplus 4 V \qquad\text{and}\qquad \Zrel^1(0) \cong 2 \CC \oplus 2 V\\
\end{equation}
and hence, since $\Brel^1(0) =0$,
\begin{equation}
  \Hrel^1(0) \cong 2 \CC \oplus 2 V \qquad\text{and}\qquad \Brel^2(0) \cong 5 \CC \oplus 2 V.
\end{equation}

\subsubsection{$\Hrel^1(0,\p)$}
\label{sec:hrel1zero-p}

If $p = (0,\p)$ with $\p \neq \bzero$, then $\phi^{(a)} = 0$ for all
$a \in \{1,2,4,8,10,11\}$ and $A_i^{(9)} = 0$.  In addition,
$A_i^{(7)} = -\tfrac1{\tau^2} A_i^{(5)}$, where
\begin{equation}
  p_i A_j^{(5)} +   p_j A_i^{(5)}  = 0.
\end{equation}
Taking the trace, we see that $p_i A_i^{(5)}=0$.  Now, contracting with $p_j$ and using that $p_j
A_j^{(5)} = 0$, we see that $\|\p\|^2 A_i^{(5)} = 0$, which implies
that $A_i^{(5)} = 0$ and hence that $A_i^{(7)} = 0$.  In summary, only
$\phi^{(6)}$ and $A_i^{(3)}$ remain unconstrained and hence the
1-cocycles are
\begin{equation}
  \Zrel^1(0,\p) = \CC \left< c_1 x^i_{-1} \ket{0,\p}, c_1 (\pi_0)_{-1}\ket{0,\p} \right> \cong 2 \CC \oplus V^\perp.
\end{equation}
Since the 1-cochains are
\begin{equation}
  \Crel^1(0,\p) \cong 11 \CC \oplus 4 V^\perp,
\end{equation}
we see that
\begin{equation}
  \Brel^2(0,\p) \cong 9 \CC \oplus 3 V^\perp.
\end{equation}
Since the 1-coboundaries are $\Brel^1(0,\p) \cong \CC$, we see that
\begin{equation}
  \label{eq:H1rel0p}
  \Hrel^1(0,\p) \cong \CC \oplus V^\perp.
\end{equation}

\subsection{Calculating $\Hrel^2(p)$}
\label{sec:hrel2}

The relative subcomplex at ghost number $2$ and momentum $p$ is
$\Crel^2(p) = \fM^2(p) \oplus C_0 \fM^1(p)$, where $\fM^\bullet(p)$ is
obtained from $\fM^\bullet = \fM^\bullet(0)$ replacing $\ket{0}$ with
$\vac$.  Using Lemmas~\ref{lem:dC0} and \ref{lem:df}, we may work out
the action of the BRST differential in terms of its action on $\fM^1$
and $\fM^2$.  In terms of modes, the former is given by
equation~\eqref{eq:d-on-basis-gh-number-1}, so it remains to collect
the expressions of $d$ on $\fM^2$ in terms of modes, which we do using
the dictionary between fields and modes in~\eqref{eq:state-field-corr}
and the expressions for the action of $d$ on fields of ghost number $2$
in~\eqref{eq:d-on-M-scalars-gh-no-2}--\eqref{eq:d-on-M-tensors-gh-no-2}:
\begin{equation}
  \label{eq:d-on-basis-gh-number-2}
  \begin{aligned}
    d c_1 c_{-1}\ket{0} &= 0\\
    d c_1 C_{-1}\ket{0} &= C_0 c_1 c_{-1}\ket{0} - 3 c_1 C_1 c_{-2} \ket{0}\\
    d c_{-1} C_1 \ket{0} &= C_0 c_1 c_{-1}\ket{0} - 3 c_1 C_1 c_{-2} \ket{0}\\
    d C_1 C_{-1}\ket{0} &= - C_0 c_1 C_{-1} \ket{0} - C_0 c_{-1} C_1 \ket{0} + 3 c_1 C_1 C_{-2}\ket{0}\\
    d c_1 C_1 (\pi_0)_{-2}\ket{0} &= 2 c_1 C_1 c_{-1} (\pi_0)_{-1}\ket{0}\\
    d c_1 C_1 x^0_{-2}\ket{0} &= c_1 C_1 c_{-1} x^0_{-1}\ket{0} + c_1 C_1 C_{-1} (\pi_0)_{-1}\ket{0} + C_0 c_1 C_1 (\pi_0)_{-2}\ket{0}\\
    d c_1 C_1 (\pi_0)_{-1}^2\ket{0} &= 0\\
    d c_1 C_1 x^0_{-1} (\pi_0)_{-1}\ket{0} &=c_1 C_1 c_{-2} \ket{0} + C_0 c_1 C_1 (\pi_0)^2_{-1}\ket{0} \\
    d c_1 C_1 (x^0_{-1})^2 \ket{0} &= c_1 C_1 C_{-2} \ket{0} + 2 C_0 c_1 C_1 x^0_{-1} (\pi_0)_{-1}\ket{0}\\
    d c_1 C_1 x^i_{-2}\ket{0} &= c_1 C_1 c_{-1} x^i_{-1}\ket{0}\\
    d c_1 C_1 (\pi_i)_{-2}\ket{0} &= 2 c_1 C_1 c_{-1} (\pi_i)_{-1}\ket{0} - 2 \tau^2 c_1 C_1 C_{-1} x^i_{-1}\ket{0} - 4 \tau^2 C_0 c_1 C_1 x^i_{-2}\ket{0}\\
    d c_1 C_1 (\pi_0)_{-1} x^i_{-1}\ket{0} &= 0\\
    d c_1 C_1 (\pi_i)_{-1} x^0_{-1}\ket{0} &= C_0 c_1 C_1 (\pi_i)_{-1} (\pi_0)_{-1} \ket{0} - \tau^2 C_0 c_1 C_1 x^i_{-1} x^0_{-1}\ket{0}\\
    d c_1 C_1 (\pi_i)_{-1} (\pi_0)_{-1}\ket{0} &= -\tau^2 C_0 c_1 C_1 x^i_{-1} (\pi_0)_{-1}\ket{0}\\    
    d c_1 C_1 x^0_{-1} x^i_{-1}\ket{0} &= C_0 c_1 C_1 x^i_{-1} (\pi_0)_{-1}\ket{0}\\
    d c_1 C_1 x^i_{-1} (\pi_j)_{-1} \ket{0} &= \delta^i_j c_1 C_1 c_{-2}\ket{0} - \tau^2 C_0 c_1 C_1 x^i_{-1} x^j_{-1}\ket{0}\\
    d c_1 C_1 x^i_{-1} x^j_{-1} \ket{0} &= 0\\
    d c_1 C_1 (\pi_i)_{-1} (\pi_j)_{-1} \ket{0} &= - \tau^2 \delta_{ij} c_1 C_1 C_{-2}\ket{0} - \tau^2 C_0 c_1 C_1 (\pi_i)_{-1} x^j_{-1} \ket{0} - \tau^2 C_0 c_1 C_1 (\pi_j)_{-1} x^i_{-1}\ket{0}.
  \end{aligned}
\end{equation}

\subsubsection{The action of the differential}
\label{sec:action-differential-2}

We are now ready to calculate $d: \Crel^2(p) \to \Crel^3(p)$.  We
start with the action of $d$ on $\fM^2(p)$, which,
using~\eqref{eq:desired-expression} and setting $p_0 = 0$, is given by
\begin{align*}
    d c_1 c_{-1}\vac &= 0\\
    d c_1 C_{-1}\vac &= C_0 c_1 c_{-1} \vac - 3  c_1 C_1 c_{-2}\vac\\
    d c_{-1} C_1\vac &= - i p_i  c_1 C_1 c_{-1} x^i_{-1}\vac + C_0 c_1 c_{-1} \vac - 3  c_1 C_1 c_{-2}\vac\\
    d C_1 C_{-1}\vac &= i p_i  c_1 C_1 C_{-1} x^i_{-1}\vac -  C_0 c_1 C_{-1}\vac -  C_0 c_{-1} C_1 \vac + 3  c_1 C_1 C_{-2}\vac\\
    d c_1 C_1 (\pi_0)_{-2}\vac &= 2  c_1 C_1 c_{-1} (\pi_0)_{-1} \vac\\
    d c_1 C_1 x^0_{-2}\vac &= c_1 C_1 c_{-1} x^0_{-1}\vac +  c_1 C_1 C_{-1} (\pi_0)_{-1} \vac +  C_0 c_1 C_1 (\pi_0)_{-2} \vac\\
    d c_1 C_1 (\pi_0)_{-1}^2\vac &= 0\\
    d c_1 C_1 x^0_{-1} (\pi_0)_{-1}\vac &= c_1 C_1 c_{-2} \vac +  C_0 c_1 C_1 (\pi_0)_{-1}^2\vac \\
    d c_1 C_1 (x^0_{-1})^2 \vac &= c_1 C_1 C_{-2} \vac + 2  C_0 c_1 C_1 x^0_{-1} (\pi_0)_{-1}\vac\\
    d c_1 C_1 x^i_{-2}\vac &= c_1 C_1 c_{-1} x^i_{-1}\vac\\
    d c_1 C_1 (\pi_i)_{-2}\vac &= 2 i p_i  c_1 C_1 c_{-2} \vac  + 2 c_1 C_1 c_{-1} (\pi_i)_{-1}\vac - 2  \tau^2  c_1 C_1 C_{-1} x^i_{-1}\vac\\
     & \qquad {} - 4\tau^2  C_0 c_1 C_1 x^i_{-2}\vac\\
    d c_1 C_1 (\pi_0)_{-1} x^i_{-1}\vac &= 0\\
    d c_1 C_1 x^0_{-1}(\pi_i)_{-1} \vac &= i p_i  c_1 C_1 c_{-1} x^0_{-1}\vac +  C_0 c_1 C_1 (\pi_i)_{-1} (\pi_0)_{-1} \vac - \tau^2  C_0 c_1 C_1 x^i_{-1} x^0_{-1}\vac\\
    d c_1 C_1 (\pi_i)_{-1} (\pi_0)_{-1}\vac &= i p_i  c_1 C_1 c_{-1}  (\pi_0)_{-1} \vac - \tau^2  C_0 c_1 C_1 x^i_{-1} (\pi_0)_{-1}\vac\\
    d c_1 C_1 x^0_{-1} x^i_{-1}\vac &= C_0 c_1 C_1 x^i_{-1} (\pi_0)_{-1}\vac\\
    d c_1 C_1 x^i_{-1} (\pi_j)_{-1}\vac &= i p_j  c_1 C_1 c_{-1} x^i_{-1} \vac +  \delta^i_j c_1 C_1 c_{-2}\vac - \tau^2  C_0 c_1 C_1 x^i_{-1} x^j_{-1}\vac\\
    d c_1 C_1 x^i_{-1} x^j_{-1} \vac &= 0\\
    d c_1 C_1 (\pi_i)_{-1} (\pi_j)_{-1} \vac &= i p_i  c_1 C_1 c_{-1} (\pi_j)_{-1} \vac + i p_j  c_1 C_1 c_{-1} (\pi_i)_{-1} \vac - \tau^2  \delta_{ij} c_1 C_1 C_{-2}\vac \\
    & \qquad {} - \tau^2  C_0 c_1 C_1 (\pi_i)_{-1} x^j_{-1} \vac - \tau^2  C_0 c_1 C_1 (\pi_j)_{-1} x^i_{-1}\vac .
\end{align*}

It remains now to calculate the action of $d$ on $C_0 \fM^1(p)$, which
is given by
\begin{equation}
    d C_0 \fM^1(p) = [d,C_0]\fM^1(p) -  C_0 d \fM^1(p),
\end{equation}
where $[d,C_0]$ is given in~\eqref{eq:[d,C_0]} and the action of $d$
on $\fM^1(p)$ is given in~\eqref{eq:d-on-Crel1}.  With this, we find
\begin{equation}
  \label{eq:d-on-C0-M1}
  \begin{aligned}
    d C_0 c_1 C_1 b_{-2}\vac &=-2 i p_i  C_0 c_1 C_1 x^i_{-2}\vac - 4  c_1 C_1 C_{-2} \vac -  C_0 c_1 C_1 x^\mu_{-1} (\pi_\mu)_{-1} \vac\\
    & \qquad {} - 3  C_0 c_1 C_{-1}\vac + 3  C_0 C_1 c_{-1} \vac\\
    d  C_0 c_1 C_1 B_{-2}\vac &= - 4  c_1 C_1 c_{-2}\vac - \tfrac12  C_0 c_1 C_1 (\pi_0)_{-1}^2 \vac + \tfrac12 \tau^2   C_0 c_1 C_1 x^i_{-1} x^i_{-1}\vac\\
    & \qquad {} - 3   C_0 c_1 c_{-1} \vac\\
    d  C_0 c_1 x^0_{-1}\vac &= 2  c_1 C_1 c_{-1}x^0_{-1}\vac +  C_0 c_1 C_1 (\pi_0)_{-2}\vac\\
    d  C_0 C_1 x^0_{-1}\vac &= - i p_i  C_0 c_1 C_1 x^i_{-1} x^0_{-1}\vac - 2  c_1 C_1 C_{-1} x^0_{-1}\vac - 2  C_0 c_1 C_1 x^0_{-2}\vac\\
    d  C_0 c_1 (\pi_0)_{-1}\vac &= 2  c_1 C_1 c_{-1} (\pi_0)_{-1}\vac\\
    d  C_0 C_1 (\pi_0)_{-1}\vac &= - i p_i  C_0 c_1 C_1 x^i_{-1} (\pi_0)_{-1}\vac - 2  c_1 C_1 C_{-1} (\pi_0)_{-1} \vac -  C_0 c_1 C_1 (\pi_0)_{-2} \vac\\
    d  C_0 c_1 x^i_{-1}\vac &= 2  c_1 C_1 c_{-1} x^i_{-1} \vac\\
    d  C_0 C_1 x^i_{-1}\vac &= -i p_j  C_0 c_1 C_1 x^j_{-1} x^i_{-1} \vac - 2  c_1 C_1 C_{-1} x^i_{-1} \vac - 2  C_0 c_1 C_1 x^i_{-2}\vac\\
    d  C_0 c_1 (\pi_i)_{-1}\vac &= i p_i  C_0 c_1 c_{-1}\vac + 2  c_1 C_1 c_{-1} (\pi_i)_{-1}\vac - 2 \tau^2  C_0 c_1 C_1 x^i_{-2}\vac\\
    d  C_0 C_1 (\pi_i)_{-1}\vac &= - i p_i  C_0 c_{-1} C_1\vac - i p_j  C_0 c_1 C_1 x^j_{-1} (\pi_i)_{-1}\vac - 2  c_1 C_1 C_{-1} (\pi_i)_{-1}\vac\\
        & \qquad {} -  C_0 c_1 C_1 (\pi_i)_{-2}\vac.
  \end{aligned}
\end{equation}

\subsubsection{Cocycles}
\label{sec:cocycles-2}

We parametrise the general cochain $\Psi \in \Crel^2(p)$ as follows:
\begin{multline}
  \label{eq:Crel2}
  \Psi = \phi^{(1)} C_0 c_1 C_1 b_{-2} \vac + \phi^{(2)} C_0 c_1 C_1 B_{-2} \vac + \phi^{(3)} C_0 c_1 x^0_{-1} \vac + A_i^{(4)} C_0 c_1 x^i_{-1}\vac + \phi^{(5)} C_0 C_1 x^0_{-1}\vac\\
  + A^{(6)}_i C_0 C_1 x^i_{-1} \vac + \phi^{(7)} C_0 c_1 (\pi_0)_{-1} \vac + A^{(8)}_i C_0 c_1 (\pi_i)_{-1} \vac + \phi^{(9)} C_0 C_1 (\pi_0)_{-1}\vac\\ + A^{(10)}_i C_0 C_1 (\pi_i)_{-1}\vac + \phi^{(11)} c_1 c_{-1} \vac + \phi^{(12)} c_1 C_{-1} \vac + \phi^{(13)} c_{-1} C_1 \vac + \phi^{(14)} C_1 C_{-1} \vac\\
  + \phi^{(15)} c_1 C_1 x^0_{-2} \vac + A_i^{(16)} c_1 C_1 x^i_{-2} \vac + \phi^{(17)} c_1 C_1 (\pi_0)_{-2}\vac + A_i^{(18)} c_1 C_1 (\pi_i)_{-2} \vac\\
  + \phi^{(19)} c_1 C_1 x^0_{-1} (\pi_0)_{-1} \vac + \phi^{(20)} c_1 C_1 (x^0_{-1})^2 \vac + \phi^{(21)} c_1 C_1 (\pi_0)_{-1}^2 \vac + A_i^{(22)} c_1 C_1 x^i_{-1} (\pi_0)_{-1} \vac\\
  + A_i^{(23)} c_1 C_1 x^0_{-1} (\pi_i)_{-1} \vac + A_i^{(24)} c_1 C_1 x^i_{-1} x^0_{-1}\vac + A_i^{(25)} c_1 C_1 (\pi_i)_{-1} (\pi_0)_{-1} \vac \\
  + T_{ij}^{(26)} c_1 C_1 x^i_{-1} x^j_{-1} \vac + T_{ij}^{(27)} c_1 C_1 x^i_{-1} (\pi_j)_{-1} \vac + T_{ij}^{(28)} c_1 C_1 (\pi_i)_{-1} (\pi_j)_{-1} \vac,
\end{multline}
where $T_{ij}^{(26)} = T_{ji}^{(26)}$ and $T_{ij}^{(28)} = T_{ji}^{(28)}$.

The equation $d\Psi = 0$ breaks up into the following component
equations:
\begin{align*}
  -4 \phi^{(2)} - 3 \phi^{(12)} - 3 \phi^{(13)} + \phi^{(19)} + 2 i p_i A_i^{(18)} + T^{(27)}_{ii} &= 0 & \tag{$c_1 C_1 c_{-2}$}\\
  -4 \phi^{(1)} + 3 \phi^{(14)} + \phi^{(20)} - \tau^2 T^{(28)}_{ii} &= 0 & \tag{$c_1 C_1 C_{-2}$}\\
  2 \phi^{(3)} + \phi^{(15)} + i p_i A_i^{(23)} &= 0 & \tag{$c_1 C_1 c_{-1} x^0_{-1}$}\\
  2 A_i^{(4)} - i p_i \phi^{(13)} + A_i^{(16)} + i p_j T^{(27)}_{ij} &= 0 & \tag{$c_1 C_1 c_{-1} x^i_{-1}$}\\
  - 2 \phi^{(5)} &= 0 & \tag{$c_1 C_1 C_{-1} x^0_{-1}$}\\
  -2 A_i^{(6)} + i p_i\phi^{(14)} - 2 \tau^2 A_i^{(18)} &= 0 & \tag{$c_1 C_1 C_{-1} x^i_{-1}$}\\
   2 \phi^{(7)} + 2 \phi^{(17)} + i p_i A_i^{(25)} &= 0 & \tag{$c_1 C_1 c_{-1} (\pi_0)_{-1}$}\\
  2 A_i^{(8)} + 2 A_i^{(18)} + 2 i p_j T^{(28)}_{ij}  &= 0 & \tag{$c_1 C_1 c_{-1} (\pi_i)_{-1}$}\\
- 2 \phi^{(9)} + \phi^{(15)} &= 0 & \tag{$c_1 C_1 C_{-1} (\pi_0)_{-1}$}\\
  -2 A_i^{(10)} &= 0 & \tag{$c_1 C_1 C_{-1} (\pi_i)_{-1}$}\\
  -3 \phi^{(2)} + i p_i A_i^{(8)} + \phi^{(12)} + \phi^{(13)}  &= 0 & \tag{$C_0 c_1 c_{-1}$}\\
  -3 \phi^{(1)} - \phi^{(14)}  &= 0 & \tag{$C_0 c_1 C_{-1}$}\\
  -3 \phi^{(1)} - i p_i A_i^{(10)} - \phi^{(14)} &= 0 & \tag{$C_0 c_{-1} C_1$}\\
  -2 \phi^{(5)} &= 0 & \tag{$C_0 c_1 C_1 x^0_{-2}$}\\
  -2 i p_i \phi^{(1)} - 2 A_i^{(6)} - 2 \tau^2 A_i^{(8)} - 4 \tau^2 A_i^{(18)} &= 0 & \tag{$C_0 c_1 C_1 x^i_{-2}$}\\
  \phi^{(3)} - \phi^{(9)} + \phi^{(15)} &= 0 & \tag{$C_0 c_1 C_1 (\pi_0)_{-2}$}\\
  -A_i^{(10)} &= 0 & \tag{$C_0 c_1 C_1 (\pi_i)_{-2}$}\\
  - \phi^{(1)} + 2 \phi^{(20)}  &= 0 & \tag{$C_0 c_1 C_1 x^0_{-1} (\pi_0)_{-1}$}\\
 - i p_i \phi^{(9)} - \tau^2 A_i^{(25)} + A_i^{(24)} &= 0 & \tag{$C_0 c_1 C_1 x^i_{-1} (\pi_0)_{-1}$}\\
 - \delta_{ij} \phi^{(1)} - i p_i A_j^{(10)}  - 2 \tau^2 T^{(28)}_{ij} &= 0 & \tag{$C_0 c_1 C_1 x^i_{-1} (\pi_j)_{-1}$}\\
 - i p_i \phi^{(5)} - \tau^2 A_i^{(23)} &= 0 & \tag{$C_0 c_1 C_1 x^0_{-1} x^i_{-1}$}\\
 \tfrac12 \delta_{ij} \tau^2 \phi^{(2)} - i p_{(i} A_{j)}^{(6)}  - \tau^2 T^{(27)}_{(ij)}  &= 0 & \tag{$C_0 c_1 C_1 x^i_{-1} x^j_{-1}$}\\
 - \tfrac12 \phi^{(2)} + \phi^{(19)} &= 0 & \tag{$C_0 c_1 C_1 (\pi_0)_{-1}^2$}\\
   A_i^{(23)} &= 0. & \tag{$C_0 c_1 C_1 (\pi_0)_{-1} (\pi_i)_{-1}$}
\end{align*}

\subsubsection{$\Hrel^2(0)$}
\label{sec:hrel20}
 
If $p=0$, the cocycle conditions become
\begin{align}
  \begin{aligned}
    -4 \phi^{(2)} - 3 \phi^{(12)} - 3 \phi^{(13)} +   \phi^{(19)} + T^{(27)}_{ii} &= 0 \\
    -4 \phi^{(1)} + 3 \phi^{(14)} + \phi^{(20)} - \tau^2 T^{(28)}_{ii} &= 0 \\
    2 \phi^{(3)} + \phi^{(15)} &= 0 \\
    2 A_i^{(4)} + A_i^{(16)}  &= 0 \\
    - 2 \phi^{(5)} &= 0 \\
    -2 A_i^{(6)} - 2 \tau^2 A_i^{(18)} &= 0 \\
    2 \phi^{(7)} + 2 \phi^{(17)}  &= 0 \\
    2 A_i^{(8)} + 2 A_i^{(18)}   &= 0 \\
    - 2 \phi^{(9)} + \phi^{(15)}  &= 0 \\
    -2 A_i^{(10)} &= 0 \\
  \end{aligned}\qquad
  \begin{aligned}
    -3 \phi^{(2)} + \phi^{(12)} + \phi^{(13)}  &= 0 \\
    -3 \phi^{(1)} - \phi^{(14)} &= 0 \\
    - 2 A_i^{(6)}  - 2 \tau^2 A_i^{(8)} - 4 \tau^2 A_i^{(18)} &= 0 \\
    \phi^{(3)} - \phi^{(9)} + \phi^{(15)} &= 0 \\
    - \phi^{(1)} + 2 \phi^{(20)}  &= 0 \\
    - \tau^2 A_i^{(25)} + A_i^{(24)} &= 0 \\
    - \delta_{ij} \phi^{(1)} - 2 \tau^2 T^{(28)}_{ij} &= 0 \\
    - \tau^2 A_i^{(23)}  &= 0 \\
    \tfrac12 \delta_{ij} \tau^2 \phi^{(2)} - \tau^2 T^{(27)}_{(ij)}  &= 0 \\
    - \tfrac12 \phi^{(2)} + \phi^{(19)} &= 0 \\
  \end{aligned}
\end{align}

The cocycle conditions allow us to solve for the following:
\begin{equation}
  \begin{aligned}
    \phi^{(5)} &= 0\\
    A_i^{(8)} &= \tfrac1{\tau^2} A_i^{(6)}\\
    \phi^{(9)} &= - \phi^{(3)}\\
    A_i^{(10)} &= 0\\
    \phi^{(13)} &= 3 \phi^{(2)} - \phi^{(12)}\\
    \phi^{(14)} &= -3 \phi^{(1)}\\
    \phi^{(15)} &= -2 \phi^{(3)}\\
    A_i^{(16)} &= -2 A_i^{(4)}
  \end{aligned}
  \qquad\qquad
  \begin{aligned}
    \phi^{(17)} &= -\phi^{(7)}\\
    A_i^{(18)} &= -\tfrac1{\tau^2} A_i^{(6)}\\
    \phi^{(19)} &= \tfrac12 \phi^{(2)}\\
    \phi^{(20)} &= \tfrac12 \phi^{(1)}\\
    A_i^{(23)} &=  0\\
    A_i^{(25)} &=  \tfrac1{\tau^2} A_i^{(24)}\\
    T^{(27)}_{(ij)} &= \tfrac12 \delta_{ij} \phi^{(2)}\\
    T^{(28)}_{ij} &= -\tfrac1{2\tau^2} \delta_{ij} \phi^{(1)}.
  \end{aligned}
\end{equation}
As representations of $\SO(25)$, the relative $2$-cochains are
\begin{equation}
  \Crel^2(0) \cong 18 \CC \oplus 10 V \oplus 3 \odot^2_0 V \oplus \ext{2} V
\end{equation}
and since the cocycle conditions set to zero $10$ scalars (including
the traces of $T^{(27)}$ and $T^{(28)}$), $6$ vectors and two
traceless symmetric tensors ($T^{(27)}_{\langle ij \rangle}$ and
$T^{(28)}_{\langle ij \rangle}$), we have that
\begin{equation}
  \Zrel^2(0) \cong 8 \CC \oplus 4 V \oplus \odot^2_0 V \oplus \ext{2} V
\end{equation}
and hence that
\begin{equation}
  \Brel^3(0) \cong 10 \CC \oplus 6 V \oplus 2 \odot^2_0 V.
\end{equation}
Since $\Brel^2(0) \cong 5\CC \oplus 2 V$, it follows that
\begin{equation}
  \Hrel^2(0) \cong 3 \CC \oplus 2 V \oplus \odot^2_0 V \oplus \ext{2} V.
\end{equation}

\subsubsection{$\Hrel^2(0,\p)$}
\label{sec:hrel20p}

The cocycle conditions reduce to the following relations:
\begin{equation}
  \begin{aligned}
    \phi^{(1)} = \phi^{(5)} = \phi^{(14)} = \phi^{(20)} &= 0\\
    A_i^{(10)} = A_i^{(23)} &= 0\\
    A_i^{(6)} &= \tau^2 A_i^{(8)}\\
    \phi^{(9)} &= - \phi^{(3)}\\
    \phi^{(12)} &= 3 \phi^{(2)} - i p_i A_i^{(8)} - \phi^{(13)}\\
    \phi^{(15)} &= - 2 \phi^{(3)}\\
    A_i^{(16)} &= -\tfrac{i}2 p_i \phi^{(2)} - 2 A_i^{(4)} -   \tfrac{p^2}2 A_i^{(8)} - \tfrac12 p_i p_j A_j^{(8)} - i p_i \phi^{(13)} - i p_j T^{(27)}_{[ij]}\\
    \phi^{(17)} &= - \phi^{(7)} - \tfrac{i}2 p_i A_i^{(25)}\\
    A_i^{(18)} &= - A_i^{(8)} \\
    \phi^{(19)} &= \tfrac12 \phi^{(2)}\\
    A_i^{(24)} &= \tau^2 A_i^{(25)} - i p_i \phi^{(3)}\\
    T^{(27)}_{(ij)} &= \tfrac12 \delta_{ij} \phi^{(2)} - i p_{(i} A_{j)}^{(8)}\\
    T^{(28)}_{ij} &= 0\\
  \end{aligned}
\end{equation}
leaving unconstrained the following: $\phi^{(2)}$,
$\phi^{(3)}$, $A_i^{(4)}$, $\phi^{(7)}$, $A_i^{(8)}$, $\phi^{(11)}$, $\phi^{(13)}$, $\phi^{(21)}$,
$A_i^{(22)}$, $A_i^{(25)}$, $T^{(26)}_{ij}$ and $T^{(27)}_{[ij]}$.

As an $\SO(24)$ representation, $V \cong \CC \oplus V^\perp$ and hence
\begin{equation}
  \Crel^2(0,\p) \cong 31\CC \oplus 14 V^\perp \oplus 3 \odot^2_0 V^\perp \oplus \ext{2} V^\perp.
\end{equation}
One sees from the cocycle conditions that
\begin{equation}
  \Zrel^2(0,\p) \cong 12 \CC \oplus 6 V^\perp \oplus \odot^2_0 V^\perp \oplus \ext{2} V^\perp,
\end{equation}
so that from the Rank Theorem,
\begin{equation}
  \Brel^3(0,\p) \cong 19 \CC \oplus 8 V^\perp \oplus 2 \odot^2_0 V^\perp.
\end{equation}
Since $\Brel^2(0,\p) \cong 9\CC \oplus 3 V^\perp$, it follows that
\begin{equation}
  \Hrel^2(0,\p) \cong 3 \CC \oplus 3 V^\perp  \oplus \odot^2_0 V^\perp \oplus \ext{2} V^\perp.
\end{equation}

\subsection{Calculating $\Hrel^3(p)$}
\label{sec:hrel3}

The relative subcomplex $\Crel^3(p) = \fM^3(p) \oplus C_0 \fM^2(p)$.
We can compute the action of the differential as usual by using
Lemmas~\ref{lem:dC0} and \ref{lem:df} and the formulae in
Appendix~\ref{sec:acti-brst-diff}.

\subsubsection{The action of the differential}
\label{sec:action-differential-3}

Doing so we calculate $d : \Crel^3(p) \to \Crel^4(p)$, to give the
following, where we have already set $p_0 = 0$
\begin{align}
  \label{eq:d-on-Crel-3}
  \begin{aligned}
    d c_1 C_1 c_{-2}\vac &= 0\\
    d c_1 C_1 C_{-2}\vac &= 2 C_0 c_1 C_1 c_{-2}\vac\\
    d c_1 C_1 c_{-1} x^0_{-1}\vac &= C_0 c_1 C_1 c_{-1} (\pi_0)_{-1}\vac\\
    d c_1 C_1 c_{-1} x^i_{-1}\vac &= 0\\
    d c_1 C_1 C_{-1} x^0_{-1}\vac &= C_0 c_1 C_1 c_{-1} x^0_{-1} \vac +  C_0 c_1 C_1 C_{-1} (\pi_0)_{-1}\vac\\
    d c_1 C_1 C_{-1} x^i_{-1}\vac &= C_0 c_1 C_1 c_{-1} x^i_{-1}\vac\\
    d c_1 C_1 c_{-1} (\pi_0)_{-1}\vac &= 0\\
    d c_1 C_1 c_{-1} (\pi_i)_{-1}\vac &= - \tau^2 C_0 c_1 C_1 c_{-1} x^i_{-1}\vac\\
    d c_1 C_1 C_{-1} (\pi_0)_{-1}\vac &= C_0 c_1 C_1 c_{-1} (\pi_0)_{-1}\vac\\
    d c_1 C_1 C_{-1} (\pi_i)_{-1}\vac &= i p_i c_1 C_1 c_{-1} C_{-1} \vac + C_0 c_1 C_1 c_{-1} (\pi_i)_{-1}\vac - \tau^2 C_0 c_1 C_1 C_{-1} x^i_{-1}\vac\\
    d C_0 c_1 c_{-1}\vac &= 0\\
    d C_0 c_1 C_{-1}\vac &= 2 c_1 C_1 c_{-1} C_{-1} \vac + 3 C_0 c_1 C_1 c_{-2} \vac\\
    d C_0 c_{-1} C_1\vac &= i p_i C_0 c_1 C_1 c_{-1} x^i_{-1} \vac - 2 c_1 C_1 c_{-1} C_{-1} \vac + 3 C_0 c_1 C_1 c_{-2} \vac\\
    d C_0 C_1 C_{-1}\vac &= - i p_i C_0 c_1 C_1 C_{-1} x^i_{-1} \vac - 3 C_0 c_1 C_1 C_{-2} \vac\\
    d C_0 c_1 C_1 x^0_{-2}\vac &= - C_0 c_1 C_1 C_{-1} (\pi_0)_{-1} \vac -  C_0 c_1 C_1 c_{-1} x^0_{-1} \vac\\
    d C_0 c_1 C_1 x^i_{-2}\vac &= - C_0 c_1 C_1 c_{-1} x^i_{-1} \vac\\
    d C_0 c_1 C_1 (\pi_0)_{-2}\vac &= - 2 C_0 c_1 C_1 c_{-1} (\pi_0)_{-1}\vac\\
    d C_0 c_1 C_1 (\pi_i)_{-2}\vac &= -2 i p_i C_0 c_1 C_1 c_{-2} \vac + 2 \tau^2 C_0 c_1 C_1 C_{-1} x^i_{-1}\vac - 2 C_0 c_1 C_1 c_{-1} (\pi_i)_{-1} \vac\\
    d C_0 c_1 C_1 x^0_{-1} (\pi_0)_{-1}\vac &= - C_0 c_1 C_1 c_{-2} \vac\\
    d C_0 c_1 C_1 x^i_{-1} (\pi_0)_{-1} \vac &= 0\\
    d C_0 c_1 C_1 x^0_{-1} (\pi_i)_{-1} \vac &= - i p_i C_0 c_1 C_1 c_{-1} x^0_{-1} \vac\\
    d C_0 c_1 C_1 x^i_{-1} (\pi_j)_{-1}\vac &= - i p_j C_0 c_1 C_1 c_{-1} x^i_{-1}\vac - \delta_{ij} C_0 c_1 C_1 c_{-2} \vac\\
    d C_0 c_1 C_1 (x^0_{-1})^2 \vac &= - C_0 c_1 C_1 C_{-2} \vac\\
    d C_0 c_1 C_1 x^0_{-1} x^i_{-1}\vac &= 0\\
    d C_0 c_1 C_1 x^i_{-1} x^j_{-1}\vac &= 0\\
    d C_0 c_1 C_1 (\pi_0)_{-1}^2 \vac &= 0\\
    d C_0 c_1 C_1 (\pi_0)_{-1} (\pi_i)_{-1}\vac &= - i p_i C_0 c_1 C_1 c_{-1} (\pi_0)_{-1}\vac\\
    d C_0 c_1 C_1 (\pi_i)_{-1} (\pi_j)_{-1}\vac &= - i C_0 c_1 C_1  c_{-1} \left( p_j (\pi_i)_{-1}+ p_i (\pi_j)_{-1} \right)\vac + \tau^2 \delta_{ij} C_0 c_1 C_1 C_{-2}\vac.
  \end{aligned}
\end{align}

\subsubsection{Cocycles}
\label{sec:cocycles-3}

The general cochain in $\Crel^3(p)$ is given by
\begin{multline}
  \label{eq:Crel3}
  \Psi = \phi^{(1)} c_1 C_1 c_{-2}\vac + \phi^{(2)} c_1 C_1 C_{-2}\vac + \phi^{(3)} c_1 C_1 c_{-1} x^0_{-1}\vac + A_i^{(4)} c_1 C_1 c_{-1} x^i_{-1}\vac \\
  + \phi^{(5)} c_1 C_1 C_{-1} x^0_{-1} \vac + A_i^{(6)} c_1 C_1 C_{-1} x^i_{-1}\vac + \phi^{(7)} c_1 C_1c_{-1} (\pi_0)_{-1}\vac + A_i^{(8)} c_1 C_1 c_{-1} (\pi_i)_{-1}\vac\\
  + \phi^{(9)} c_1 C_1 C_{-1} (\pi_0)_{-1}\vac  + A_i^{(10)} c_1 C_1 C_{-1} (\pi_i)_{-1}\vac + \phi^{(11)} C_0 c_1 c_{-1}\vac + \phi^{(12)} C_0 c_1 C_{-1}\vac \\
  + \phi^{(13)} C_0 c_{-1} C_1\vac + \phi^{(14)} C_0 C_1 C_{-1}\vac  + \phi^{(15)} C_0 c_1 C_1 x^0_{-2}\vac + A_i^{(16)} C_0 c_1 C_1 x^i_{-2}\vac\\
  + \phi^{(17)} C_0 c_1 C_1 (\pi_0)_{-2}\vac + A_i^{(18)} C_0 c_1 C_1 (\pi_i)_{-2}\vac + \phi^{(19)} C_0 c_1 C_1 x^0_{-1} (\pi_0)_{-1}\vac\\
   + A_i^{(20)} C_0 c_1 C_1 x^i_{-1} (\pi_0)_{-1}\vac + A_i^{(21)} C_0 c_1 C_1 x^0_{-1} (\pi_i)_{-1}\vac + T_{ij}^{(22)} C_0 c_1 C_1 x^i_{-1} (\pi_j)_{-1}\vac\\
   + \phi^{(23)} C_0 c_1 C_1 (x^0_{-1})^2\vac+ A_i^{(24)} C_0 c_1 C_1 x^0_{-1} x^i_{-1}\vac + T_{ij}^{(25)} C_0 c_1 C_1 x^i_{-1} x^j_{-1}\vac\\
   + \phi^{(26)} C_0 c_1 C_1 (\pi_0)_{-1}^2\vac + A_i^{(27)} C_0 c_1 C_1 (\pi_0)_{-1} (\pi_i)_{-1}\vac +  T^{(28)}_{ij} C_0 c_1 C_1 (\pi_i)_{-1} (\pi_j)_{-1}\vac.
\end{multline}

The cocycle equation $d\Psi=0$ breaks up into the following components: 
\begin{align*}
  i p_i A_i^{(10)} + 2 \phi^{(12)} -2 \phi^{(13)} &= 0 & \tag{$c_1 C_1 c_{-1} C_{-1}$}\\
  2 \phi^{(2)} + 3 \phi^{(12)} + 3 \phi^{(13)} - 2 i p_i A_i^{(18)} - \phi^{(19)}- T^{(22)}_{ii}  &= 0 & \tag{$C_0 c_1 C_1 c_{-2}$}\\
  - 3 \phi^{(14)} - \phi^{(23)} + \tau^2 T^{(28)}_{ii} &= 0 & \tag{$C_0 c_1 C_1 C_{-2}$}\\
  \phi^{(5)} - \phi^{(15)} - i p_i A_i^{(21)}  &= 0 & \tag{$C_0 c_1 C_1 c_{-1} x^0_{-1}$}\\ 
  A_i^{(6)} - \tau^2 A_i^{(8)} + i p_i \phi^{(13)} - A_i^{(16)} - i p_j T_{ij}^{(22)} &= 0 & \tag{$C_0 c_1 C_1 c_{-1}x^i_{-1}$}\\
  - \tau^2 A_i^{(10)} - i p_i \phi^{(14)} + 2 \tau^2 A_i^{(18)} &= 0 & \tag{$C_0 c_1 C_1 C_{-1} x^i_{-1}$}\\
  \phi^{(3)}  + \phi^{(9)} - 2 \phi^{(17)} - i p_i A_i^{(27)} &= 0 & \tag{$C_0 c_1 C_1 c_{-1} (\pi_0)_{-1}$}\\
   A_i^{(10)} - 2 A_i^{(18)} - 2 i p_j T_{ij}^{(28)} &= 0 & \tag{$C_0 c_1 C_1 c_{-1} (\pi_i)_{-1}$}\\
  \phi^{(5)} - \phi^{(15)} &= 0. & \tag{$C_0 c_1 C_1 C_{-1} (\pi_0)_{-1}$}
\end{align*}

\subsubsection{$\Hrel^3(0)$}
\label{sec:hrel30}

Setting $p=0$ in the cocycle equations, these can be solved as
follows:
\begin{equation}
  \begin{aligned}
    \phi^{(13)} &= \phi^{(12)}\\
    \phi^{(15)} &= \phi^{(5)}\\
    A_i^{(16)} &= A_i^{(6)} - \tau^2 A_i^{(8)}\\
    \phi^{(17)} &= \tfrac12 \phi^{(3)} + \tfrac12 \phi^{(9)}\\
    A_i^{(18)} &= \tfrac12 A^{(10)}\\
    \phi^{(19)} &= 2 \phi^{(2)} + 6 \phi^{(12)} - \Tr T^{(22)}\\
    \phi^{(23)} &= -3 \phi^{(14)} + \tau^2 \Tr T^{(28)}.
  \end{aligned}
\end{equation}
In terms of representations of $\SO(25)$, the relative $3$-cochains
decompose as
\begin{equation}
  \Crel^3(0) \cong 18 \CC \oplus 10 V \oplus 3 \odot^2_0 V \oplus \ext{2} V.
\end{equation}
The cocycle conditions allow us to solve for five scalars $\phi^{(a)}$
for $a = 13, 15, 17, 19, 23$ and two vectors $A_i^{(a)}$ for $a =
16,18$, hence in terms of representations of $\SO(25)$, the relative
$3$-cocycles decompose as
\begin{equation}
  \Zrel^3(0) \cong 13 \CC \oplus 8 V \oplus  3 \odot^2_0 V \oplus
  \ext{2} V
\end{equation}
so that
\begin{equation}
  \Brel^4(0) \cong 5 \CC \oplus 2 V.
\end{equation}
Since $\Brel^3(0) \cong 10 \CC \oplus 6 V \oplus 2 \odot^2_0 V$, we
find that
\begin{equation}
  \Hrel^3(0) \cong 3 \CC \oplus 2 V \oplus  \odot^2_0 V \oplus \ext{2} V.
\end{equation}

\subsubsection{$\Hrel^3(0,\p)$}
\label{sec:hrel30p}

The cocycle conditions are equivalent to
\begin{equation}
  \begin{aligned}
    \phi^{(13)} &= \phi^{(12)} + \tfrac{i}2 p_i A_i^{(10)}\\
    \phi^{(15)} &= \phi^{(5)}\\
    A_i^{(16)} &= A_i^{(6)} - \tau^2 A_i^{(8)} - \tfrac12 p_i p_j A_j^{(10)} + i p_i \phi^{(12)} - i p_j T_{ij}^{(22)}\\
    \phi^{(17)} &= \tfrac12 \phi^{(3)} + \tfrac12 \phi^{(9)} - \tfrac{i}2 p_i A_i^{(27)}\\
    A_i^{(18)} &= \tfrac12 A_i^{(10)} + \tfrac{i}{2\tau^2} p_i \phi^{(14)}\\
    \phi^{(19)} &= 2 \phi^{(2)} + 6 \phi^{(12)} + \tfrac{i}2 p_i A_i^{(10)} + \tfrac{1}{\tau^2} p^2 \phi^{(14)} - \Tr T^{(22)}\\
    p_i A_i^{(21)} &= 0\\
    \phi^{(23)} &= -3 \phi^{(14)} + \tau^2 \Tr T^{(28)}\\
    p_j T_{ij}^{(28)} &= -\tfrac1{2\tau^2} p_i \phi^{(14)},
  \end{aligned}
\end{equation}
where the last equation is imposing the vanishing of a $\CC \oplus
V^\perp$ subspace of $\odot^2 V$.  Counting the representations that
are put to zero, we see that
\begin{equation}
  \Brel^4(0,\p) \cong 9 \CC \oplus 3 V^\perp.
\end{equation}
Since
\begin{equation}
  \Crel^3(0,\p) \cong 31 \CC \oplus 14 V^\perp \oplus 3 \odot^2_0 V^\perp \oplus \ext{2} V^\perp,
\end{equation}
we see that
\begin{equation}
  \Zrel^3(0,\p) \cong 22 \CC \oplus 11 V^\perp \oplus 3 \odot^2_0 V^\perp \oplus \ext{2} V^\perp.
\end{equation}
Since $\Brel^3(0,\p) \cong 19\CC \oplus 8 V^\perp \oplus 2 \odot^2_0 V^\perp$,
we deduce that
\begin{equation}
  \Hrel^3(0,\p) \cong 3 \CC \oplus 3 V^\perp \oplus \odot^2_0 V^\perp \oplus \ext{2} V^\perp.
\end{equation}

\subsection{Calculating $\Hrel^4(p)$}
\label{sec:hrel4}

The general cochain in $\Crel^4(p)$ is given by
\begin{multline}
  \label{eq:Crel4}
  \Psi = \phi^{(1)} c_1 C_1 c_{-1} C_{-1} \vac + \phi^{(2)} C_0 c_1 C_1 c_{-2} \vac + \phi^{(3)} C_0 c_1 C_1 C_{-2} \vac + \phi^{(4)} C_0 c_1 C_1 c_{-1} x^0_{-1}\vac \\
  + A_i^{(5)} C_0 c_1 C_1 c_{-1} x^i_{-1}\vac + \phi^{(6)} C_0 c_1 C_1 C_{-1} x^0_{-1}\vac + A_i^{(7)} C_0 c_1 C_1 C_{-1} x^i_{-1} \vac\\
  + \phi^{(8)} C_0 c_1 C_1 c_{-1} (\pi_0)_{-1}\vac + A_i^{(9)} C_0 c_1 C_1 c_{-1} (\pi_i)_{-1} \vac\\
  + \phi^{(10)} C_0 c_1 C_1 C_{-1} (\pi_0)_{-1}\vac + A_i^{(11)} C_0 c_1 C_1 C_{-1} (\pi_i)_{-1}\vac
\end{multline}
and its differential is easily calculated to give (when $p_0=0$)
\begin{equation}
  d\Psi = \left( - i p_i A_i^{(11)} \right) C_0 c_1 C_1 c_{-1} C_{-1} \vac.
\end{equation}
This gives one cocycle equation:
\begin{equation}
  \label{eq:4-cocycle}
  p_i A_i^{(11)}  = 0.
\end{equation}

\subsubsection{$\Hrel^4(0)$}
\label{sec:hrel40}

In this case, the cocycle condition is trivially satisfied, so that
$\Zrel^4(0) = \Crel^4(0) \cong 7 \CC \oplus 4 V$.  This shows that
$\Brel^5(0) = 0$ and, since $\Brel^4(0) \cong 5 \CC \oplus 2 V$, it
follows that
\begin{equation}
  \Hrel^4(0) \cong 2 \CC \oplus 2 V.
\end{equation}

\subsubsection{$\Hrel^4(0,\p)$}
\label{sec:hrel4p_0p}

The cocycle condition is simply $p_i A_i^{(11)} = 0$, which says that
$A^{(11)}$ is transverse.  Since  $\Crel^4(0,\p) \cong 11 \CC \oplus 4
V^\perp$, it follows that $\Zrel^4(0,\p) \cong 10 \CC \oplus 4
V^\perp$ and hence $\Brel^5(0,\p) \cong \CC$.  Since
$\Brel^4(0,\p)\cong 9\CC \oplus 3 V^\perp$, it
follows that
\begin{equation}
  \Hrel^4(0,\p) \cong \CC \oplus V^\perp.
\end{equation}

\subsection{Calculating $\Hrel^5(p)$}
\label{sec:hrel5}

Since $\Crel^6(p) = 0$, every $5$-cochain is a cocycle.  There is a
one-dimensional space of $5$-cochains, spanned by $C_0 c_1 C_1 c_{-1}
C_{-1}\vac$.  Hence $\Zrel^5(p) \cong \Crel^5(p)$.  Since
$\Brel^5(0)=0$, but $\Brel^5(0,\p) \cong \CC$, it follows that
\begin{equation}
  \Hrel^5(p) \cong
  \begin{cases}
    \CC & p = 0\\
    0 & \text{otherwise}.
  \end{cases}
\end{equation}

\subsection{Summary}
\label{sec:summary}

The results for $p = (0,\p)$, with $\p \neq \bzero$ are summarised in
Table~\ref{tab:Hrel0p}, whereas the results for $p=0$ are summarised
in Table~\ref{tab:Hrelpeqzero}.  As a corollary of the calculation, we
have the following:

\begin{corollary}
  The relative cohomology $\Hrel^\bullet(p)$ displays Poincaré
  duality:
  \begin{equation}
    \Hrel^n(p) \cong \Hrel^{5-n}(p)
  \end{equation}
  as representations of the stabiliser of $p$.
\end{corollary}

\begin{table}[h!]
\centering
\caption{Summary of the calculation of $\Hrel^\bullet(0,\p)$}
\label{tab:Hrel0p}
\resizebox{\linewidth}{!}{
  \begin{tabular}{>{$}c<{$}|*{4}{>{$}l<{$}}}
    \toprule\\
    n & \Crel^n(0,\p) & \Zrel^n(0,\p) & \Brel^n(0,\p) & \Hrel^n(0,\p)\\
    \midrule
    0 & \CC & 0 & 0 & 0 \\
    1 & 11\CC \oplus 4 V^\perp & 2\CC \oplus V^\perp & \CC & \CC \oplus V^\perp\\
    2 & 31\CC \oplus 14 V^\perp \oplus 3 \odot^2_0 V^\perp \oplus \ext{2} V^\perp & 12\CC \oplus 6 V^\perp \oplus \odot^2_0 V^\perp \oplus \ext{2} V^\perp  & 9\CC \oplus 3 V^\perp & 3 \CC \oplus 3 V^\perp \oplus \odot^2_0 V^\perp \oplus \ext{2} V^\perp \\
    3 & 31\CC \oplus 14 V^\perp \oplus 3 \odot^2_0 V^\perp \oplus \ext{2} V^\perp  & 22\CC \oplus 11 V^\perp \oplus 3 \odot^2_0 V^\perp \oplus \ext{2} V^\perp  & 19\CC \oplus 8 V^\perp \oplus 2 \odot^2_0 V^\perp & 3\CC \oplus 3 V^\perp \oplus \odot^2_0 V^\perp \oplus \ext{2} V^\perp \\
    4 & 11\CC \oplus 4 V^\perp & 10\CC \oplus 4 V^\perp & 9 \CC \oplus 3 V^\perp & \CC \oplus V^\perp\\
    5 &\CC & \CC & \CC & 0 \\
    \bottomrule
  \end{tabular}
}
\end{table}

\begin{table}[h!]
  \centering
  \caption{Summary of the calculation of $\Hrel^\bullet(0)$}
  \label{tab:Hrelpeqzero}
  \resizebox{\linewidth}{!}{
  \begin{tabular}{>{$}c<{$}|*{4}{>{$}l<{$}}}
    \toprule\\
    n & \Crel^n(0) & \Zrel^n(0) & \Brel^n(0) & \Hrel^n(0)\\
    \midrule
    0 & \CC & \CC & 0 & \CC \\
    1 & 7\CC \oplus 4 V & 2\CC \oplus 2 V & 0 & 2 \CC \oplus 2 V\\
    2 & 18\CC \oplus 10 V \oplus 3 \odot^2_0 V \oplus \ext{2} V & 8\CC \oplus 4 V \oplus \odot^2_0 V \oplus \ext{2} V  & 5\CC \oplus 2 V & 3\CC \oplus 2 V \oplus \odot^2_0 V \oplus \ext{2} V\\
    3 & 18\CC \oplus 10 V \oplus 3 \odot^2_0 V \oplus \ext{2} V & 13\CC \oplus 8 V \oplus 3 \odot^2_0 V \oplus \ext{2} V  & 10\CC \oplus 6 V \oplus 2 \odot^2_0 V & 3\CC \oplus 2 V \oplus \odot^2_0 V \oplus \ext{2} V\\
    4 & 7\CC \oplus 4 V& 7\CC \oplus 4 V & 5\CC \oplus 2 V & 2 \CC \oplus 2 V\\
    5 &\CC & \CC & 0 & \CC \\
    \bottomrule
  \end{tabular}
}
\end{table}

\section{From relative to absolute cohomology}
\label{sec:from-relat-absol}

As observed earlier, the complexes $\sC^\bullet(p)$ and
$\Crel^\bullet(p)$ are related by a split short exact
sequence~\eqref{eq:ses-brst}.   The celebrated Snake Lemma (see, e.g.,
\cite[Ch.III, §9]{MR1878556}) asserts that a short exact sequence of
complexes induces a long exact sequence in cohomology:
\begin{equation}
  \label{eq:les}
  \begin{tikzcd}
    \cdots \arrow[r] & \Hrel^{n-2}(p) \arrow[r] & \Hrel^n(p) \arrow[r] & \sH^n(p) \arrow[r] & \Hrel^{n-1}(p)
    \arrow[r] & \Hrel^{n+1}(p) \arrow[r] & \cdots
  \end{tikzcd}
\end{equation}
where $\sH^n(p)$ and $\Hrel^n(p)$ are the cohomologies of the absolute and
relative complexes, respectively, at momentum $p$.  The maps $\Hrel^n(p) \to \sH^n(p)$ and $\sH^n(p)
\to \Hrel^{n-1}(p)$ are induced by the inclusion and $b_0$ respectively,
whereas the connecting homomorphism $\Hrel^{n-1}(p) \to \Hrel^{n+1}(p)$
deserves further attention.  Let $\Psi \in \Crel^{n-1}(p)$ be a cocycle,
so $d\Psi = 0$.  Then $d c_0 \Psi \in \sC^{n+1}(p)$ lies in the kernel of
$b_0$ and hence lies in $\Crel^{n+1}(p)$:
\begin{align*}
  b_0 d c_0 \Psi &= [b_0,d] c_0 \Psi - d b_0 c_0 \Psi \\
                 &= \Ltot_0 c_0\Psi - d \Psi \\
                 &= 0. & \tag{since $c_0 \Psi \in \ker \Ltot_0$ and
                        $d\Psi = 0$}
\end{align*}
Moreover, since $d^2 =0$, it follows that $d c_0 \Psi$ is a cocycle of
the relative subcomplex.  Although it is clearly a coboundary in the
absolute complex, that need not be the case in the relative
subcomplex.  We can calculate the connecting homomorphism directly.
Since $d\Psi = 0$, it follows that $d c_0 \Psi = [d,c_0] \Psi$ and we
can calculate $[d,c_0]$ in a way similar to how we proved
Lemma~\ref{lem:dC0}:
\begin{equation}
  [d,c_0] = [J_0,c_0] = \left( [J,c]_1 \right)_0.
\end{equation}
It follows by calculation that $[J,c]_1 = c\d c$ and hence $[d,c_0] =
(c \d c)_0$.  As in the proof of Lemma~\ref{lem:dC0}, we find that
\begin{equation}
  (c \d c)_0 = \sum_{m \in \ZZ} (1-m) c_{-m} c_m,
\end{equation}
which, acting on the relative subcomplex, reduces to
\begin{equation}
  (c \d c)_0 = 4 c_2 c_{-2} + 2 c_1 c_{-1}
\end{equation}
and hence the connecting homomorphism is given at the level of
cocycles by
\begin{equation}
  \label{eq:conn-hom}
  \Psi \mapsto d c_0 \Psi = (4 c_2 c_{-2} + 2 c_1 c_{-1}) \Psi.
\end{equation}

From our calculations in Section~\ref{sec:relat-brst-cohom}, we can
already deduce some consequences from the long exact
sequence~\eqref{eq:les}.  Since $\Hrel^\bullet(p) =0$ for $p_0 \neq
0$, the long exact sequence decomposes into segments of the form
\begin{equation}
  \begin{tikzcd}
    0 \arrow[r] & \sH^n(p) \arrow[r] & 0
  \end{tikzcd}
\end{equation}
and hence we have
\begin{proposition}
  \begin{equation}
    \sH^\bullet(p) = 0 \quad \text{for $p_0 \neq 0$.}
  \end{equation}
\end{proposition}

Since $\Hrel^{n<0}(p)=0$, taking $n=0$ in the long exact
sequence~\eqref{eq:les} gives that $\sH^0(p) \cong \Hrel^0(p)$, so that
\begin{equation}
  \sH^0(p) \cong
  \begin{cases}
    \CC & p = 0\\
    0 & p \neq 0.
  \end{cases}
\end{equation}
Similarly, since $\Hrel^{n>5}(p)=0$, taking $n=6$ in the long exact
sequence~\eqref{eq:les} gives that $\sH^6(p) \cong \Hrel^5(p)$, so
that
\begin{equation}
  \sH^6(p) \cong
  \begin{cases}
    \CC & p = 0\\
    0 & p \neq 0,
  \end{cases}
\end{equation}
which already hints at the persistence of the Poincaré duality of the
relative cohomology in the absolute cohomology.

To make further progress in the calculation of $\sH^\bullet(p)$, we
need to look more closely at the cocycle condition.  In the sequel we
will use the notation $\sZ^\bullet(p)$, $\sB^\bullet(p)$ for the
cocycles and coboundaries of the absolute complex, respectively.  Let
$\Psi \in \sC^n(p)$ for $n=1,\dots,5$.  Then $\Psi = \psi + c_0
\zeta$, where $\psi \in \Crel^n(p)$ and $\zeta \in \Crel^{n-1}(p)$.
The cocycle condition $d\Psi = 0$ becomes
\begin{equation}
  d \Psi = d \psi + [d,c_0] \zeta - c_0 d\zeta = 0,
\end{equation}
from where it follows that the last term is the only one involving
$c_0$ and hence it has to vanish on its own.  Indeed, if $c_0 d\zeta =
0$, then also $b_0 c_0 d\zeta = 0$, which using that $b_0 d\zeta = 0$,
results in $d\zeta = [b_0,c_0]d\zeta = 0$ and hence $\zeta \in
\Zrel^{n-1}(p)$.  However, $\psi$ is not necessarily a cocycle, instead
it obeys
\begin{equation}
  d\psi = - \Delta \zeta,
\end{equation}
where $\Delta := 4c_2 c_{-2} + 2 c_1 c_{-1}$ is the connecting
homomorphism given in equation~\eqref{eq:conn-hom}.  Our strategy is
then as follows.  We work ghost number by ghost number and determine
the image of the connecting homomorphism acting on relative cocycles
at ghost number $n-1$ and explore how, if at all, this changes the
cocycle conditions at ghost number $n$.  In many cases this will turn
out to have no effect and in those cases
$\sZ^n(p) \cong \Zrel^n(p) \oplus \Zrel^{n-1}(p)$.  In a small number
of cases there will be a change to the cocycle conditions at ghost
number $n$.

\subsection{Calculating $\sH^1(p)$}
\label{sec:H1p}

We first calculate the action of the connecting homomorphism $\Delta$
on $\Crel^0(p)$:
\begin{equation}
  \Delta \vac = 2 c_1 c_{-1} \vac.
\end{equation}

\subsubsection{Calculating $\sH^1(0,\p)$}
\label{sec:H10p}

Since $\p \neq \bzero$, $\Zrel^0(0,\p) = 0$ and hence $\sZ^1(0,\p) =
\Zrel^1(0,\p)$.  Since $\sB^1(0,\p) = \Brel^1(0,\p)$,  it follows that
\begin{equation}
  \sH^1(0,\p) \cong \Hrel^1(0,\p) \cong \CC \oplus V^\perp.
\end{equation}
We record for later use that since $\sC^1(0,\p) \cong \Crel^1(0,\p)
\oplus \Crel^0(0,\p)$,
\begin{equation}
  \sC^1(0,\p) \cong 12 \CC \oplus 4 V^\perp
\end{equation}
and hence that
\begin{equation}
  \sB^2(0,\p) \cong 10 \CC \oplus 3 V^\perp.
\end{equation}

\subsubsection{Calculating $\sH^1(0)$}
\label{sec:H1pzero}

If $p=0$, $\Zrel^0(0) = \CC \ket{0}$ and $\Brel^1(0)=0$.  The cocycle
condition for $\psi + c_0 \zeta \ket{0} \in \sZ^1(0)$ gets modified.  In the
notation of equation~\eqref{eq:Crel1}, the component equation along $c_1 c_{-1}
\ket{0}$ is modified so that now $\phi^{(11)} + 2 \zeta = 0$, but the
component equation along $c_1 C_1 x^i_{-1} x^j_{-1}\ket{0}$ still says
that $\phi^{(11)} = 0$, so that $\zeta = 0$ as well.  In summary,
$\sZ^1(0) = \Zrel^1(0)$ and since $\sB^1(0) = \Brel^1(0) = 0$, it
follows that
\begin{equation}
  \sH^1(0) \cong \Hrel^1(0) \cong 2 \CC \oplus 2 V.
\end{equation}
We record for later use that since $\sC^1(0) \cong \Crel^1(0) \oplus
\Crel^0(0)$,
\begin{equation}
  \sC^1(0) \cong 8 \CC \oplus 4 V
\end{equation}
and hence
\begin{equation}
  \sB^2(0) \cong 6 \CC \oplus 2 V.
\end{equation}

\subsection{Calculating $\sH^2(p)$}
\label{sec:H2p}

We first calculate the action of the connecting homomorphism $\Delta$
on $\Crel^1(p)$ and only then specialise its action on cocycles.  The
general relative $1$-cochain was given in equation~\eqref{eq:Crel1},
but we will change the notation of the coefficients:
\begin{multline}
  \Psi = \psi^{(1)} C_0 \vac + \psi^{(2)} c_1 x^0_{-1} \vac + \chi^{(3)}_i c_1 x^i_{-1} \vac + \psi^{(4)} C_1 x^0_{-1}\vac\\
  + \chi^{(5)}_i C_1 x^i_{-1}\vac + \psi^{(6)} c_1 (\pi_0)_{-1}\vac + \chi^{(7)}_i c_1 (\pi_i)_{-1} \vac + \psi^{(8)} C_1 (\pi_0)_{-1}\vac\\
+ \chi^{(9)}_i C_1 (\pi_i)_{-1} \vac + \psi^{(10)} c_1 C_1 b_{-2} \vac + \psi^{(11)} c_1 C_1 B_{-2} \vac.
\end{multline}
We calculate
\begin{multline}
  \Delta \Psi = 2 \psi^{(1)} C_0 c_1 c_{-1} \vac - 2 \psi^{(4)} c_1 C_1  c_{-1} x^0_{-1}\vac - 2 \chi^{(5)}_i c_1 C_1 c_{-1} x^i_{-1}\vac\\
  - 2 \psi^{(8)} c_1 C_1 c_{-1} (\pi_0)_{-1}\vac - 2 \chi^{(9)}_i c_1 C_1 c_{-1} (\pi_i)_{-1} \vac - 4 \psi^{(10)} c_1 C_1 c_{-2} \vac.
\end{multline}

\subsubsection{Calculating $\sH^2(0,\p)$}
\label{sec:H20p}

The most general relative $1$-cocycle is of the form
\begin{equation}
  \zeta = \chi^{(3)}_i c_1 x^i_{-1} \ket{0,\p} + \psi^{(6)} c_1
  (\pi_0)_{-1} \ket{0,\p},
\end{equation}
and hence $\Delta \zeta = 0$.  Therefore $\sZ^2(0,\p) \cong
\Zrel^2(0,\p) \oplus \Zrel^1(0,\p)$ and hence
\begin{equation}
  \sZ^2(0,\p) \cong 14 \CC \oplus 7 V^\perp \oplus \odot^2_0 V^\perp
  \oplus \ext{2} V^\perp.
\end{equation}
Since $\sC^2(0,\p) \cong \Crel^2(0,\p) \oplus \Crel^1(0,\p)$, it
follows that
\begin{equation}
  \sC^2(0,\p) \cong 42\CC \oplus 18 V^\perp \oplus 3 \odot^2_0 V^\perp
  \oplus \ext{2} V^\perp
\end{equation}
and hence $\sB^3(0,\p) \cong 28 \CC \oplus 11 V^\perp \oplus 2
\odot^2_0 V^\perp$.  Finally, using that $\sB^2(0,\p) \cong 10 \CC
\oplus 3 V^\perp$, we find that
\begin{equation}
  \sH^2(0,\p) \cong 4\CC \oplus 4 V^\perp \oplus \odot^2_0 V^\perp
  \oplus \ext{2} V^\perp.
\end{equation}

\subsubsection{Calculating $\sH^2(0)$}
\label{sec:H2pzero}

Let $\zeta \in \Zrel^1(0)$ be a general cocycle.  Its image under the
connecting homomorphism is
\begin{equation}
  \Delta \zeta = -2 \psi^{(2)} c_1 C_1 c_{-1} (\pi_0)_{-1} \ket{0} - 2 \chi^{(25)}_i c_1 C_1 c_{-1} x^i_{-1} \ket{0}.
\end{equation}
Two of the cocycle conditions in $\Crel^2(0)$ are modified:
\begin{equation}
  \begin{aligned}
    2 \phi^{(7)} + 2 \phi^{(17)} - 2 \psi^{(2)} &= 0\\
    2 A_i^{(4)} + A_i^{(16)} - 2 \chi^{(5)} &= 0,  
  \end{aligned}
\end{equation}
but, since $\phi^{(17)}$ and $A_i^{(16)}$ do not appear in any other
cocycle condition, we may solve for them, resulting in
\begin{equation}
  \sZ^2(0) \cong \Zrel^2(0) \oplus \Zrel^1(0) \cong 10 \CC \oplus 6 V
  \oplus \odot^2_0 V \oplus \ext{2} V.
\end{equation}
Since $\sC^2(0) \cong 25\CC \oplus 14 V \oplus 3 \odot^2_0 V \oplus
\ext{2} V$, we see that
\begin{equation}
  \sB^3(0) \cong 15 \CC \oplus 8 V \oplus 2 \odot^2_0 V.
\end{equation}
Finally, since $\sB^2(0) \cong  6 \CC \oplus 2 V$, we have that
\begin{equation}
  \sH^2(0) \cong 4 \CC \oplus 4 V \oplus \odot^2_0 V \oplus \ext{2} V.
\end{equation}

\subsection{Calculating $\sH^3(p)$}
\label{sec:H3p}

We work out the action of the connecting homomorphism on $\Crel^2(p)$
and only then specialise to the relevant cocycles.  The general
cochain in $\Crel^2(p)$ is given in equation~\eqref{eq:Crel2}, but we
change the notation of the coefficients:
\begin{multline}
  \Psi = \psi^{(1)} C_0 c_1 C_1 b_{-2} \vac + \psi^{(2)} C_0 c_1 C_1 B_{-2} \vac + \psi^{(3)} C_0 c_1 x^0_{-1} \vac + \chi_i^{(4)} C_0 c_1 x^i_{-1}\vac + \psi^{(5)} C_0 C_1 x^0_{-1}\vac\\
  + \chi^{(6)}_i C_0 C_1 x^i_{-1} \vac + \psi^{(7)} C_0 c_1 (\pi_0)_{-1} \vac + \chi^{(8)}_i C_0 c_1 (\pi_i)_{-1} \vac + \psi^{(9)} C_0 C_1 (\pi_0)_{-1}\vac\\ + \chi^{(10)}_i C_0 C_1 (\pi_i)_{-1}\vac + \psi^{(11)} c_1 c_{-1} \vac + \psi^{(12)} c_1 C_{-1} \vac + \psi^{(13)} c_{-1} C_1 \vac + \psi^{(14)} C_1 C_{-1} \vac\\
  + \psi^{(15)} c_1 C_1 x^0_{-2} \vac + \chi_i^{(16)} c_1 C_1 x^i_{-2} \vac + \psi^{(17)} c_1 C_1 (\pi_0)_{-2}\vac + \chi_i^{(18)} c_1 C_1 (\pi_i)_{-2} \vac\\
  + \psi^{(19)} c_1 C_1 x^0_{-1} (\pi_0)_{-1} \vac + \psi^{(20)} c_1 C_1 (x^0_{-1})^2 \vac + \psi^{(21)} c_1 C_1 (\pi_0)_{-1}^2 \vac + \chi_i^{(22)} c_1 C_1 x^i_{-1} (\pi_0)_{-1} \vac\\
  + \chi_i^{(23)} c_1 C_1 x^0_{-1} (\pi_i)_{-1} \vac + \chi_i^{(24)} c_1 C_1 x^i_{-1} x^0_{-1}\vac + \chi_i^{(25)} c_1 C_1 (\pi_i)_{-1} (\pi_0)_{-1} \vac \\
  + \theta_{ij}^{(26)} c_1 C_1 x^i_{-1} x^j_{-1} \vac + \theta_{ij}^{(27)} c_1 C_1 x^i_{-1} (\pi_j)_{-1} \vac + \theta_{ij}^{(28)} c_1 C_1 (\pi_i)_{-1} (\pi_j)_{-1} \vac,
\end{multline}
where $\theta_{ij}^{(26)} = \theta_{ji}^{(26)}$ and
$\theta_{ij}^{(28)} = \theta_{ji}^{(28)}$.  We calculate
\begin{multline}
  \Delta\Psi = - 4 \psi^{(1)} C_0 c_1 C_1 c_{-2} \vac - 2 \psi^{(5)} C_0 c_1 C_1 c_{-1} x^0_{-1}\vac - 2 \chi^{(6)}_i C_0 c_1 C_1 c_{-1} x^i_{-1} \vac\\ - 2 \psi^{(9)} C_0 c_1 C_1 c_{-1} (\pi_0)_{-1}\vac - 2 \chi^{(10)}_i C_0 c_1 C_1 c_{-1} (\pi_i)_{-1}\vac - 2  \psi^{(14)} c_1 C_1 c_{-1} C_{-1} \vac
\end{multline}

\subsubsection{Calculating $\sH^3(0,\p)$}
\label{sec:H30p}

The action of the connecting homomorphism on $\zeta \in \Crel^2(0,\p)$ is given by
\begin{equation}
  \Delta \zeta = 2 \psi^{(3)} C_0 c_1 C_1 c_{-1} (\pi_0)_{-1}\ket{0,\p} - 2 \tau^2 \chi_i^{(8)} C_0 c_1 C_1 c_{-1} x^i_{-1} \ket{0,\p}.
\end{equation}
This modifies two of the cocycle conditions on $\Crel^3(0,\p)$:
\begin{equation}
  \begin{aligned}
    \phi^{(7)} &= \tfrac12 \phi^{(3)} + \tfrac12 \phi^{(9)} - \tfrac{i}2 p_i A_i^{(27)} + 2 \psi^{(3)}\\
    A_i^{(16)} &= A_i^{(6)} - \tau^2 A_i^{(8)} + i p_i \phi^{(13)} - i p_j T_{ij}^{(22)} - 2 \tau^2 \chi_i^{(8)}.
  \end{aligned}
\end{equation}
Neither $\phi^{(7)}$ nor $A_i^{(16)}$ appear in other cocycle
conditions and therefore it is still the case that $\sZ^3(0,\p) \cong
\Zrel^3(0,\p) \oplus \Zrel^2(0,\p)$, or explicitly,
\begin{equation}
  \sZ^3(0,\p) \cong 34 \CC \oplus 17 V^\perp \oplus 4 \odot^2_0
  V^\perp \oplus 2 \ext{2} V^\perp.
\end{equation}
Since
\begin{equation}
  \sC^3(0,\p) \cong \Crel^3(0,\p) \oplus \Crel^2(0,\p) \cong 62 \CC
  \oplus 28 V^\perp \oplus 6 \odot^2_0 V^\perp \oplus 2 \ext{2} V^\perp,
\end{equation}
it follows that
\begin{equation}
  \sB^4(0,\p) \cong 28 \CC \oplus 11 V^\perp \oplus 2 \odot^2_0 V^\perp.
\end{equation}
Since $\sB^3(0,\p) \cong \sB^2(0,\p)$, we conclude that
\begin{equation}
  \sH^3(0,\p) \cong 6 \CC \oplus 6 V^\perp \oplus 2 \odot^2_0 V^\perp
  \oplus 2 \ext{2} V^\perp.
\end{equation}

\subsubsection{Calculating $\sH^3(0)$}
\label{sec:H3pzero}

The action of the connecting homomorphism on $\zeta \in \Crel^2(0)$ is
given by
\begin{equation}
  \Delta\zeta = -4 \psi^{(1)} C_0 c_1 C_1 c_{-2} \ket{0} + 6
  \psi^{(1)} c_1 C_1 c_{-1} C_{-1} \ket{0} + 2 \psi^{(3)} C_0 c_1 C_1
  c_{-1} (\pi_0)_{-1}\ket{0},
\end{equation}
and this modifies three cocycle conditions in $\Crel^3(0)$:
\begin{equation}
  \begin{aligned}
    \phi^{(19)}&= 2 \phi^{(2)} + 6 \phi^{(12)} - \Tr T^{(22)} + 5 \psi^{(1)}\\
    \phi^{(13)}&= \phi^{(12)} + 3 \psi^{(1)}\\
    \phi^{(17)}&= \tfrac12 \phi^{(3)} + \tfrac12 \phi^{(9)} + \psi^{(3)}.
  \end{aligned}
\end{equation}
Since $\phi^{(19)}, \phi^{(13)}, \phi^{(17)}$ do not appear in any
other cocycle relation, we again have that $\sZ^3(0) \cong \Zrel^3(0)
\oplus \Zrel^2(0)$, or explicitly,
\begin{equation}
  \sZ^3(0) \cong 21 \CC \oplus 12 V \oplus 4 \odot^2_0 V \oplus
  2 \ext{2} V.
\end{equation}
Since
\begin{equation}
  \sC^3(0) \cong \Crel^3(0) \oplus \Crel^2(0) \cong 36 \CC \oplus 20 V
  \oplus 6 \odot^2_0 V \oplus 2 \ext{2} V,
\end{equation}
it follows that
\begin{equation}
  \sB^4(0) \cong 15 \CC \oplus 8 V \oplus 2 \odot^2_0 V.
\end{equation}
Since $\sB^3(0) \cong \sB^4(0)$, it follows finally that
\begin{equation}
  \sH^3(0) \cong 6 \CC \oplus 4 V \oplus 2 \odot^2_0 V \oplus 2
  \ext{2} V.
\end{equation}

\subsection{Calculating $\sH^4(p)$}
\label{sec:H4p}

We again work out the action of the connecting homomorphism on
$\Crel^3(p)$ before specialising to its action on $\Zrel^3(p)$.  The
general cochain in $\Crel^3(p)$ is given in equation~\eqref{eq:Crel3},
but we change the notation of the coefficients to give
\begin{multline}
  \Psi = \psi^{(1)} c_1 C_1 c_{-2}\vac + \psi^{(2)} c_1 C_1 C_{-2}\vac + \psi^{(3)} c_1 C_1 c_{-1} x^0_{-1}\vac + \chi_i^{(4)} c_1 C_1 c_{-1} x^i_{-1}\vac \\
  + \psi^{(5)} c_1 C_1 C_{-1} x^0_{-1} \vac + \chi_i^{(6)} c_1 C_1 C_{-1} x^i_{-1}\vac + \psi^{(7)} c_1 C_1c_{-1} (\pi_0)_{-1}\vac + \chi_i^{(8)} c_1 C_1 c_{-1} (\pi_i)_{-1}\vac\\
  + \psi^{(9)} c_1 C_1 C_{-1} (\pi_0)_{-1}\vac  + \chi_i^{(10)} c_1 C_1 C_{-1} (\pi_i)_{-1}\vac + \psi^{(11)} C_0 c_1 c_{-1}\vac + \psi^{(12)} C_0 c_1 C_{-1}\vac \\
  + \psi^{(13)} C_0 c_{-1} C_1\vac + \psi^{(14)} C_0 C_1 C_{-1}\vac  + \psi^{(15)} C_0 c_1 C_1 x^0_{-2}\vac + \chi_i^{(16)} C_0 c_1 C_1 x^i_{-2}\vac\\
  + \psi^{(17)} C_0 c_1 C_1 (\pi_0)_{-2}\vac + \chi_i^{(18)} C_0 c_1 C_1 (\pi_i)_{-2}\vac + \psi^{(19)} C_0 c_1 C_1 x^0_{-1} (\pi_0)_{-1}\vac\\
  + \chi_i^{(20)} C_0 c_1 C_1 x^i_{-1} (\pi_0)_{-1}\vac + \chi_i^{(21)} C_0 c_1 C_1 x^0_{-1} (\pi_i)_{-1}\vac + \theta_{ij}^{(22)} C_0 c_1 C_1 x^i_{-1} (\pi_j)_{-1}\vac\\
  + \psi^{(23)} C_0 c_1 C_1 (x^0_{-1})^2\vac+ \chi_i^{(24)} C_0 c_1 C_1 x^0_{-1} x^i_{-1}\vac + \theta_{ij}^{(25)} C_0 c_1 C_1 x^i_{-1} x^j_{-1}\vac\\
  + \psi^{(26)} C_0 c_1 C_1 (\pi_0)_{-1}^2\vac + \chi_i^{(27)} C_0 c_1 C_1 (\pi_0)_{-1} (\pi_i)_{-1}\vac +  \theta^{(28)}_{ij} C_0 c_1 C_1 (\pi_i)_{-1} (\pi_j)_{-1}\vac,
\end{multline}
so that
\begin{equation}
  \Delta \Psi = -2 \psi^{(14)} C_0 c_1 C_1 c_{-1} C_{-1} \vac.
\end{equation}

\subsubsection{Calculating $\sH^4(0,\p)$}
\label{sec:H40p}

Only one cocycle equation in $\Crel^4(0,\p)$ is modified: namely,
\begin{equation}
  i p_i A_i^{(11)} = 2 \psi^{(14)}.
\end{equation}
This does not modify the counting and hence
\begin{equation}
  \sZ^4(0,\p) \cong \Zrel^4(0,\p) \oplus \Zrel^3(0,\p) \cong 32 \CC
  \oplus 15 V^\perp \oplus 3 \odot^2_0 V^\perp \oplus \ext{2} V^\perp.
\end{equation}
Since
\begin{equation}
  \sC^4(0,\p) \cong \Crel^4(0,\p) \oplus \Crel^3(0,\p) \cong 42 \CC
  \oplus 18 V^\perp \oplus 3 \odot^2_0 V^\perp \oplus \ext{2} V^\perp,
\end{equation}
we deduce that
\begin{equation}
  \sB^5(0,\p) \cong 10 \CC \oplus 3 V^\perp
\end{equation}
and since $\sB^4(0,\p) \cong 28 \CC \oplus 11 V^\perp \oplus 2 \odot^2_0
V^\perp$, we finally have that
\begin{equation}
  \sH^4(0,\p) \cong 4 \CC \oplus 4  V^\perp \oplus \odot^2_0 V^\perp
  \oplus \ext{2} V^\perp.
\end{equation}

\subsubsection{Calculating $\sH^4(0)$}
\label{sec:H4pzero}

Let $\zeta \in \Zrel^3(0)$ so that
\begin{equation}
  \Delta \zeta = - 2 \psi^{(14)} C_0 c_1 C_1 c_{-1} C_{-1} \ket{0}.
\end{equation}
However for all $\psi \in \Crel^4(0)$, $d\psi = 0$, hence the cocycle
condition in $\sC^4(0)$ says that $\psi^{(14)} = 0$.  In other words,
\begin{equation}
  \sZ^4(0) \cong \Zrel^4(0) \oplus \Zrel^3(0) \ominus \CC \cong 19 \CC
  \oplus 12 V \oplus 3 \odot^2_0 V \oplus \ext{2} V.
\end{equation}
Since
\begin{equation}
  \sC^4(0) \cong \Crel^4(0) \oplus \Crel^3(0) \cong 25 \CC
  \oplus 14 V \oplus 3 \odot^2_0 V \oplus \ext{2} V,
\end{equation}
we have that
\begin{equation}
  \sB^5(0) \cong 6 \CC \oplus 2 V.
\end{equation}
Since $\sB^4(0) \cong 15 \CC \oplus 8 V \oplus 2 \odot^2_0 V$, we
finally have that
\begin{equation}
  \sH^4(0) \cong 4 \CC \oplus 4 V \oplus \odot^2_0 V \oplus \ext{2} V.
\end{equation}

\subsection{Calculating $\sH^5(p)$}
\label{sec:H5p}

It is clear from the expression for the general cochain in
$\Crel^4(p)$ given by equation~\eqref{eq:Crel4} that for all $\zeta
\in \Crel^4(p)$, $\Delta \zeta = 0$.  Therefore it is always the case
that $\sZ^5(p) \cong \Zrel^5(p) \oplus \Zrel^4(p)$.

\subsubsection{Calculating $\sH^5(0,\p)$}
\label{sec:H50p}

We see that
\begin{equation}
  \sZ^5(0,\p) \cong \Zrel^5(0,\p) \oplus \Zrel^4(0,\p) \cong 11 \CC
  \oplus 4 V^\perp.
\end{equation}
Since $\sB^5(0,\p) \cong 10 \CC \oplus 3 V^\perp$, we conclude that
\begin{equation}
  \sH^5(0,\p) \cong \CC \oplus V^\perp.
\end{equation}

\subsubsection{Calculating $\sH^5(0)$}
\label{sec:H5pzero}

Now we have that
\begin{equation}
  \sZ^5(0) \cong \Zrel^5(0) \oplus \Zrel^4(0) \cong 8 \CC \oplus 4 V.
\end{equation}
Since $\sB^5(0) \cong 6 \CC \oplus 2 V$, we conclude that
\begin{equation}
  \sH^5(0) \cong 2 \CC \oplus 2 V.
\end{equation}

\subsection{Summary}
\label{sec:summary-absolute-cohomology}

We summarise the results of the calculation of $\sH^\bullet(p)$ in two
tables.  Table~\ref{tab:Habs0p} summarises the results for $p =
(0,\p)$, with $\p \neq \bzero$, whereas Table~\ref{tab:Habspeqzero}
summarises the results for $p=0$.   The connecting homomorphisms are
trivial for $\sH^\bullet(0,\p)$.  Both $\sH^\bullet(0,\p)$  and
$\sH^\bullet(0)$ display Poincaré duality.  Moreover the Euler
characteristic of the absolute cohomology, which we define for each
isotypical representation as the alternating sum of its
multiplicities, is identically zero; to wit, for $\sH^\bullet(0,\p)$,
\begin{equation}
  \begin{aligned}
    (\CC~ \text{and}~ V^\perp) \qquad & \qquad 0 - 1 + 4 - 6 + 4 - 1 + 0 = 0\\
    (\odot^2_0V^\perp~\text{and}~ \ext{2} V^\perp) \qquad & \qquad 0 - 0 + 1 - 2 + 1 - 0 + 0 = 0,
  \end{aligned}
\end{equation}
whereas for $\sH^\bullet(0)$,
\begin{equation}
  \begin{aligned}
    (\CC) \qquad & \qquad 1 - 2 + 4 - 6 + 4 - 2 + 1 = 0\\
    (V) \qquad & \qquad 0 - 2 + 4 - 4 + 4 - 2 + 0 = 0\\
    (\odot^2_0V~\text{and}~ \ext{2} V) \qquad & \qquad 0 - 0 + 1 - 2 + 1 - 0 + 0 = 0.
  \end{aligned}
\end{equation}

\begin{table}[h!]
\centering
\caption{Summary of the calculation of $\sH^\bullet(0,\p)$}
\label{tab:Habs0p}
\resizebox{\linewidth}{!}{
  \begin{tabular}{>{$}c<{$}|*{4}{>{$}l<{$}}}
    \toprule\\
    n & \sC^n(0,\p) & \sZ^n(0,\p) & \sB^n(0,\p) & \sH^n(0,\p)\\
    \midrule
    0 & \CC & 0 & 0 & 0 \\
    1 & 12\CC \oplus 4 V^\perp & 2\CC \oplus V^\perp & \CC & \CC \oplus V^\perp\\
    2 & 42\CC \oplus 18 V^\perp \oplus 3 \odot^2_0 V^\perp \oplus \ext{2} V^\perp & 14\CC \oplus 7 V^\perp \oplus \odot^2_0 V^\perp \oplus \ext{2} V^\perp  & 10\CC \oplus 3 V^\perp & 4 \CC \oplus 4 V^\perp \oplus \odot^2_0 V^\perp \oplus \ext{2} V^\perp \\
    3 & 62\CC \oplus 28 V^\perp \oplus 6 \odot^2_0 V^\perp \oplus 2 \ext{2} V^\perp  & 34\CC \oplus 17 V^\perp \oplus 4 \odot^2_0 V^\perp \oplus 2 \ext{2} V^\perp  & 28\CC \oplus 11 V^\perp \oplus 2 \odot^2_0 V^\perp & 6\CC \oplus 6 V^\perp \oplus 2 \odot^2_0 V^\perp \oplus 2 \ext{2} V^\perp \\
    4 & 42\CC \oplus 18 V^\perp \oplus 3 \odot^2_0 V^\perp \oplus \ext{2} V^\perp & 32\CC \oplus 15 V^\perp \oplus 3 \odot^2_0 V^\perp \oplus \ext{2} V^\perp  & 28\CC \oplus 11 V^\perp \oplus 2  \odot^2_0 V^\perp & 4 \CC \oplus 4 V^\perp \oplus \odot^2_0 V^\perp \oplus \ext{2} V^\perp \\
    5 & 12\CC \oplus 4 V^\perp & 11\CC \oplus 4 V^\perp & 10 \CC \oplus 3 V^\perp & \CC \oplus V^\perp\\
    6 & \CC & \CC & \CC & 0 \\
    \bottomrule
  \end{tabular}
}
\end{table}

\begin{table}[h!]
  \centering
  \caption{Summary of the calculation of $\sH^\bullet(0)$}
  \label{tab:Habspeqzero}
  \resizebox{\linewidth}{!}{
  \begin{tabular}{>{$}c<{$}|*{4}{>{$}l<{$}}}
    \toprule\\
    n & \sC^n(0) & \sZ^n(0) & \sB^n(0) & \sH^n(0)\\
    \midrule
    0 & \CC & \CC & 0 & \CC \\
    1 & 8\CC \oplus 4 V & 2\CC \oplus 2 V & 0 & 2 \CC \oplus 2 V\\
    2 & 25\CC \oplus 14 V \oplus 3 \odot^2_0 V \oplus \ext{2} V &
                                                                   10\CC \oplus 6 V \oplus \odot^2_0 V \oplus \ext{2} V  & 6\CC \oplus 2 V & 4 \CC \oplus 4 V \oplus \odot^2_0 V \oplus \ext{2} V\\
    3 & 36\CC \oplus 20 V \oplus 6 \odot^2_0 V \oplus 2 \ext{2} V & 21\CC \oplus 12 V \oplus 4 \odot^2_0 V \oplus 2 \ext{2} V  & 15\CC \oplus 8 V \oplus 2 \odot^2_0 V & 6\CC \oplus 4 V \oplus 2 \odot^2_0 V \oplus 2 \ext{2} V\\
    4 & 25\CC \oplus 14 V \oplus 3 \odot^2_0 V \oplus \ext{2} V & 19\CC \oplus 12 V \oplus 3 \odot^2_0 V \oplus \ext{2} V & 15 \CC \oplus 8 V \oplus 2 \odot^2_0 V & 4 \CC \oplus 4 V \oplus \odot^2_0 V \oplus \ext{2} V \\
    5 &8 \CC \oplus 4 V & 8 \CC \oplus 4 V & 6 \CC \oplus 2 V & 2 \CC \oplus 2 V \\
    6 &\CC & \CC & 0 & \CC \\
    \bottomrule
  \end{tabular}
}
\end{table}

\subsection{Interpretation in terms of Carroll UIRs}
\label{sec:interpr-terms-carr}

The Carroll group $G$ (in any dimension) is a semi-direct product $G =
K \ltimes T$, where (here, in 26 dimensions) $K \cong \ISO(25)$ and
$T \cong \RR^{26}$, and hence its unitary irreducible representations
(UIRs) are constructed using the Mackey--Wigner method.  This has been
described in detail for the Carroll group (in four dimensions) in
\cite{Figueroa-OFarrill:2023qty}.  In a nutshell, UIRs are described
as square-integrable sections of homogeneous vector bundles on the
$K$-orbits $\eO_\tau$ in the unitary dual of $T$, which since $T$ is
abelian is the dual $\ft^*$ of the Lie algebra $\ft$ of $T$.  These
homogeneous vector bundles are associated to the representations of
the generic stabiliser $G_\tau$ of the $K$-orbit $\eO_\tau$.  Here,
$\tau = (p_0,\p)$ and the stabiliser is either $\SO(25) \ltimes
\RR^{25}$ for $p=0$ or $\SO(24) \ltimes \RR^{25}$ for $p = (0,\p)$.

The representations corresponding to $\sH^\bullet(0)$ are
finite-dimensional UIRs of the Carroll group where all generators
except the rotations act trivially.  The orbit $\eO_{(0,\bzero)}$ is a
point and hence the UIRs coincide with the inducing UIRs of $\SO(25)$
appearing in $\sH^\bullet(0)$.

The representations corresponding to $\sH^\bullet(0,\p)$ are
infinite-dimensional UIRs of the Carroll group induced from
finite-dimensional UIRs of $\SO(24)$, so that the boosts act
trivially.  The orbit $\eO_{(0,\p)}$ is the $24$-dimensional sphere in
$\RR^{25}$ consisting of momenta $\p$ with a fixed (positive) norm
$\|\p\|$.  The UIRs are realised as square-integrable sections of
homogeneous vector bundles $\SO(25) \times_{\SO(24)} \sH^\bullet(0,\p)
\longrightarrow \eO_{(0,\p)}$ over the $24$-dimensional sphere.  These
are the 26-dimensional analogues of the aristotelion UIRs in
\cite{Figueroa-OFarrill:2023qty}.

\subsection{Interpretation in terms of carrollian geometry}
\label{sec:interpr-terms-carr-1}

It is possible to give a geometric interpretation of most of the
cohomology in terms of deformations of the carrollian structure of the
Carroll spacetime.  Since the ghost zero modes $c_0,C_0$ are responsible for
a four-fold degeneracy in the cohomology --- the different ``pictures''
familiar from usual string theory --- it is enough to concentrate on
$\sH^2(0,\p)$ for $\p \neq \bzero$.

From the above results $\sH^2(0,\p) \cong 4 \CC \oplus 4 V^\perp
\oplus \odot^2_0 V^\perp \oplus \ext{2} V^\perp$ and we will now show
how most of this accounts for the deformations of the carrollian structure
(including the Kalb--Ramond field to which the carrollian string
couples) of the flat Carroll spacetime.  A carrollian structure with
Kalb--Ramond field $B$ is given by a triple $(\kappa^\mu, h_{\mu\nu},
B_{\mu\nu})$ subject to the condition $h_{\mu\nu} \kappa^\nu = 0$.
For Carroll spacetime we have $\kappa^\mu = \delta^\mu_0$, $h_{\mu\nu}
= \delta_{ij} \delta^i_\mu \delta^j_\nu$ and $B_{\mu\nu} = 0$.  A
first-order deformation of this structure is then
\begin{equation}
  \begin{aligned}
    \kappa^\mu &= \delta^\mu_0 + \kappa^{(1)\mu}\\
    h_{\mu\nu} &= \delta_{ij} \delta^i_\mu \delta^j_\nu +  h^{(1)}_{\mu\nu}\\
    B_{\mu\nu} &= 0 + B^{(1)}_{\mu\nu},
  \end{aligned}
\end{equation}
where the condition $h_{\mu\nu}\kappa^\nu = 0$ becomes to first order
\begin{equation}
  h^{(1)}_{00} = 0 \qquad\text{and}\qquad h^{(1)}_{i0} = - \delta_{ij} \kappa^{(1)j}.
\end{equation}
In summary, the first-order deformations (in terms of representations
of the stabiliser of $\p$) are given as
\begin{equation}
  \begin{aligned}
    \kappa^{(1)0} &:\quad \CC\\
    \kappa^{(1)i} &:\quad V \cong \CC\p \oplus V^\perp\\
    h^{(1)}_{ij} &:\quad \odot^2 V \cong \CC\tr^\perp{} \oplus \CC \p^2 \oplus \odot^2_0 V^\perp \oplus  V^\perp\\
    B^{(1)}_{ij} &:\quad \ext{2} V \cong \ext{2} V^\perp \oplus  V^\perp\\
    B^{(1)}_{0i} &:\quad V \cong \CC\p \oplus V^\perp,
  \end{aligned}
\end{equation}
for a combined total of $5 \CC \oplus 4 V^\perp \oplus \odot^2_0
V^\perp \oplus \ext{2} V^\perp$.  Not all of these first-order
deformations are physical, some may be due to diffeomorphisms or the
abelian $1$-form gauge transformations of the Kalb--Ramond field.
Let $\zeta = \zeta^\mu \d_\mu$ be a vector field and $\theta = \theta_\mu
dx^\mu$ be a one-form, which we think of as generators of
infinitesimal diffeomorphisms and abelian gauge transformations,
respectively.  Then the fields $(\kappa, h, B)$ transform as follows:
\begin{equation}
  \delta \kappa = [\zeta,\kappa],\quad \delta h = \eL_\zeta h
  \quad\text{and}\quad \delta B = \eL_\zeta B + d\theta.
\end{equation}
Working these expressions out for the $(\kappa,h, B)$ of Carroll
spacetime, we arrive at
\begin{equation}
  \delta \kappa^0 = \delta \kappa^i = \delta h_{00} = \delta h_{0i} =
  0,\quad \delta h_{ij} = i p_{(i} \zeta_{j)},\quad \delta B_{0i} =-i
  \theta_0 p_i \quad\text{and}\quad \delta B_{ij} = 2 i p_{[i} \theta_{j]},
\end{equation}
where we have Fourier-transformed to momentum space for ease of
comparison with the cohomology calculations.  We observe that only the
perpendicular component $\theta_i^\perp$ of $\theta_j$ actually
contributes.  Hence the effective gauge parameters are
$\zeta_i, \theta_0, \theta_i^\perp$.  In terms of representations of the
stabiliser of $\p$, these account for $2 \CC \oplus 2 V^\perp$.
Discounting these from the first-order deformations of the carrollian
structure
$5 \CC \oplus 4 V^\perp \oplus \odot^2_0 V^\perp \oplus \ext{2}
V^\perp$, we arrive at the subspace of $\sH^2(0,\p)$ which can be
interpreted as deformations of the carrollian structure (augmented by
the Kalb--Ramond field):
$3\CC \oplus 2 V^\perp \oplus \odot^2_0 V^\perp \oplus \ext{2}
V^\perp$.  This leaves $\CC \oplus 2 V^\perp$ which we cannot seem to
account for geometrically.

We can see more explicitly how this correspondence arises.  The (not
obviously trivial) cocycles in $\sZ^2(0,\p)$ are given explicitly as
follows (with everything acting on $\vac$ with $p = (0,\p)$)
\begin{align*}
  \begin{aligned}
    &\phi^{(36)} c_0 c_1 (\pi_0)_{-1} \\
    &A_i^{(33)} c_0 c_1 x^i_{-1} \\
    &\phi^{(2)} \left( C_0 c_1 C_1 B_{-2} + 2 c_1 C_{-1} - \tfrac{i}2 p_i c_1 C_1 x^i_{-2}  + \tfrac12 c_1 C_1 x^\mu_{-1} (\pi_\mu)_{-1}  \right)\\
    &\phi^{(3)} \left( C_0 c_1 x^0_{-1} - C_0 C_1 (\pi_0)_{-1} - 2 c_1 C_1 x^0_{-2} - i p_i c_1 C_1 x^i_{-1}x^0_{-1} \right)\\
    &A_i^{(4)} \left( C_0 c_1 x^i_{-1} - 2 c_1 C_1 x^i_{-2} \right) \\
    &\phi^{(7)} \left(  C_0 c_1 (\pi_0)_{-1} - c_1 C_1 (\pi_0)_{-2} \right)\\
    &A_i^{(8)} \left( C_0 c_1 (\pi_i)_{-1} + \tau^2 C_o C_1 x^i_{-1} - i p_i c_1 C_{-1} - \tfrac12 (p^2 \delta_{ij} + p_i p_j) c_1 C_1 x^j_{-2}\right. \\
    & \qquad {} \left. - c_1 C_1 (\pi_i)_{-2} - \tfrac{i}2 p_j c_1 C_1 x^i_{-1} (\pi_j)_{-1} - \tfrac{i}2 p_j c_1 C_1 x^j_{-1} (\pi_i)_{-1}\right) \\
    &\phi^{(11)} c_1 c_{-1}\\
    & \phi^{(13)} \left( -C_1 c_{-1} - c_1 C_{-1} - i p_i c_1 C_1 x^i_{-2} \right)\\
    & \phi^{(21)} c_1 C_1 (\pi_0)_{-1}^2\\
    & A_i^{(22)} c_1 C_1 x^i_{-1} (\pi_0)_{-1}\\
    & A_i^{(25)} \left(  c_1 C_1 (\pi_i)_{-1} (\pi_0)_{-1} + \tau^2 c_1 C_1 x^i_{-1} x^0_{-1}\right)\\
    & T^{(26)}_{ij} c_1 C_1 x^i_{-1} x^j_{-1} \\
    & T^{(27)}_{[ij]} \left( c_1 C_1 x^i_{-1} (\pi_j)_{-1} + i c_1 C_1 p^{[i} x^{j]}_{-2}\right),
  \end{aligned}
\end{align*}
where we have introduced coefficient scalars $\phi$, vectors $A$ and
$2$-tensors $T$.

The transformations $\delta \Psi = d \Lambda$ corresponding to adding
coboundaries in $\sB^2(0,\p)$, where (everything acting on $\vac$)
\begin{multline}
  \Lambda = \lambda^{(1)} C_0 + \lambda^{(2)} c_0 + \lambda^{(3)} c_1 x^0_{-1} + \chi^{(4)}_i c_1 x^i_{-1} + \lambda^{(5)} C_1 x^0_{-1} + \chi^{(6)}_i C_1 x^i_{-1} + \lambda^{(7)} c_1 (\pi_0)_{-1}\\
  + \chi^{(8)}_i c_1 (\pi_i)_{-1} + \lambda^{(9)} C_1 (\pi_0)_{-1} + \chi^{(10)}_i C_1 (\pi_i)_{-1} + \lambda^{(11)} c_1 C_1 b_{-2} + \lambda^{(12)} c_1 C_1 B_{-2},
\end{multline}
are given by
\begin{align*}
  \begin{aligned}
    \delta \phi^{(36)} &= 0\\
    \delta A_i^{(33)} &= - i p_i \lambda^{(2)}\\
    \delta\phi^{(2)} &=  2 \lambda^{(11)}\\
    \delta \phi^{(3)} &= - \lambda^{(5)}\\
    \delta A_i^{(4)} &= -i p_i \lambda^{(1)} - \chi^{(6)}_i - \tau^2 \chi^{(8)}_i\\
    \delta\phi^{(7)} &= \lambda^{(3)} -\lambda^{(9)}\\
    \delta A_i^{(8)} &= - \chi_i^{(10)}\\
  \end{aligned}
  \qquad\qquad
  \begin{aligned}
    \delta\phi^{(11)} &= 2 \lambda^{(2)} - i p \cdot \chi^{(8)} + 2 \lambda^{(12)}\\
    \delta\phi^{(13)} &= 2 \lambda^{(1)} - i p \cdot \chi^{(10)} - 3 \lambda^{(11)} \\
    \delta\phi^{(21)} &= \tfrac12 \lambda^{(12)}\\
    \delta A_i^{(22)} &= i p_i \lambda^{(9)}\\
    \delta A_i^{(25)} &= 0\\
    \delta T^{(26)}_{ij} &= -\tfrac12 \tau^2 \delta_{ij} \lambda^{(12)} + i p_{(i} \chi_{j)}^{(6)}\\
    \delta T^{(27)}_{[ij]} &= i p_{[i} \chi_{j]}^{(10)}.
  \end{aligned}
\end{align*}

The Kalb--Ramond field is the easiest to deal with: $B_{\mu\nu}$
decomposes into $(B_{ij},B_{0i})$, while the linearised transformation
of $B_{\mu\nu}$ is
\begin{equation}
  \delta B_{\mu\nu} = \delta^j_\mu i p_{[j}\theta_{\nu]} \Rightarrow
  \delta B_{ij} = ip_{[i}\theta_{j]}~~\text{and}~~ \delta B_{0i} =
  -ip_i \theta_0,
\end{equation}
where we used that $p_0 = 0$, and where $\theta$ is a $1$-form (note
that diffeomorphisms are absent since the background value of the
$B$-field vanishes). There is only a single skewsymmetric tensor in
$\sZ^2(0,\p)$, namely $T^{(27)}_{[ij]}$, which transforms as $\delta
T^{(27)}_{[ij]} = i p_{[i} \chi_{j]}^{(10)}$, allowing us to identify
$\theta^\perp_i = \chi_i^{(10)}{}^\perp$. Furthermore, we identify
$B_{0i} = A^{(22)}_i$ so that $\theta_0 = -\lambda^{(9)}$.

Next up is the carrollian ruler $h_{\mu\nu}$: since, as discussed
above, $h^{(1)}_{00} = 0$, we get $h^{(1)}_{\mu\nu} =
(h^{(1)}_{ij},h^{(1)}_{0i} )$, where $h^{(1)}_{ij}$ is symmetric.  The
ruler transforms only under linearised diffeomorphisms generated by a
vector $\zeta^\mu$, leading to
\begin{equation}
  \delta h^{(1)}_{ij} = ip_{(i}\zeta_{j)},\qquad \delta h^{(1)}_{0i} = ip_i \zeta^0.
\end{equation}
The only candidate for $h^{(1)}_{ij}$ is $T^{(26)}_{ij}$, but its
trace does not transform as it should.  We can modify this by
combining $T^{(26)}_{ij}$ with $\phi^{(21)}$ as
\begin{equation}
    h_{ij}^{(1)} = T^{(26)}_{ij} + \tau^2\delta_{ij}\phi^{(21)},
\end{equation}
which transforms as $\delta h_{ij}^{(1)} = ip_{(i}\chi_{j)}^{(6)}$, allowing
us to identify the spatial components of the linearised diffeomorphism
with $\chi_i^{(6)}$.

This leaves the carrollian vector field $\kappa^\mu$.  We saw how
$\kappa^{(1)}_i = - h^{(1)}_{0i}$ and $\delta\kappa^{(1)} = 0$ for the
Carroll spacetime, so we may identify $\kappa^{(1)}_i$ with
$A_i^{(25)}$.  How about $\kappa^{(1)}{}^0$?  This should be a scalar
cocycle which does not receive any contributions from coboundaries.
There are precisely four such scalar cocycles: $\phi^{(36)}$, $\p \cdot
A^{(25)}$ (which has already been identified with the longitudinal
component $\p \cdot \kappa^{(1)}$) and the following two complicated
expressions:
\begin{equation}
  \begin{aligned}
    e_1 &:= 2 i \tau^2 \p \cdot A^{(33)} - \tfrac34 p^4 \phi^{(2)} +
    i p^2 \p \cdot A^{(4)} + \tfrac{i p^4}2 \p \cdot A^{(8)} - p^2
    \tau^2 \phi^{(11)} - \tfrac{p^4}2 \phi^{(13)}\\
    & \qquad {} + 5 p^2 \tau^2 \phi^{(21)} + T^{(26)}_{ij} p_i p_j\\
    e_2 &:= 2 i \tau^2 \p \cdot A^{(33)} - \tfrac34 p^4 \phi^{(2)} +
    i p^2 \p \cdot A^{(4)} + \tfrac{i p^4}2 \p \cdot A^{(8)} - p^2
    \tau^2 \phi^{(11)} - \tfrac{p^4}2 \phi^{(13)}\\
    & \qquad {} + 29 p^2 \tau^2 \phi^{(21)} + p^2 \tr T^{(26)}\\
  \end{aligned}
\end{equation}
which obey
\begin{equation}
  \tfrac1{p^2}(e_2 - e_1) = \left(\delta_{ij} - \frac{p_i p_j}{p^2}\right)
  T^{(26)}_{ij} + 24 \tau^2 \phi^{(21)}.
\end{equation}
It would be tempting to associate this with $\kappa^{(1)}{}^0$, but it
is actually the transverse trace of $h^{(1)}$ and we don't see why
there should be a priori any such relationship.  This leaves
$\phi^{(36)}$ (or one of $e_1, e_2$) as possible $\kappa^{(1)}{}^0$.
Alas, these do not seem to be of the type $c C \Phi$ for some field
$\Phi$ depending only on the matter fields.  Thus, taking the
cohomological calculations at face value, we are led to conclude that
there is no candidate cohomology class corresponding to
$\kappa^{(1)}{}^0$. This remains a puzzle thus far.

We may summarise this discussion as follows.  The fields which create
$2$-cocycle representatives for the remaining Carroll geometry
deformations are as follows:
\begin{equation}
  \begin{aligned}
    \kappa^{(1)}_i = - h^{(1)}_{0i} &: \qquad c C \left( \tau^2 \d X^i \d X^0 + \Pi^i \Pi_0 \right) V_p\\
    h^{(1)}_{ij} &: \qquad c C \left( \d X^i \d X^j + \tau^2 \delta^{ij} \Pi_0^2 \right) V_p\\
    B^{(1)}_{0i} &: \qquad c C \d X^i \Pi_0 V_p\\
    B^{(1)}_{ij} &: \qquad c C \left( \d X^{[i} \Pi^{j]} + \tfrac12 p^{[i} \d^2 X^{j]} \right) V_p,\\
  \end{aligned}
\end{equation}
where we have introduced the notation $V_p := \exp (i p_\mu X^\mu)$.
The gauge parameters correspond to the following fields:
\begin{equation}
  \theta_i^\perp = C \Pi_i^\perp V_p,\qquad \theta_0 = C \Pi_0 V_p
  \qquad\text{and}\qquad \zeta_i = C \d X^i V_p.
\end{equation}

These cohomology classes appear at other ghost numbers.  Every
BRST-invariant field at ghost number $2$ of the form $c C W$, for some
field $W$ not involving ghosts, appears three more times as
$\d c c C W$ and $\d C c C W$ at ghost number $3$ and
$\d c \d C c C W$ at ghost number $4$.

\section{Conclusions and outlook}
\label{sec:conclusions}

In this work, we have quantised the carrollian bosonic string
propagating on Carroll spacetime given by~\eqref{eq:flat-carroll} 
and worked out the full associated BRST cohomology using the framework
developed in~\cite{Figueroa-OFarrill:2024wgs}. This analysis
revealed that the cohomology is localised on zero energy, $p_0 = 0$,
suggesting that the spectrum is ``magnetic'' in the sense
of~\cite{deBoer:2021jej}.  In contrast with the standard bosonic
string, the spectrum is finite-dimensional, lacking the infinite
Kaluza--Klein tower of massive states.  These
cohomologies appear in Tables~\ref{tab:Habs0p} ($\p \neq 0)$ 
and~\ref{tab:Habspeqzero} ($\p = 0$). We interpreted them in terms of
UIRs of the $26$-dimensional Carroll group: when $\p = 0$, these
representations are finite-dimensional, while they are
infinite-dimensional when $\p \neq 0$.  We furthermore identified
BRST-invariant vertex operators (defining classes in $\sH^2(0,\p)$)
corresponding to (most of) the first-order deformations of the
carrollian structure (augmented by the Kalb--Ramond field) of the
Carroll spacetime.   These cohomology classes also appear at other
ghost numbers, corresponding to the phenomenon of ``pictures''
associated to the ghost zero modes.

Our results suggest a number of interesting avenues for further study,
and we list some of these below.
\begin{description} 
    \item[Supersymmetry] It would be interesting to compute the BRST
      cohomology of the quantum carrollian superstring.  By
      appropriately contracting two copies of the $\mathcal{N}{=}1$ Virasoro
      superalgebra (see, e.g.,~\cite{Mandal:2016wrw, Banerjee:2016nio,
        Lodato:2016alv, Bagchi:2022owq}), one obtains the $\mathcal{N}{=}(1,1)$
      BMS$_3$ superalgebra with operator product expansions augmenting
      the ones in equation~\eqref{eq:bms-as-opes} with
      \begin{equation}
        \label{eq:super-BMS-opes}
        \begin{aligned}
          T(z) G^\pm(w) &= \frac{\tfrac32 G^\pm(w)}{(z-w)^2} + \frac{\D G^\pm(w)}{z-w} + \reg\\
          G^+(z)G^-(z) &= \frac{\tfrac13 c_M \1}{(z-w)^3} + \frac{2M(w)}{z-w} + \reg\\
          G^\pm(z) G^\pm(w) &= \reg\\
          M(z) G^\pm(w) &= \reg,
        \end{aligned}
      \end{equation}
      where we have added two fermionic fields\footnote{There is
        another real form of the complexification of this superalgebra
        generated by $G^I$, $I=1,2$, where $G^\pm = G^1 \pm i G^2$.}
      $G^\pm$ of conformal weight $\tfrac32$.  The BRST complex
      exists provided that $c_M = 0$ and $c_L = 30$, resulting in the
      following mode algebra
    \begin{equation}
    \label{eq:super-BMS}
    \begin{aligned}\relax
        [L_n,L_m] &= (n-m)L_{m+n} + \tfrac{5}{2}n(n^2 -1)\delta_{m+n,0}\\
        [L_n,M_m] &= (n-m)M_{m+n}\\
        [L_m,G^\pm_r] &= \left(\tfrac m2 -r \right) G^\pm_{m+r}\\
        [G_r^+,G^-_s] &= M_{r+s},
    \end{aligned}
  \end{equation}
  where
  \begin{equation}
    G^\pm(z) = \sum_{r\in \ZZ + \tfrac12}G^\pm_r\, z^{-r-\tfrac32}.
  \end{equation}
  In addition to the fermionic ghost systems $(b,c)$ and $(B,C)$ with
  weights $(2,-1)$ introduced in the main text, we must now also
  introduce two bosonic ghost systems $(\beta^\pm,\gamma_\pm)$ with
  weights $(\tfrac32,-\tfrac12)$.  The BRST operator is the zero mode
  of the BRST   current given by the usual Lie algebraic form:
  \begin{equation}
    \label{eq:super-BRST}
    j_{\text{BRST}} = cT + CM + \gamma_+ G^+ + \gamma_- G^- + \tfrac12
    c \Tgh + \tfrac12 C \Mgh + \tfrac12 \gamma_+ \Gpgh + \tfrac12
    \gamma_- \Gmgh,
  \end{equation}
  where now
  \begin{equation}
    \begin{aligned}
      \Tgh &= -2 b\d c - \d b c - 2 B \d C - \d B C - \tfrac32 \beta^+
      \d \gamma_+ - \tfrac12 \d \beta^+ \gamma_+ - \tfrac32 \beta^-
      \d \gamma_- - \tfrac12 \d \beta^- \gamma_-\\
      \Mgh &= - 2 B\d c - \d B c\\
      \Gpgh &= \d \beta^+ c - B \gamma_- + \tfrac32 \beta^+ \d c\\
      \Gmgh &= \d \beta^- c - B \gamma_+ + \tfrac32 \beta^- \d c.
    \end{aligned}
  \end{equation}
  The cohomological analysis must be redone, of course.
  Quantisation of the $\mathcal{N}{=}(1,1)$ carrollian superstring in
  the flipped vacuum was recently considered
  in~\cite{Chen:2025gaz}. More generally, from a purely algebraic
  perspective, there exists an $\mathcal{N}{=}(1,1)$ extension of the
  BMS$_3$-like algebra $\g_\lambda$
  of~\cite{Figueroa-OFarrill:2024wgs}. Now $M$ has conformal weight
  $1-\lambda$ and $G^\pm$ have conformal weight $1-\tfrac12\lambda$.
  There is again a BRST complex of Lie algebraic form where now the
  critical central charge is $c_L = 6(4 + \lambda^2)$.
   \item[Interpolating string theories] The considerations at the end of
    Section~\ref{sec:mode-exps-and-constraints} suggest a whole family of 
    interpolating representations which describe a carrollian string
     propagating in a target manifold with a $p$-brane carrollian
     structure:
     \begin{equation}
         T(z) = \D X^\mu \Pi_\mu\qquad{\text{and}}\qquad M(z) =
         - \tfrac12\eta^{AB}\Pi_A \Pi_B - \tfrac12 \delta_{IJ}\D X^I \D X^J,
     \end{equation}
     where $A,B = 0,\dots,p$ and $I,J = 1,\dots,D-p$. In particular, the
     case $p = D$ corresponds to the ambitwistor string
     (cf.~\eqref{eq:ambi-twistor}), and $p=0$ corresponds to the case
     in the present paper. The $p$-brane carrollian geometry of the target 
     spacetime is obtained by increasing the dimensionality of the 
     ``longitudinal'' space of degenerate directions of the ruler $h$. 
     Concretely, in a $p$-brane carrollian geometry, we replace the carrollian 
     vector field $\kappa \in \mathscr{X}(M)$ by $\kappa_A \in \mathscr{X}
     (M)$, where $A = 0,\dots,p$. This vector field satisfies $h(\kappa_A,-) = 
     0$, and so $h$ has corank $p+1$. Furthermore, the object 
    \begin{equation}
    \label{eq:carr-metric}
        \bar\kappa := \eta^{AB}\kappa_A \kappa_B,
    \end{equation}
    where $\eta_{AB} = \text{diag}(-1,1,\dots,1)$ is the
    $(p+1)$-dimensional Minkowski metric, has signature
    $(\underbrace{-1,1,\dots,1}_{p+1\text{ entries}},0,\dots,0)$.  The
    dual to $\bar\kappa$ is denoted $\bar\xi$.\footnote{For more
      details about string ($p=1$) and more general $p$-brane
      carrollian geometries, we refer the reader
      to~\cite{Bergshoeff:2023rkk, Blair:2023noj, Bagchi:2023cfp,
        Gomis:2023eav, Bagchi:2024rje, Blair:2024aqz}.}  The $p$-brane
    carrollian generalisation of the carrollian string
    action~\eqref{eq:carroll-action} is then%
    \footnote{We remark that this class of string actions naturally
      appears when considering decoupling limits leading to matrix
      theories on instantonic branes, relevant to dS holography
      \cite{Blair:2025nno}.}
    \begin{equation}
        S = \int_\Sigma \,\dvol(\Sigma)\left[ \tfrac12 \tau^2 \left( \vv^\alpha \D_\alpha X^\mu \tilde P_\mu + h_{\mu\nu} \ee^\alpha\ee^\beta \D_\alpha X^\mu \D_\beta X^\nu \right) + \tfrac12 \bar\xi_{\mu\nu} \vv^\alpha \vv^\beta \D_\alpha X^\mu\D_\beta X^\nu \right]
    \end{equation}
    In addition to the ambitwistor string, appearing for $p=D$, the
    other extreme occurs when $p=-1$, in which case the ruler
    $h_{\mu\nu}$ has full rank, and as such can be taken to be (pseudo)-riemannian, with a whole range of interpolating
    string theories arising in between for $-1 < p < D$.   It would be
    very interesting to repeat the quantisation procedure for all
    these string theories and this currently work in progress.

    \item[Scattering amplitudes] It would be interesting to study amplitudes for
    the carrollian string; for example, one could start with the (unintegrated)
    vertex operators
    \begin{equation}
      h^{(1)}_{ij} c C \left( \d X^i \d X^j + \tau^2 \delta^{ij} \Pi_0^2 \right) \exp(i p_k X^k)
    \end{equation}
    corresponding to deformations of the carrollian ruler, and use
    standard arguments to extract expressions for scattering between
    what one might call ``carrollian gravitons''.  Moreover, for the
    interpolating string theories we mention above, the longitudinal
    sector is related to the ambitwistor string~\cite{Mason:2013sva},
    and it would be interesting to explore amplitudes using the
    formalism developed in that paper (see~\cite{Gomis:2023eav} for an
    example of this in the case of a galilean string). 

  \item[Spin matrix strings and other generalisations] Spin matrix
    strings arise from an 
    additional nonrelativistic worldsheet limit combined with a certain scaling 
    of the target spacetime geometry that turns it into what has been dubbed
    $\U(1)$-galilean geometry~\cite{Harmark:2014mpa, Harmark:2017rpg,
    Harmark:2018cdl, Harmark:2020vll,Bidussi:2023rfs}.
  Spin matrix strings are trivial when the background is Galilei spacetime, 
  while they reduce to a free theory on ``flat-fluxed 
  backgrounds'' that arise from the Penrose limit of $\zAdS^5{\times} \mathsf{S}^5$ 
  combined with the previously mentioned scaling. They are holographically dual 
  to near-BPS corners of $N{=}4$ super Yang--Mills theory known as spin matrix 
  theory, from which the strings derive their name. Spin matrix theory is 
  closely related to spin chains and integrability. This furnishes a 
  tantalisingly tractable version of the AdS/CFT correspondence, though a 
  crucial missing ingredient has been the quantisation of the bulk spin matrix
  string theory. The spin matrix theory limit described above leads to a worldsheet
  described by the two-dimensional galilean conformal algebra, which is isomorphic
  to the BMS$_3$ algebra. The methods developed in this work therefore form
  a natural starting point for the quantisation of the SMT string,
  something we plan to pursue in the near future. 
  Such novel worldsheet actions are also related to the fundamental string action
  obtained in multicritical near-BPS limits~\cite{Blair:2023noj,Gomis:2023eav}
  along with carrollian versions of these multiple limits~\cite{Blair:2025nno}.
  
    \item[Wess--Zumino--Witten models] These are some of the best
      understood conformal field theories and describe string
      propagation on Lie groups admitting a bi-invariant metric.
      Recently in \cite{Figueroa-OFarrill:2025nmo}, taking inspiration
      from the construction in \cite{Figueroa-OFarrill:2022pus} of Lie
      groups with a bi-invariant galilean structure, a
      non-relativistic string theory reminiscent of the Gomis--Ooguri
      string was obtained by null gauging a lorentzian WZW model.  In
      \cite{Figueroa-OFarrill:2022pus} it was also shown that Lie
      groups admitting bi-invariant carrollian structures arise as
      null codimension-one normal subgroups of certain Lorentzian Lie
      groups.  It would be interesting to implement this in a WZW
      model and in this way obtain novel strings with lorentzian
      worldsheet and carrollian target. En passant, we note that 
      the ``magnetic carrollian string'' 
      of~\cite{Bagchi:2023cfp,Bagchi:2024rje} and the string models
      proposed in~\cite{Harksen:2024bnh} are of this type, 
      though they do not arise as WZW models.

    \item[Strings on nontrivial backgrounds and beta functions] It
      would be interesting to consider carrollian strings in
      background fields as defined by the states that appear at ghost
      number $2$, as discussed in
      Section~\ref{sec:interpr-terms-carr-1}, which in particular
      involves the carrollian structure and a $B$-field. It remains an
      open question to interpret the remaining states in the
      cohomology not accounted for by the carrollian structure and the
      Kalb--Ramond field; one way to get a handle on this would be to
      compute the string beta functions. These should be related the
      equations of motion for the bosonic sector of an appropriate
      carrollian supergravity (see, e.g.,~\cite{Hartong:2015xda,
        Bergshoeff:2017btm, Hansen:2021fxi, Figueroa-OFarrill:2022mcy}
      for details about carrollian gravity). Moreover, it could be
      interesting to consider strings propagating on nontrivial
      carrollian backgrounds; for example, one could consider
      carrollian strings on backgrounds that include the carrollian
      version of anti-de Sitter space, which, as shown
      in~\cite{Figueroa-OFarrill:2021sxz} (see
      also~\cite{Have:2024dff,Borthwick:2024skd}) admits an
      interpretation as the blowup $\mathsf{Ti}$ of timelike infinity
      $i^+$ in asymptotically flat spacetimes.

    \item[BV structure in BRST cohomology] It is well-known that the
      BRST cohomology of (super)string theories admits the structure
      of a Batalin--Vilkovisky (BV) algebra
      \cite{Lian:1992mn,Getzler:1994yd,MR1314668} and, more generally,
      that this is the case for any topological conformal algebra
      \cite{Figueroa-OFarrill:1995qkv,MR1466615}.  The BMS BRST
      complex (regardless the matter realisation) satisfies the axioms
      of a topological conformal algebra and hence the BRST cohomology
      admits the structure of a BV algebra, with the Virasoro
      antighost zero mode $b_0$ playing the rôle of the BV operator.
      As shown in \cite{Figueroa-OFarrill:2024wgs}, this BV algebra is
      isomorphic to the chiral ring of a topologically twisted $N{=}2$
      superconformal field theory obtained by coupling the BMS string
      to topological gravity \cite{Getzler:1994yd}.  The BV structure
      is particularly rich in the case of non-critical strings and it
      might be interesting to explore non-critical BMS strings.  One
      obvious question is what is the BMS-analogue of the Liouville
      theory.  It is well known that any 2d generally covariant theory
      can be promoted to a conformally invariant theory by coupling to
      Liouville gravity.  So a natural question is what plays the rôle
      of Liouville gravity in promoting a 2d generally covariant
      theory to a BMS field theory.

    \item[Carroll limit of the bosonic string BRST cohomology] The
     carrollian string action~\eqref{eq:carroll-action} may be obtained
     by taking a Carroll limit of the phase space action for the bosonic
     string. Together with the fact that the spectrum of the
     carrollian string is finite-dimensional, suggests that the BRST
     cohomology of the carrollian string summarised in Tables~\ref{tab:Habs0p}
     and~\ref{tab:Habspeqzero} can be obtained by taking an
     appropriate Carroll limit of the cohomology of the bosonic
     string.  Based on the similarity with the ambitwistor string, we
     expect that this limit mixes massless states of open and closed
     strings in a non-trivial way akin to what happens in~\cite{Jusinskas:2021bdj}.
    
\end{description}

\section*{Acknowledgments}
\label{sec:acknowledgments}

We are grateful to Stefan Fredenhagen, Troels Harmark, Jelle Hartong, Girish
Vishwa, and Ziqi Yan for useful discussions.  JMF is grateful to the NBI
for hospitality and support during the final stages of this work. 
The work of EH and NO is supported by Villum Foundation Experiment project
00050317,
``Exploring the wonderland of Carrollian physics''. 
The Center of Gravity is a Center of Excellence funded by the Danish National
Research
Foundation under grant No.~184.

\appendix

\section{The BRST differential on the relative subcomplex}
\label{sec:acti-brst-diff}

In this appendix we collect formulae used in the computation of the
relative BRST cohomology.  We will calculate the action of the BRST
differential on a field basis for $\fM^\bullet(0)$.  Together with
Lemmas~\ref{lem:dC0} and \ref{lem:df} and using the state-field
correspondence~\eqref{eq:state-field-corr}, this suffices to compute
the action of $d$ on $\Crel^\bullet(p)$.  The computations on fields
were performed with \texttt{OPEdefs}
\cite{Thielemans:1991uw,Thielemans:1992mu, Thielemans:1994er}.

Since the BRST differential is $\SO(25)$-equivariant, we may decompose
the calculations based on the $\SO(25)$ representation the basis
elements belong to: scalars, vectors and $2$-tensors.

There is only one basis field at ghost number $0$: the identity, which
is a scalar and a cocycle.

At ghost number $1$ we have that the action on
$\SO(25)$-scalars is given by
\begin{align}
  \label{eq:d-on-M-scalars-gh-no-1}
  \begin{aligned}
    d(c \d X^0) &= - c \d(C \Pi_0)\\
    d(c \Pi_0) &= 0\\
    d(C \Pi_0) &= c \d (C \Pi_0)\\
    d(c C B) &= c C M - \tfrac32 \d^2 c c\\
    d(C \d X^0) &= c \d(C \d X^0) + \d C C \Pi_0\\
    d(c C b) &= c C T - 2 c B \d C C + \tfrac32 c \d^2 C + \tfrac32 \d^2 c C,
  \end{aligned}
\end{align}
whereas on vectors it is given by
\begin{align}
  \label{eq:d-on-M-vectors-gh-no-1}
  \begin{aligned}
    d(c \d X^i) &= 0\\
    d(c \Pi_i)&=  \tau^2 c \d( C \d X^i)\\
    d(C \d X^i) &= c \d (C \d X^i)\\
    d(C \Pi_i) &= c \d (C \Pi_i) - \tau^2 \d C C \d X^i.
  \end{aligned}
\end{align}
There are no tensors at ghost number $1$.

At ghost number $2$, the action on scalars\footnote{At this and higher
  ghost number there are additional scalars obtained by taking traces
  of tensors.  Since $d$ commutes with taking trace, the action of $d$
  on them can be read from that on tensors by taking the trace and so we
  will not list them here.} is given by
\begin{align}
  \label{eq:d-on-M-scalars-gh-no-2}
  \begin{aligned}
    d(c C \d \Pi_0) &= \d^2c c C \Pi_0\\
    d(c C \Pi_0 \Pi_0) &= 0\\
    d(c C \d^2 X^0) &= \d^2 c c C \d X^0 - c \d^2 C C \Pi_0 - 2 c \d C C \d \Pi_0\\
    d(c C \d X^0 \Pi_0) &=\tfrac16 \d^3 c c C- c \d C C (\Pi_0)^2\\
    d(c C \d X^0 \d X^0) &= -\tfrac16 c \d^3 C C - 2 c \d C C \d X^0 \Pi_0\\
    d(c\d^2 c) &= 0\\
    d(c\d^2 C) &= c\d(\d^2 c C)\\
    d(\d^2c C) &= c\d(\d^2 c C)\\
    d(C\d^2C) &= -c \d^2 C \d C - c \d^3 C C + \d^2 c  \d C C;
  \end{aligned}
\end{align}
whereas that on vectors is given by
\begin{equation}
  \label{eq:d-on-M-vectors-gh-no-2}
  \begin{aligned}
    d(c C \d^2 X^i) &= \d^2 c c C \d X^i\\
    d(c C \d X^i \Pi_0) &= 0\\
    d(c C \d X^0 \d X^i) &= -c \d C C \d X^i \Pi_0\\
    d(c C \d \Pi_i) &= \d^2 c c C \Pi_i - \tau^2 c C \left( \d^2 C \d X^i + 2 \d C \d^2 X^i\right)\\
    d(c C \d X^0 \Pi_i) &= -c \d C C \Pi_i \Pi_0 + \tau^2 c \d C C \d X^i \d X^0\\
    d(c C \Pi_0 \Pi_i) &= \tau^2 c \d C C \d X^i \Pi_0,
  \end{aligned}
\end{equation}
and that on tensors is given by
\begin{equation}
  \label{eq:d-on-M-tensors-gh-no-2}
  \begin{aligned}
    d(c C \d X^i \Pi_j) &= \tfrac16 \delta^i_j \d^3c c C + \tau^2 c \d C C \d X^i \d X^j\\
    d(c C \d X^i \d X^j) &= 0\\
    d(c C \Pi_i \Pi_j) &= \tfrac16 \tau^2 c \d^3 C C \delta_{ij} + \tau^2 \left( c \d C C \Pi_i \d X^j + c \d C C \d X^i \Pi_j \right).
  \end{aligned}
\end{equation}

At ghost number $3$, the action on scalars is given by
\begin{equation}
  \label{eq:d-on-M-scalars-gh-no-3}
  \begin{aligned}
    d(c C \d^2 c \Pi_0) &= 0\\
    d(c C \d^2 C \d X^0) &= \d^2 c c \d C C \d X^0 - c \d^2 C \d C C \Pi_0\\
    d(c C \d^2 c \d X^0) &= \d^2 c c \d C C \Pi_0\\
    d(c C \d^2 C \Pi_0) &= \d^2 c c \d C C \Pi_0\\
    d(c C \d^3c) &= 0\\
    d(c C \d^3C) &= 2\d^3 c c \d C C;
  \end{aligned}
\end{equation}
whereas that on vectors is given by
\begin{equation}
  \label{eq:d-on-M-vectors-gh-no-3}
  \begin{aligned}
    d(c C \d^2 c \d X^i) &= 0\\
    d(c C \d^2 c \Pi_i) &= - \tau^2 \d^2 c c \d C C \d X^i\\
    d(c C \d^2 C \Pi_i) &= \d^2 c c \d C C \Pi_i + \tau^2 c \d^2 C \d C C \d X^i\\
    d(c C \d^2 C \d X^i) &= \d^2 c c \d C C \d X^i.
  \end{aligned}
\end{equation}
At ghost number $4$ we have only one scalar and the BRST differential
is identically zero:
\begin{align}
  \label{eq:d-on-M-scalars-gh-no-4}
  \begin{aligned}
    d(c C \d^2 c \d^2 C) &= 0.
  \end{aligned}
\end{align}

\providecommand{\href}[2]{#2}\begingroup\raggedright\endgroup
\end{document}